\documentclass[11pt,letterpaper]{article}
\usepackage[lmargin=1.0in,rmargin=1.0in,bottom=1.0in,top=1.0in,twoside=False]{geometry}

\usepackage[utf8]{inputenc}
\usepackage[T1]{fontenc}
\usepackage{lmodern}

\usepackage{fullpage,amssymb,amsmath}
\usepackage{graphicx}
\usepackage{enumerate}
\usepackage{paralist}
\usepackage{tikz}
\usetikzlibrary{shapes}

\usepackage{xcolor}
\usepackage{mathtools}
\usepackage{microtype}
\usepackage{amsfonts}
\usepackage{comment}
\usepackage[english]{babel}
\usepackage{mathrsfs}
\usepackage[ruled,vlined]{algorithm2e}
\usepackage{blindtext}
\usepackage{thmtools}
\usepackage{thm-restate}
\usepackage{array}

\usepackage{trimspaces}
\usepackage{nccfoots}
\usepackage{setspace}
\usepackage{inconsolata}
\usepackage{libertine}
\usepackage[absolute]{textpos}

\usepackage{enumitem}
\usepackage[backgroundcolor = blue!50]{todonotes}
\usepackage{longtable}

\definecolor{blue}{rgb}{0.1,0.2,0.5}
\definecolor{brown}{rgb}{0.6,0.6,0.2}
\usepackage[ocgcolorlinks, linkcolor={blue}, citecolor={brown}]{hyperref}
\usepackage[amsmath,amsthm,thmmarks,hyperref]{ntheorem}
\usepackage{enumerate}
\usepackage{latexsym}
\usepackage{thm-restate}

\usepackage[nameinlink]{cleveref}

\crefformat{page}{#2page~#1#3}%
\Crefformat{page}{#2Page~#1#3}%
\crefformat{equation}{#2(#1)#3}%
\Crefformat{equation}{#2(#1)#3}%
\crefformat{figure}{#2Figure~#1#3}%
\Crefformat{figure}{#2Figure~#1#3}%
\crefformat{section}{#2Section~#1#3}
\Crefformat{section}{#2Section~#1#3}
\crefformat{chapter}{#2Chapter~#1#3}
\Crefformat{chapter}{#2Chapter~#1#3}
\crefformat{chapter*}{#2Chapter~#1#3}
\Crefformat{chapter*}{#2Chapter~#1#3}
\crefformat{part}{#2Part~#1#3}
\Crefformat{part}{#2Part~#1#3}
\crefformat{enumi}{#2(#1)#3}
\Crefformat{enumi}{#2(#1)#3}

\usepackage[mathlines]{lineno}

\newcommand*\patchAmsMathEnvironmentForLineno[1]{%
  \expandafter\let\csname old#1\expandafter\endcsname\csname #1\endcsname
  \expandafter\let\csname oldend#1\expandafter\endcsname\csname end#1\endcsname
  \renewenvironment{#1}%
     {\linenomath\csname old#1\endcsname}%
     {\csname oldend#1\endcsname\endlinenomath}}%
\newcommand*\patchBothAmsMathEnvironmentsForLineno[1]{%
  \patchAmsMathEnvironmentForLineno{#1}%
  \patchAmsMathEnvironmentForLineno{#1*}}%
\AtBeginDocument{%
  \patchBothAmsMathEnvironmentsForLineno{equation}%
  \patchBothAmsMathEnvironmentsForLineno{align}%
  \patchBothAmsMathEnvironmentsForLineno{flalign}%
  \patchBothAmsMathEnvironmentsForLineno{alignat}%
  \patchBothAmsMathEnvironmentsForLineno{gather}%
  \patchAmsMathEnvironmentForLineno{split}
  \patchBothAmsMathEnvironmentsForLineno{multline}}

\theoremnumbering{arabic}
\theoremstyle{plain}
\theoremsymbol{}
\theorembodyfont{\itshape}
\theoremheaderfont{\normalfont\bfseries}
\theoremseparator{.}

\newtheorem{theorem}{Theorem}
\crefformat{theorem}{#2Theorem~#1#3}
\Crefformat{theorem}{#2Theorem~#1#3}

\crefformat{definition}{#2Definition~#1#3}
\Crefformat{definition}{#2Definition~#1#3}

\newcommand{\newtheoremwithcrefformat}[2]{%
  \newtheorem{#1}[theorem]{#2}%
  \crefformat{#1}{##2\MakeUppercase#1~##1##3}%
  \Crefformat{#1}{##2\MakeUppercase#1~##1##3}%
}
\newcommand{\newseptheoremwithcrefformat}[2]{%
  \newtheorem{#1}{#2}%
  \crefformat{#1}{##2\MakeUppercase#1~##1##3}%
  \Crefformat{#1}{##2\MakeUppercase#1~##1##3}%
}
\newcommand{\newclaimwithcrefformat}[2]{%
  \newtheorem{#1}{#2}[theorem]%
  \crefformat{#1}{##2\MakeUppercase#1~##1##3}%
  \Crefformat{#1}{##2\MakeUppercase#1~##1##3}%
}

\newtheoremwithcrefformat{lemma}{Lemma}

\newtheoremwithcrefformat{proposition}{Proposition}
\newtheoremwithcrefformat{observation}{Observation}
\newseptheoremwithcrefformat{conjecture}{Conjecture}
\newtheoremwithcrefformat{corollary}{Corollary}
\newclaimwithcrefformat{claim}{Claim}

\crefformat{cthmin}{#2Theorem~#1#3}
\Crefformat{cthmin}{#2Theorem~#1#3}
\newtheorem{cthmin}{Theorem}%
\newenvironment{cthm}[1]
  {\renewcommand\thecthmin{#1}\cthmin}
  {\endcthmin}

\crefformat{ccorin}{#2Corollary~#1#3}
\Crefformat{ccorin}{#2Corollary~#1#3}
\newtheorem{ccorin}{Corollary}%
\newenvironment{ccor}[1]
  {\renewcommand\theccorin{#1}\ccorin}
  {\endccorin}

\theoremstyle{definition}
\newtheorem{definition}{Definition}
\theorembodyfont{\upshape}

\theoremstyle{nonumberplain}
\theoremheaderfont{\scshape}
\theorembodyfont{\normalfont}
\theoremsymbol{\ensuremath{\square}}

\theoremsymbol{\ensuremath{\lrcorner}}

\newcommand{\qedhere}{\hfill$\square$}

\newcommand{\cA}{\mathcal{A}}
\newcommand{\cB}{\mathcal{B}}
\newcommand{\cC}{\mathcal{C}}
\newcommand{\cD}{\mathcal{D}}
\newcommand{\cE}{\mathcal{E}}
\newcommand{\cF}{\mathcal{F}}

\newcommand{\cH}{\mathcal{H}}

\newcommand{\cO}{\mathcal{O}}
\newcommand{\cP}{\mathcal{P}}
\newcommand{\cQ}{\mathcal{Q}}
\newcommand{\cR}{\mathcal{R}}
\newcommand{\cS}{\mathcal{S}}
\newcommand{\cT}{\mathcal{T}}
\newcommand{\cU}{\mathcal{U}}
\newcommand{\cV}{\mathcal{V}}
\newcommand{\cW}{\mathcal{W}}
\newcommand{\cX}{\mathcal{X}}
\newcommand{\cY}{\mathcal{Y}}
\newcommand{\cZ}{\mathcal{Z}}

\newcommand{\ovalp}{\alpha'}
\newcommand{\ovbet}{\beta'}

\newcommand{\nobeg}{^\dashv}
\newcommand{\noend}{^\vdash}

\newcommand{\rneq}{\mathrm{NEQ}}
\newcommand{\ork}[1]{\mathrm{OR}_{#1}}
\newcommand{\nand}[1]{\mathrm{NAND}_{#1}}
\DeclareMathOperator{\dist}{dist}

\newcommand{\N}{\mathbb{N}}
\newcommand{\Ob}{\mathbb{O}}
\newcommand{\R}{\mathbb{R}}
\renewcommand{\phi}{\varphi}

\renewcommand{\epsilon}{\varepsilon}
\newcommand{\Oh}{\mathcal{O}}
\newcommand{\DD}{\mathbb{D}}
\newcommand{\concomp}[1]{\operatorname{comp}({#1})}

\renewcommand{\leq}{\leqslant}
\renewcommand{\geq}{\geqslant}

\renewcommand{\ge}{\geqslant}
\renewcommand{\setminus}{-}

\newcommand{\homo}[1]{\textsc{Hom}(\ensuremath{#1})\xspace}
\newcommand{\lhomo}[1]{\textsc{LHom}(\ensuremath{#1})\xspace}

\newcommand{\coloring}[1]{\ensuremath{#1}-\textsc{Coloring}\xspace}
\newcommand{\yourfavouritecolor}{orange!70}

\newcommand{\cw}[1]{{\operatorname{cw}(#1)}}
\newcommand{\tw}[1]{{\operatorname{tw}(#1)}}
\newcommand{\pw}[1]{{\operatorname{pw}(#1)}}
\newenvironment{claimproof}{\noindent {\emph{Proof of Claim.}}}{\hfill$\blacksquare$\smallskip}

\declaretheorem[sibling=theorem]{lemma}

\begin{document}
\title{Full complexity classification of the list homomorphism\\ problem for bounded-treewidth graphs
\thanks{The first author was supported by the ERC Starting Grant CUTACOMBS (grant agreement No.~714704). The second and the third author were supported by the Polish National Science Centre grant no. 2018/31/D/ST6/00062.}}

\author{Karolina Okrasa\thanks{Warsaw University of Technology, Faculty of Mathematics and Information Science and University of Warsaw, Institute of Informatics,  \texttt{k.okrasa@mini.pw.edu.pl}} \and Marta Piecyk\thanks{Warsaw University of Technology, Faculty of Mathematics and Information Science, \texttt{m.piecyk@mini.pw.edu.pl}} \and Paweł Rzążewski\thanks{Warsaw University of Technology, Faculty of Mathematics and Information Science and University of Warsaw, Institute of Informatics, \texttt{p.rzazewski@mini.pw.edu.pl}}
}

\begin{titlepage}
\def\thepage{}
\thispagestyle{empty}
\maketitle

\begin{abstract}
A homomorphism from a graph $G$ to a graph $H$ is an edge-preserving mapping from $V(G)$ to $V(H)$.
Let $H$ be a fixed graph with possible loops. In the list homomorphism problem, denoted by \textsc{LHom}($H$), we are given 
a graph $G$, whose every vertex $v$ is assigned with a list $L(v)$ of vertices of $H$.
We ask whether there exists a homomorphism $h$ from $G$ to $H$, which respects lists $L$, i.e., for every $v \in V(G)$ it 
holds that $h(v) \in L(v)$.

The complexity dichotomy for  \textsc{LHom}($H$) was proven by Feder, Hell, and Huang [JGT 2003]. The authors showed that 
the problem is polynomial-time solvable if $H$ belongs to the class called \emph{bi-arc graphs},
and for all other graphs $H$ it is NP-complete.

We are interested in the complexity of the  \textsc{LHom}($H$)  problem, parameterized by the treewidth of the input graph.
This problem was investigated by Egri, Marx, and Rz\k{a}\.zewski [STACS 2018], who obtained tight complexity bounds for the special case of reflexive graphs $H$, i.e., if every vertex has a loop.

In this paper we extend and generalize their results for \emph{all} relevant graphs $H$, i.e., those, for which the \lhomo{H} problem is NP-hard.
For every such $H$ we find a constant $k = k(H)$, such that the \textsc{LHom}($H$) problem on instances with $n$ vertices and treewidth $t$
\begin{compactitem}
\item can be solved in time $k^{t} \cdot n^{\mathcal{O}(1)}$, provided that the input graph is given along with a tree decomposition of width $t$,
\item cannot be solved in time $(k-\varepsilon)^{t} \cdot n^{\mathcal{O}(1)}$, for any $\varepsilon >0$, unless the SETH fails.
\end{compactitem}
For some graphs $H$ the value of $k(H)$ is much smaller than the trivial upper bound, i.e., $|V(H)|$.

Obtaining matching upper and lower bounds shows that the set of algorithmic tools we have discovered cannot be extended in order to obtain faster algorithms for \textsc{LHom}($H$) in bounded-treewidth graphs. Furthermore, neither the algorithm, nor the proof of the lower bound, is very specific to treewidth. We believe that they can be used for other variants of the \textsc{LHom}($H$) problem, e.g. with different parameterizations.
\end{abstract}

\end{titlepage}

\section{Introduction}
A popular line of research in studying computationally hard problems is to consider restricted instances, in order to understand the boundary between easy and hard cases.
For example, most of natural problems can be efficiently solved on trees, using a bottom-up dynamic programming.
This observation led to the definition of \emph{treewidth}, which, informally speaking, measures how much a given graph resembles a tree. The notion of treewidth appears to be very natural and it was independently discovered by several authors in different contexts~\cite{DBLP:journals/dam/ArnborgP89,Bodlaender:2008:COG:2479371.2479375,10.1007/978-3-642-22993-0_47,cygan2015parameterized}.

For many problems, polynomial-time algorithms for graphs with bounded treewidth can be obtained by adapting the dynamic programming algorithms for trees. Most of these straightforward algorithms follow the same pattern, which was captured by the seminal meta-theorem by Courcelle~\cite{DBLP:journals/iandc/Courcelle90}: he proved that each problem expressible in \emph{monadic second order logic} (MSO$_2$) can be solved in time $f(\tw{G}) \cdot n^{\Oh(1)}$ on graphs $G$ with $n$ vertices and treewidth $\tw{G}$, where $f$ is some function, depending on the MSO$_2$ formula describing the particular problem.
As a consequence of this meta-theorem, in order to show that some problem $\Pi$ is \emph{fixed-parameter tractable} (FPT), parameterized by the treewidth, it is sufficient to show that $\Pi$ can be described in a certain way.

The main problem with using the meta-theorem as a black-box is that the function $f$ it produces is huge (non-elementary). We also know that this bound cannot be improved, if we want to keep the full generality of the statement~\cite{DBLP:journals/apal/FrickG04}.
Because of this, as the area of the so-called \emph{fine-grained complexity} gained popularity~\cite{DBLP:journals/dagstuhl-reports/KratschLMR14}, researchers turned back to studying particular problems, and asking about the best possible dependence on treewidth, i.e., the function $f$.
This led to many exciting algorithmic results which significantly improved the naive dynamic programming approach~\cite{DBLP:conf/esa/RooijBR09,DBLP:journals/iandc/BodlaenderCKN15,DBLP:journals/corr/KociumakaP17}.

There is a little caveat that applies to most of algorithms mentioned above: Usually we assume that the input graph is given along with its tree decomposition, and the running time is expressed in terms of the width of this decomposition.
This might be a serious drawback, since finding an optimal tree decomposition is NP-hard~\cite{10.1137/0608024}.
However, finding a tree decomposition of given width (if one exists) can be done in FPT time~\cite{DBLP:journals/siamcomp/Bodlaender96,DBLP:journals/siamcomp/BodlaenderDDFLP16}, which is often sufficient. Since we are interested in the complexity of certain problems, parameterized by the treewidth, we will not discuss the time needed to find a decomposition. Thus we will always assume that the input graph is given along with its tree decomposition.

In parallel to improving the algorithms, many lower bounds were also developed~\cite{DBLP:conf/soda/LokshtanovMS11a,10.1007/978-3-642-22993-0_47,cygan2015parameterized}.
Let us point out that the main assumption from the classical complexity theory, i.e., P $\neq$ NP, is too weak to provide any meaningful lower bounds in our setting. The most commonly used assumptions in the fine-grained complexity world, are the \emph{Exponential-Time Hypothesis} (ETH) and the \emph{Strong Exponential-Time Hypothesis} (SETH), both introduced by Impagliazzo and Paturi~\cite{DBLP:journals/jcss/ImpagliazzoP01,DBLP:journals/jcss/ImpagliazzoPZ01}. Informally speaking, the ETH implies that \textsc{3-Sat} with $n$ variables cannot be solved in subexponential time, i.e., in time $2^{o(n)}$, and the SETH implies that \textsc{CNF-Sat} with $n$ variables and $m$ clauses cannot be solved in time $(2-\varepsilon)^{n} \cdot m^{\Oh(1)}$, for any $\epsilon >0$.

For example the straightforward dynamic programming algorithm for $k \coloring$ works in time $k^{\tw{G}} \cdot n^{\Oh(1)}$. As one of the first SETH-based lower bounds for problems parameterized by the treewidth, Lokshtanov, Marx, and Saurabh~\cite{DBLP:conf/soda/LokshtanovMS11a} proved that this bound is tight.

\begin{theorem}[Lokshtanov, Marx, Saurabh~\cite{DBLP:conf/soda/LokshtanovMS11a}]\label{thm:LMS}
For any $k \geq 3$, there is no algorithm solving $k \coloring$ on a graph with $n$ vertices and treewidth $t$ in time $(k-\epsilon)^t \cdot n^{\Oh(1)}$, unless the SETH fails.
\end{theorem}

\paragraph{Graph homomorphisms.}
In this paper we are interested in extending  \cref{thm:LMS} for one of possible generalization of the $k \coloring$ problem.
For graph $G$ and $H$ (both with possible loops on vertices), a \emph{homomorphism} from $G$ to $H$ is a mapping $h \colon V(G) \to V(H)$, which preserves edges, i.e., for every edge $xy$ of $G$ it holds that $h(x)h(y) \in E(H)$.
The graph $H$ is called a \emph{target}.
If $h$ is a homomorphism from $G$ to $H$, we denote it by writing $h : G \to H$.
We also write $G \to H$ to indicate that some homomorphism from $G$ to $H$ exists.

By \homo{H} we denote the computational problem of deciding whether an instance graph $G$ admits a homomorphism to $H$ (usually we consider $H$ a fixed graph, but we might also treat it as a part of the input).
Observe that if $H = K_k$, then \homo{H} is equivalent to the $k \coloring$ problem.
Because of that, homomorphisms to $H$ are often called \emph{$H$-colorings}.
We will also refer to vertices of $H$ as \emph{colors}.

Let us briefly survey the some results concerning the complexity of variants of the \homo{H} problem.
For more information, we refer the reader to the comprehensive monograph by Hell and Ne\v{s}et\v{r}il~\cite{hell2004graphs}.
The complexity dichotomy for \homo{H} was shown by Hell and Ne\v{s}et\v{r}il~\cite{DBLP:journals/jct/HellN90}: the problem is polynomial-time-solvable if $H$ contains a vertex with a loop or is bipartite, and NP-complete for all other graphs $H$.
Since then, many interesting results concerning the complexity of graph homomorphisms have appeared~\cite{DBLP:journals/mst/FominHK07,DBLP:journals/mst/Wahlstrom11,DBLP:journals/ipl/Rzazewski14,DBLP:journals/jacm/CyganFGKMPS17,DBLP:journals/mst/EgriKLT12,DBLP:journals/algorithmica/ChitnisEM17}.
The fine-grained complexity of the \homo{H} problem, parameterized by the treewidth of the input graph, was very recently studied by Okrasa and Rz\k{a}\.zewski~\cite{OkrasaSODA}.
They were able to find tight SETH-bounds, conditioned on two conjectures from algebraic graph theory from early 2000s.
As these conjectures remain wide open, we know no graph, for which the bounds from~\cite{OkrasaSODA} do not apply.

A natural and interesting extension of the \homo{H} problem is its \emph{list} version.
In the \emph{list homomorphism problem}, denoted by \lhomo{H}, the input consists of a graph $G$ and a \emph{$H$-lists} $L$, which means that $L$ is a function which assigns to each vertex of $G$ a subset of vertices of $H$. We ask whether there is a homomorphism $h$ from $G$ to $H$, which respects lists $L$, i.e., for each $x \in V(G)$ it holds that $h(x) \in L(x)$.
If $h$ is such a list homomorphism, we denote it by $h: (G,L) \to H$.
We also write $(G,L) \to H$ to indicate that some homomorphism $h: (G,L) \to H$ exists.

The complexity of the \lhomo{H} problem was shown in three steps. First, Feder and Hell~\cite{FEDER1998236} provided a classification for the case that $H$ is \emph{reflexive}, i.e., every vertex has a loop. They proved that if $H$ is an interval graph, then the problem is polynomial-time solvable, and otherwise it is NP-complete.
The next step was showing the complexity dichotomy for \emph{irreflexive graphs} (i.e., with no loops).
Feder, Hell, and Huang~\cite{DBLP:journals/combinatorica/FederHH99} proved that if $H$ is bipartite and its complement is a circular-arc graph, then the problem is polynomial-time solvable, and otherwise it is NP-complete. Interestingly, bipartite graphs whose complement is circular-arc were studied independently by Trotter and Moore~\cite{TROTTER1976361} in the context of some poset problems.
Finally, Feder, Hell, and Huang~\cite{DBLP:journals/jgt/FederHH03} provided the full classification for general graph $H$: the polynomial cases appear to be \emph{bi-arc graphs}, which are also defined in terms of some geometric representation. Let us now skip the exact definition of bi-arc graphs, and we will get back to it in~\cref{sec:associated}.

Let us point out that in all three papers mentioned above, the polynomial-time algorithms for \lhomo{H} exploited the geometric representation of $H$. On the other hand, all hardness proofs followed the same pattern.
First, for each ``easy'' class $\cC$ (i.e., interval graphs, bipartite co-circular-arc graphs, and bi-arc graphs), the authors provided an alternative characterization in terms of forbidden subgraphs. In other words, they defined a (non-necessarily finite) family $\cF$ of graphs, such that $H \in \cC$ if and only if $H$ does not contain any $F \in \cF$ as an induced subgraph. Then, for each $F \in \cF$, the authors showed that \lhomo{F} is NP-complete.
Note that this is sufficient, as every ``hard'' graph $H$ contains some $F \in \cF$, and every instance of \lhomo{F} is also an instance of \lhomo{H}, where no vertex from $V(H) \setminus V(F)$ appears in any list.

If the input graph $G$  is given with a tree decomposition of width $\tw{G}$, then the straightforward dynamic programming solves the \lhomo{H} problem in time  $|V(H)|^{\tw{G}} \cdot |V(G)|^{\Oh(1)}$. 
The study of the fine-grained complexity of the \lhomo{H} problem, parameterized by the treewidth of the input graph, was initiated by Egri, Marx, and Rz\k{a}\.zewski~\cite{DBLP:conf/stacs/EgriMR18}.
They were able to provide the full complexity classification for the case of reflexive graphs $H$, i.e., corresponding to the first step of the above-mentioned complexity dichotomy.

The authors defined a new and simple graph invariant, denoted by $i^*$, which is based on \emph{incomparable sets} and the existence of a certain decompositions in $H$, and proved the following tight bounds.

\begin{theorem}[Egri, Marx, Rz\k{a}\.zewski~\cite{DBLP:conf/stacs/EgriMR18}]\label{thm:EMRreflexive}
Let $H$ be a fixed connected reflexive non-interval graph, and let $k=i^*(H)$.
The \lhomo{H} problem on instances $(G,L)$ with $n$ vertices,
\begin{compactenum}[(a)]
\item can be solved in time $k^{\tw{G}} \cdot n^{\Oh(1)}$, provided that an optimal tree decomposition of $G$ is given,
\item cannot be solved in time $(k-\epsilon)^{\tw{G}} \cdot n^{\Oh(1)}$ for any $\epsilon > 0$, unless the SETH fails.
\end{compactenum}
\end{theorem}

In this paper we continue this line of research and provide the full complexity classification for all graphs $H$.
Our results heavily extend the framework of Egri, Marx, Rz\k{a}\.zewski~\cite{DBLP:conf/stacs/EgriMR18} and generalize~\cref{thm:EMRreflexive}.

\subsection{Our results.} 
In this paper we provide a full complexity classification of \lhomo{H}, parameterized by the treewidth of an instance graph.
Our results heavily extend the ones of Egri, Marx, Rz\k{a}\.zewski~\cite{DBLP:conf/stacs/EgriMR18} and generalize~\cref{thm:EMRreflexive} to all relevant graphs $H$.
Let us point out that instead of designing ad-hoc algorithms and reductions that are fine-tailored for a particular problem,
we rather build a general framework that allows us to provide tight bounds for a natural and important family of problems.

We prove the complexity classification for \lhomo{H} in two steps. 
First, we consider the case that the target graph $H$ is bipartite.
Then we extend the results to  general graphs $H$, with loops allowed.

\paragraph{Bipartite graphs $H$.} We first deal with the case that $H$ is bipartite (in particular, irreflexive), with bipartition classes $X$ and $Y$. Recall that we are interested in graphs $H$, for which the \lhomo{H} problem is NP-hard, i.e., graphs that are not co-circular-arc graphs. Moreover, we consider only connected graphs $H$ (as otherwise we can reduce to this case in polynomial time).

Let us present the high-level idea behind our algorithm for \lhomo{H}. Consider an instance $(G,L)$, such that $G$ is connected, and let $n=|V(G)|$. 
We may assume that $G$ is bipartite, as otherwise $(G,L)$ is clearly a no-instance. Furthermore, in any homomorphism from $G$ to $H$, each bipartition class of $G$ is mapped to a different bipartition class of $H$. We can assume that this is already reflected in the lists (we might have to solve two independent instances).

The algorithm is based on two main ideas.
First, observe that if $H$ contains two vertices $u,v$, which are in the same bipartition class, and each neighbor of $u$ is a neighbor of $v$, then the only thing preventing us from using $v$ instead $u$ is the fact that $v$ might not appear in some list containing $u$. Thus we might always assume that each list is an \emph{incomparable set}, i.e., it does not contain two vertices $u,v$ as above. By $i(H)$ we denote the size of a largest incomparable set contained in one bipartition class.

The second idea is related to a certain decomposition of $H$.
By a \emph{bipartite decomposition} we mean a partition of the vertex set of $H$ into three subsets $D,N,R$, such that: 
\begin{compactitem}
\item at least one of sets $(D \cap X) $ and $(D \cap Y)$ has at least 2 elements,
\item $N$ is non-empty and induces a biclique in $H$,
\item the sets $(D \cap X) \cup (N \cap Y)$ and $(D \cap Y) \cup (N \cap X)$ induce bicliques in $H$,
\item $N$ is a $D$-$R$-separator.
\end{compactitem}
We show that if $H$ has a bipartite decomposition, then we can reduce solving an instance $(G,L)$ of \lhomo{H} to solving several instances of \lhomo{H_1} and \lhomo{H_2}, where $H_1$ is the subgraph of $H$ induced by $D$, and $H_2$ is obtained from $H_2$ by collapsing $D \cap X$ and $D \cap Y$ to single vertices. 

This leads to the definition of $i^*(H)$ as the maximum value of $i(H')$ over all connected undecomposable induced subgraphs $H'$ of $H$, which are not complements of a circular-arc graph (a graph is undecomposable if it has no bipartite decomposition).
As our first result, we show that the algorithm exploiting decompositions recursively runs in time $i^*(H)^{\tw{G}} \cdot n^{\Oh(1)}$.

One might wonder whether some additional observations could be used to improve the algorithm.
As our second result, we show that this is not possible, assuming the SETH. This means that unless something unexpected happens in complexity theory, our algorithmic toolbox allows to solve \lhomo{H}, parameterized by the treewidth, as fast as possible.
More formally, we show the following theorem, which fully classifies the complexity of \lhomo{H} for bipartite graphs $H$.

\begin{theorem} \label{thm:main-bipartite}
Let $H$ be a connected bipartite graph, whose complement is not a circular-arc graph, and let $k=i^*(H)$.
Let $G$ be a bipartite graph with $n$ vertices and treewidth  $\tw{G}$.
\begin{compactenum}[(a)]
\item Even if $H$ is given as an input, the $\lhomo{H}$ problem with instance $(G,L)$ can be solved in time $k^{\tw{G}}\cdot(n\cdot |H|)^{\Oh(1)}$ for any lists $L$, provided that $G$ is given with an optimal tree decomposition.
\item Even if $H$ is fixed, there is no algorithm that solves  $\lhomo{H}$  for every $G$ and $L$ in time $(k-\epsilon)^{\tw{G}} \cdot n^{\Oh(1)}$ for any $\epsilon > 0$, unless the SETH fails.
\end{compactenum}
\end{theorem}

Note that for~\cref{thm:main-bipartite} a), if $H$ is not considered to be a constant, $n \cdot |H|$ is a natural measure of the size of an instance, as it is an upper estimate on the sum of sizes of all lists.

The main tool used in the proof of~\cref{thm:main-bipartite} b) is the following technical lemma.

\begin{restatable}[Constructing a $\rneq(S)$-gadget]{lemma}{inequality}%
\label{lem:edge-gadget}
Let $H$ be a connected, bipartite, undecomposable graph, whose complement is not a circular-arc graph.
Let $S$ be an incomparable set of $k\geq 2$ vertices of $H$, contained in one bipartition class.
Then there exists a $\rneq(S)$-gadget, i.e., a graph $F$ with $H$-lists $L$ and two special vertices $x,x' \in V(F)$, such that $L(x)=L(x')=S$ and
\begin{compactitem}
\item for any list homomorphism $h:(F,L)\to H$, it holds that $h(x) \neq  h(x')$,
\item for any distinct $s,s' \in S$ there is a list homomorphism $h:(F,L)\to H$, such that $h(x)=s$ and $h(x')=s'$.
\end{compactitem}
\end{restatable}

Let us point out that the graph constructed in \cref{lem:edge-gadget} can be seen as a \emph{primitive-positive definition} of the inequality relation on $S$ (see e.g. Bulatov~\cite[Section 2.1]{10.1007/978-3-319-77313-1_1}).
However, we prefer to present our results using purely combinatorial terms.

The proof of \cref{lem:edge-gadget} is technically involved, but as soon as we have it, the proof of~\cref{thm:main-bipartite} b) is straightforward. Consider an instance $G$ of $k \coloring$, where $k = i^*(H)$.
Let $H'$ be a connected, undecomposable, induced subgraph of $H$, whose complement is not a circular-arc graph, and contains an incomparable set $S$ of size $k$.
We construct a graph $G^*$ by replacing each edge $uv$ of $G$ with a copy of the $\rneq(S)$-gadget, given by \cref{lem:edge-gadget} (invoked for $H'$ and $S$), so that $u$ is identified with $x$ and $v$ is identified with $x'$.
By the properties of the gadget, we observe that $G^*$ has a list homomorphism to $H$ if and only if $G$ is a yes-instance of $k \coloring$. Furthermore, the construction of the $\rneq(S)$-gadget depends on $H$ only, and $H$ is assumed fixed, so we conclude that $\tw{G^*} = \tw{G} + \Oh(1)$. Therefore the statement of~\cref{thm:main-bipartite} b) follows from \cref{thm:LMS}.

\paragraph{General graphs $H$.} Next, we move to the general case. We aim to reduce the problem to the bipartite case.
The main idea comes from Feder, Hell, and Huang~\cite{DBLP:journals/jgt/FederHH03} who showed a close connection between the \lhomo{H} problem and the \lhomo{H^*} problem, where $H^*$ is the \emph{associated bipartite graph} of $H$, i.e., the bipartite graph with bipartition classes $\{v' \colon v \in V(H) \}$ and $\{ v'' \colon v \in V(H)\}$, where $u'v'' \in E(H^*)$ if and only if $uv \in E(H)$. Let us point out that we can equivalently define $H^*$ as a categorical (direct) product of $H$ and $K_2$~\cite{MR2817074}.

We extend the definition of $i^*$ to non-bipartite graphs by setting $i^*(H):=i^*(H^*)$. Let us point out that this definition is consistent with the definition for bipartite graphs, and with the definition of $i^*$ for reflexive graphs, introduced by Egri, Marx, and Rz\k{a}\.zewski~\cite{DBLP:conf/stacs/EgriMR18}.
We show the following theorem, fully classifying the complexity of the \lhomo{H} problem, parameterized by the treewidth of the instance graph.

\begin{theorem} \label{thm:main}
Let $H$ be a connected non-bi-arc graph (with possible loops), and let $k=i^*(H)$. Let $G$ be a graph with $n$ vertices and treewidth  $\tw{G}$.
\begin{compactenum}[(a)]
\item Even if $H$ is given as an input, the $\lhomo{H}$ problem with instance $(G,L)$ can be solved in time $k^{\tw{G}} \cdot (n\cdot |H|)^{\Oh(1)}$ for any lists $L$, provided that $G$ is given with an optimal tree decomposition.
\item  Even if $H$ is fixed, there is no algorithm that solves  $\lhomo{H}$  for every $G$ and $L$ in time $(k-\epsilon)^{\tw{G}} \cdot n^{\Oh(1)}$ for any $\epsilon > 0$, unless the SETH fails.
\end{compactenum}
\end{theorem}

As we mentioned before, both statements of~\cref{thm:main} follow from the corresponding statements in~\cref{thm:main-bipartite}.
For the algorithmic part, we define certain decompositions of general graphs $H$ and show that they coincide with bipartite decompositions $H^*$. This lets us reduce solving an instance $(G,L)$ of \lhomo{H} to solving some instances of \lhomo{H^*}.

On the complexity side, the reduction is even more direct: we show that an algorithm solving the \lhomo{H} problem on instances with treewidth $t$ in time $(i^*(H)-\epsilon)^t \cdot n^{\Oh(1)}$  could be used to solve the \lhomo{H^*} problem on instances with treewidth $t$ in time $(i^*(H^*)-\epsilon)^t \cdot n^{\Oh(1)}$, thus contradicting ~\cref{thm:main-bipartite}~b). 

In the conclusion of the paper we discuss how the decompositions defined for general graphs $H$ behave in two natural special cases: if $H$ is either reflexive of irreflexive. Recall that they correspond to the first two steps of the complexity dichotomy for \lhomo{H}~\cite{FEDER1998236,DBLP:journals/combinatorica/FederHH99}. We also analyze the complexity for \emph{typical} graphs $H$, and prove the following.

\begin{corollary}
For almost all graphs $H$ with possible loops the following holds.
Even if $H$ is fixed, there is no algorithm that solves  $\lhomo{H}$  for every instance $(G,L)$ in time $\Oh\left((|V(H)|-\epsilon)^{\tw{G}} \cdot n^{\Oh(1)}\right)$ for any $\epsilon > 0$, unless the SETH fails.
\end{corollary}

Finally, we show how to generalize our approach of reducing instances of \lhomo{H} to instances of \lhomo{H'}, where $H'$ is undecomposable.
We believe that this idea could be exploited to study the complexity of \lhomo{H} in various regimes, e.g., for different parameterizations of input instances.

\subsection{Comparison to the previous work.}

Let us briefly discuss similarities and differences between our work and previous, closely related results by Egri, Marx, and Rz\k{a}\.zewski~\cite{DBLP:conf/stacs/EgriMR18} (about the complexity of the \lhomo{H} problem for reflexive $H$), and by Okrasa and Rz\k{a}\.zewski~\cite{OkrasaSODA} (about the complexity of the \homo{H} problem).

At the high level, we follow the direction used by Egri \emph{et al.}~\cite{DBLP:conf/stacs/EgriMR18}, but since we generalize their result to \emph{all} relevant graphs $H$, the techniques become much more involved. The crucial idea was to reduce the problem to the bipartite case,
and to define decompositions of general graphs that correspond to the decompositions of $H^*$.
On the contrary, the case of reflexive graphs $H$ is much more straightforward.
In particular, there is just one type of decomposition that could be exploited algorithmically.
Also, the structure of ``hard'' subgraphs is much simpler in this case, so the necessary gadgets are significantly easier to construct.

On the other hand, in order to prove hardness for the \homo{H} problem, Okrasa and Rz\k{a}\.zewski~\cite{OkrasaSODA} used mostly \emph{algebraic tools} that are able to capture the global structure of a graph.
In contrast, our proofs are purely combinatorial. Furthermore, we are able to provide the full complexity classification for all graphs $H$, while the results of \cite{OkrasaSODA} are conditioned on two twenty-year-old conjectures.

\subsection{Notation.}

Let $H$ be a graph. By $\concomp{H}$ we denote the set of connected components of $H$. 
For a vertex $v$, by $N(v)$ we denote the \emph{neighborhood} of $v$, i.e., the set of vertices adjacent to $v$ (note that $v \in N(v)$ if and only if $v$ has a loop). For a set $U \subseteq V(H)$, we define $N(U):= \bigcup_{u \in U} N(u) \setminus U$ and $N[U]:= \bigcup_{u \in U} N(u) \cup U$. If $U=\{u_1, \ldots, u_k\}$, we omit one pair of brackets and write $N(u_1, \ldots, u_k)$ (respectively $N[u_1, \ldots, u_k]$) instead of $N(\{u_1, \ldots, u_k\})$ (respectively $N[\{u_1, \ldots, u_k\}]$).

We say that two vertices $x,y$ are \emph{comparable} if $N(y) \subseteq N(x)$ or $N(x) \subseteq N(y)$.
If two vertices are not comparable, we say that they are \emph{incomparable}. A set of vertices is \emph{incomparable} if all vertices are pairwise incomparable.

We say that a set $A \subseteq V(H)$ is \emph{complete} to a set $B$ if for every $a \in A$ and $b \in B$ the edge $ab$ exists. On the other hand, $A$ is \emph{non-adjacent} to $B$ if there are no edges with one endvertex in $A$ and the other in $B$.

Let $H$ be a bipartite graph, whose bipartition classes are denoted by $X$ and $Y$. For a set $S \subseteq V(H)$ and $Z \in \{X,Y\}$, by $S_Z$ we denote $S \cap Z$. For $A,B \subseteq V(H)$, we say that $A$ is \emph{bipartite-complete} to $B$ if $A_X$ is complete to $B_Y$ and $A_Y$ is complete to $B_X$.

\newpage
\section{Algorithm for bipartite target graphs}
Observe that we might always assume that $H$ is connected, as otherwise we can solve the problem for each connected component of $H$ separately. Furthermore, without losing the generality we may assume certain properties of instances of \lhomo{H} that we need to solve.

\begin{observation} \label{prop:lists-bipartite}
Let $(G,L)$ be an instance of \lhomo{H}, where $H$ is connected and bipartite with bipartition classes $X,Y$. Without loss of generality, we might assume the following.
\begin{compactenum}
\item The graph $G$ is connected and bipartite, with bipartition classes $X_G$ and $Y_G$,
\item $\bigcup_{x \in X_G} L(x) \subseteq X$ and $\bigcup_{y \in Y_G} L(y) \subseteq Y$,
\item for each $x \in V(G)$, the set $L(x)$ is incomparable.
\end{compactenum}
\end{observation}
\begin{proof}
\begin{compactenum}
\item If $G$ is not connected, we need to solve the problem separately for each connected component. If $G$ is not bipartite, then we can immediately report a no-instance.
\item Observe that in any homomorphism $f \colon (G,L) \to H$, either $f(X_G) \subseteq X$ and $f(Y_G) \subseteq Y$, or $f(X_G) \subseteq Y$ and $f(Y_G) \subseteq X$. Thus in order to solve $(G,L)$, we can separately solve two instances $(G,L')$ and $(G,L'')$ of \lhomo{H}, defined as follows. For each $x \in X_G$ we define $L'(x) := L(x) \cap X$ and $L''(x) := L(x) \cap Y$, and for each $y \in Y_G$ we define $L'(y) := L(y) \cap Y$ and $L''(y) := L(y) \cap X$. Then $(G,L)$ is a yes-instance of \lhomo{H} if and only if at least one of $(G,L')$ and $(G,L'')$ is a yes-instance of \lhomo{H}. Thus if we can solve each of instances $(G,L')$ and $(G,L'')$ in time $T(G)$, then we can solve the instance $(G,L)$ in time $2T(G)$.
\item If $N(u) \subseteq N(v)$ and both $u$ and $v$ appear on a list of some $x \in V(G)$, then in any homomorphism $f \colon (G,L) \to H$ with $f(x)=u$ we can always recolor $x$ to the color $v$. The obtained mapping is still a list homomorphism from $(G,L)$ to $H$, so we can safely remove $u$ from $L(x)$.
\end{compactenum}
\end{proof}

An instance of \lhomo{H} that respects conditions in \cref{prop:lists-bipartite} is called \emph{consistent}. From now on  we will restrict ourselves to consistent instances. Let us introduce a graph parameter, which will play a crucial role in our investigations.

\begin{definition}[$i(H)$]
For a bipartite graph $H$, by $i(H)$ we denote the maximum size of an incomparable set in $H$, which is fully contained in one bipartition class.
\end{definition}
 
Clearly for every $H$ we have $i(H) \leq |H|$. Note that by \cref{prop:lists-bipartite} we obtain the following.

\begin{corollary} \label{cor:listsize}
Let $(G,L)$ be a consistent instance of \lhomo{H}, where $H$ is bipartite. Then $\max\limits_{v \in V(G)} |L(v)|~\leq~i(H)$.
\end{corollary}

\subsection{Decomposition of bipartite graphs}\label{sec:decomposition}

Throughout this section we assume that the target graph $H$ is bipartite with bipartition classes $X$ and $Y$. In particular, it has no loops. 
Our algorithm for \lhomo{H} is based on the existence of a certain decomposition of $H$.

\begin{definition}[Bipartite decomposition]\label{def:bipartite-decomposition}
A partition of $V(H)$ into an ordered triple of sets $(D,N,R)$ is a \emph{bipartite decomposition} if the following conditions are satisfied.
\begin{compactenum}
\item $N$ is non-empty and separates $D$ and $R$, \label{it:bipdecomp-separator}
\item $|D_X| \geq 2$ or $|D_Y| \geq 2$, \label{it:bipdecomp-geq2}
\item $N$ induces a biclique in $H$, \label{it:bipdecomp-biclique}
\item $D$ is bipartite-complete to $N$. \label{it:bipdecomp-complete}
\end{compactenum}
\end{definition}

Since so far we only consider bipartite decompositions, we will just call them decompositions. Later on we will introduce other types of decompositions and then the distinction will be important.
If $H$ admits a decomposition, then it is \emph{decomposable}, otherwise it is \emph{undecomposable}.

For a graph $H$ with a decomposition $(D,N,R)$, the \emph{factors of the decomposition} are two graphs $H_1,H_2$ defined as follows. The graph $H_1$ is the subgraph of $H$ induced by the set $D$. The graph $H_2$ is obtained in the following way. 
For $Z \in \{X,Y\}$, if $D_Z$ is non-empty, then we contract it to a vertex $d_Z$. If there is at least one edge between the sets $D_X$ and $D_Y$, we add the edge $d_Xd_Y$. 

Note that both $H_1$ and $H_2$ are proper induced subgraphs of $H$. For $H_1$ is follows directly from the definition and $H_2$ can be equivalently defined as a graph obtained from $H$ by removing all but one vertex from $D_X$ (if $D_X \neq \emptyset$) and all but one vertex from $D_Y$ (if $D_Y \neq \emptyset$). We leave the vertices that are joined by an edge, provided that such a pair exists.

Now let us demonstrate how the bipartite decomposition can be used algorithmically.
Let $T(H,n,t)$ denote an upper bound for the complexity of \lhomo{H} on instances with $n$ vertices, given along a tree decomposition of width $t$. In the following lemma we do not assume that $|H|$ is a constant.

\begin{lemma}[Bipartite decomposition lemma]
\label{lem:decomposition}
Let $H$ be a bipartite graph with bipartition classes $X$ and $Y$, whose complement is not a circular-arc graph,
and suppose $H$ has a bipartite decomposition with factors $H_1,H_2$.
Assume that there are constants $\alpha \geq 1$, $c \geq 1$, and $d > 2$, such that $T(H_1,n,t) \leq \alpha \cdot c^t \cdot (n \cdot |H_1|)^d$ and $T(H_2,n,t) \leq \alpha \cdot c^t \cdot (n \cdot |H_2|)^d$.
Then $T(H,n,t)  \leq \alpha \cdot c^t \cdot (n \cdot |H|)^d$, if $n$ is sufficiently large.
\end{lemma}
\begin{proof}
Consider an instance $(G,L)$ of \lhomo{H}, recall that without loss of generality we may assume that it is consistent.
Let the bipartition classes of $G$ be $X_G$ and $Y_G$ and assume that $\bigcup_{x \in X_{G}} L(x) \subseteq X$ and $\bigcup_{y \in Y_{G}} L(y) \subseteq Y$.

Let $(D,N,R)$ be a bipartite decomposition of $H$.
We observe that for $Z \in \{X,Y\}$, and any two vertices $v \in D_Z, s \in N_Z$, we have $N(v) \subseteq N(s)$. Thus we may assume that no list contains both $s$ and $v$. Let $Q$ be the set of vertices of $G$ which have at least one vertex from $N$ in their lists.

\begin{claim} \label{clm:inside}
If there exists a list homomorphism $h \colon (G,L) \to H$, the image of each $C \in \concomp{G \setminus Q}$ is entirely contained either in $D$ or in $R$.
\end{claim}
\begin{claimproof}
By the definition of $\concomp{G \setminus Q}$, the image of $C$ is disjoint with $N$.  Suppose there exist $a,b \in C$, such that $h(a) = u \in D$ and $h(b) = r \in R$. Since $C$ is connected, there exists an $a$-$b$-path $P$ in $C$. The image of $P$ is an $u$-$r$-walk in $H$. But since $N$ separates $D$ and $R$ in $H$, there is a vertex of $P$, which is mapped to a vertex of $N$, a contradiction.
\end{claimproof}

Let us define lists $L_1$ as $L_1(x):=L(x) \cap D$, for every $x \in V(G) \setminus Q$.
For each $C \in \concomp{G \setminus Q}$, we check if there exists a homomorphism $h_C \colon (C,L_1) \to H_1$. Let $\mathbb{C}_1$ be the set of those $C \in \concomp{G \setminus Q}$, for which $h_C$ exists. By \cref{clm:inside}, we observe that if $C \notin \mathbb{C}_1$, then we can safely remove all vertices from $D$ from the lists of vertices of $C$.

Now consider the graph $H_2$. Let $Z \in \{X,Y\}$ and let $d_Z$ be the vertex to which the set $D_Z$ is collapsed (if it exists).
Let us define an instance $(G,L_2)$ of the \lhomo{H_2} problem, where the lists $L_2$ are as follows.
If $v \in Z_G$ is a vertex from some component of $\mathbb{C}_1$, then $L_2(v) := L(v) \setminus D_Z \cup \{d_Z\}$ (note that in this case $d_Z$ must exist).
If $v$ is a vertex from some component of $\concomp{G \setminus Q} \setminus \mathbb{C}_1$, then $L_2(v) := L(v) \setminus D_Z$. Finally, if $v \in Q$, then $L_2(v) := L(v)$. Note that the image of each list is contained in $V(H_2)$. Moreover, note that 
$\bigcup_{x \in X_G} L_2(x) \subseteq R_X \cup N_X \cup \{d_X\}$ and $\bigcup_{y \in Y_G} L_2(y) \subseteq R_Y \cup N_Y \cup \{d_Y\}$.

\begin{claim}\label{clm:equiv}
There is a list homomorphism $h : (G,L) \to H$ if an only if there is a list homomorphism $h' : (G,L_2) \to H_2$.
\end{claim}
\begin{claimproof}
First, assume that $h: (G,L) \to H$ exists. Define $h':V(G) \to V(H_2)$ in the following way:
\[
h'(v) = 
\begin{cases}
d_X & \text{ if } h(v) \in D_X,\\
d_Y & \text{ if } h(v) \in D_Y,\\
h(v) & \text{ otherwise.}
\end{cases}
\]
Clearly $h'$ is a homomorphism from $G$ to $H_2$, we need to show that it also respects lists $L_2$. Suppose otherwise and let $v$ be a vertex of $G$, such that $h'(v) \notin L_2(v)$. By symmetry, assume that $v \in X_G$, and thus $h'(v) \in X$. If $h'(v) \neq d_X$, then $h'(v) = h(v) \in (L(v) \setminus D_X) \subseteq L_2(v)$.
So suppose $h'(v) =d_X$ (and thus $h(v) \in D_X$) and $d_X \notin L_2(v)$. Observe that $v$ cannot be a vertex from $Q$, since $h$ maps $v$ to a vertex of $D_X$ and vertices from $Q$ do not have any vertices of $D$ in their lists. So the only case left is that $v$ belongs to some connected component $C$ of $G-Q$, which cannot be mapped to $H_1$. But then, by  \cref{clm:inside}, no vertex of $C$ is mapped to any vertex of $D$, so $h(v) \notin D_X$, a contradiction.

Now suppose there exists a list homomorphism $h' : (G,L_2) \to H_2$.  Define the following mapping $h$ from $V(G)$ to $V(H)$. If $h'(v) \notin \{d_X,d_Y\}$, then we set $h(v) := h'(v)$. Otherwise, if $h'(v) \in \{d_X,d_Y\}$, then $v$ is a vertex of some connected component $C \in \mathbb{C}_1$, and we define $h(v):=h_C(v)$.
Clearly $h$ preserves lists $L$: if $h(v) \notin D$, then $h(v) = h'(v) \in L_2(v) \setminus \{d_X,d_Y\} \subseteq L(v)$; otherwise we use $h_C$, which preserves lists $L$ by the definition.

Now suppose $h$ does not preserve edges, so there are vertices $u \in X_G$ and $v \in Y_G$, such that $uv$ is an edge of $G$ and $h(u)h(v)$ is not an edge of $H$. If $h'(u)= d_X, h'(v)=d_Y$, or $h'(u)\neq d_X, h'(v)\neq d_Y$, then $h(u)h(v)$ must be an edge of $H$, otherwise we get a contradiction by the definitions of $h_C$ (as since $x$ and $y$ are neighbors, they belong to the same $C$) and $h'$, respectively.
So, by symmetry, suppose $h'(u) =d_X$ and $h'(v) \neq d_Y$. But then we observe that $h(u) \in D_X$ and $h(v)=h'(v) \in N_Y$, since $h'$ is a homomorphism. And because $N_Y$ is complete to $D_X$, so $h(u)h(v)$ is an edge -- a contradiction.
\end{claimproof}

Computing $\concomp{G \setminus Q}$ can be done in time $\Oh(n \cdot |H| + n^2) = \Oh((n \cdot |H|)^2)$. Note that given a tree decomposition of $G$ of width at most $t$, we can easily obtain a tree decomposition of each $C \in \concomp{G \setminus Q}$ of width at most $t$.
Computing $h_C$ for all $C \in \concomp{G \setminus Q}$ requires time at most
\[
\sum_{C \in \concomp{G \setminus Q}} T(H_1,|C|,t) \leq \sum_{C \in \concomp{G \setminus Q}} \alpha \cdot c^t \cdot (|H_1|\cdot |C|)^d \leq \alpha \cdot c^t \cdot (|H_1|\cdot n)^d.
\]
The estimation follows from the facts that $\sum_{C \in \concomp{G \setminus Q}} |C| \leq n$, and $n^d$ is superadditive with respect to $n$, i.e., $n_1^d + n_2^d \leq (n_1+n_2)^d$.
Computing lists $L_2$ can be performed in time $\Oh(|H| \cdot n)$. Finally, computing $h'$ requires time $T(H_2,n,t) \leq \alpha \cdot c^t \cdot (|H_2|\cdot n)^d$.
The total running time is therefore bounded by:
\[
\Oh\left((n \cdot |H|)^2\right) + \alpha \cdot c^t \cdot (|H_1|\cdot n)^d + \Oh(|H| \cdot n) + \alpha \cdot c^t \cdot (|H_2|\cdot n)^d.
\]
With a careful analysis one can verify that the above expression is bounded by $\alpha \cdot c^t \cdot (|H|\cdot n)^d$, provided that $n$ is sufficiently large.
\end{proof}

\subsection{Solving \lhomo{H} for bipartite targets}

Let us define the main combinatorial invariant of the paper, $i^*(H)$:

\begin{definition}[$i^*(H)$ for bipartite $H$]\label{def:i_star}
Let $H$ be a connected bipartite graph, whose complement is not a circular-arc graph. Define
\begin{align*}
i^*(H) := \max \{ & i(H') \colon H' \text{ is an undecomposable, connected, induced}\\
& \text{ subgraph of }H, \text{ whose complement is not a circular-arc graph}\}.
\end{align*}
\end{definition}

Observe that if $H'$ is an induced subgraph of $H$, then $i(H') \leq i(H)$ and $i^*(H') \leq i^*(H)$, and thus $i^*(H) = i(H)$ for undecomposable $H$.  Furthermore, we always have $i^*(H) \leq i(H)$, which in turn is upper-bounded by the size of the larger bipartition class of $H$, i.e., the natural bound on the size of the lists in an instance of \lhomo{H}.

Let us point out that using \cref{def:bipartite-decomposition}  it is easy to create graphs $H$, for which $i^*(H)$ is arbitrarily smaller than this natural upper bound. For example, let $k \geq 3$ and let $H_0$ be obtained from $K_{k,k}$ by removing a perfect matching. 
Then $H_0$ is not the complement of a circular-arc graph, as it contains an induced $C_6$.
Now, for $j \geq 1$, the graph $H_j$ is obtained from $H_{j-1}$ by introducing another copy of $H_0$ and making all vertices from one bipartition class of $H_{j-1}$ complete to a vertex of $H_0$.
It is straightforward to verify that for every $j \geq 1$ it holds that  $i^*(H_j) = k+1$, while each bipartition class of $H_j$ has $j \cdot k$ vertices.

Now we are ready to present an algorithm solving \lhomo{H}, note that again we do not assume that $|H|$ is a constant.
We present the following, slightly more general variant of \cref{thm:main-bipartite} a), where we also do not assume that the tree decomposition of $G$ is optimal.

\begin{cthm}{3' a)} \label{thm:main-bipartite-algo} 
Let $H$ be a connected bipartite graph (given as an input) and let $(G,L)$ be an instance of \lhomo{H}, where $G$ has $n$ vertices and is given along a tree decomposition of width $t$.
Then there is an algorithm which decides whether $(G,L) \to H$ in time $\Oh\left(i^*(H)^t \cdot (n \cdot |H|)^{\Oh(1)}\right)$.
\end{cthm}

\begin{proof}
Clearly we can assume that $n$ is sufficiently large, as otherwise we can solve the problem in polynomial time by brute-force.

Observe that with $H$ we can associate a {\em recursion tree} $\cR$, whose nodes are labeled with induced subgraphs of $H$.
The root, denoted by $node(H)$ corresponds to the whole graph $H$.
If $H$ is undecomposable or is a complement of a circular-arc graph, then the recursion tree has just one node.
Otherwise $H$ has a decomposition with factors $H_1$ and $H_2$, and then $node(H)$ has two children, $node(H_1)$ and $node(H_2)$, respectively.
Recall that each factor has strictly fewer vertices than $H$, so we can construct $\cR$ recursively.
Clearly, each leaf of $\cR$ is either the complement of a circular-arc graph (and thus the corresponding problem is polynomial-time solvable), or is an undecomposable induced subgraph of $H$.
Note that a recursion tree may not be unique, as a graph may have more than one decomposition.
However, the number of leaves is bounded by $\Oh(|H|)$ (actually, with a careful analysis we can show that it is at most $|H|-2$), so the total number of nodes is $\Oh(|H|)$.
Furthermore,  it can be shown that in time polynomial in $H$ we can check if $H$ is undecomposable, or find a decomposition (we will prove it in \cref{lem:walks-s-v}).
Since recognizing circular-arc graphs (and therefore of course their complements) is also polynomial-time solvable, we conclude that $\cR$ can be constructed in time polynomial in $|H|$.

If $H$ is the complement of a circular-arc graph, then we solve the problem in polynomial time~\cite{DBLP:journals/combinatorica/FederHH99}.
If $H$ is undecomposable, we run a standard dynamic programming algorithm on a tree decomposition of $G$ (see \cite{Bodlaender:2008:COG:2479371.2479375,DBLP:conf/iwpec/BodlaenderBL13}). For each bag of the tree decomposition, and every partial list homomorphism from the graph induced by this bag to $H$ we indicate whether this particular partial homomorphism can be extended to a list homomorphism of the graph induced by the subtree rooted at this bag.
By \cref{cor:listsize}, the size of each list $L(x)$ for $x\in V(G)$ is at most $i(H)$, thus the complexity of the algorithm is bounded by $\alpha \cdot i(H)^t \cdot (n \cdot |H|)^d$ for some constants  $\alpha,d$. We can assume that $d > 2$, as otherwise we can always increase it.

So suppose $H$ is decomposable and let us show that we can solve the problem in time  $\alpha \cdot i^*(H)^t \cdot (n \cdot |H|)^d$. 
Let $\cR$ be a recursion tree of $H$ and recall that its every leaf corresponds to an induced subgraph of $H$ with strictly fewer vertices. Therefore, for any leaf node of $\cR$, corresponding to the subgraph $H'$ of $H$, we can solve every instance of \lhomo{H'} with $n$ vertices and a tree decomposition of width at most $t$ in time $\alpha \cdot i(H')^t \cdot (n \cdot |H'|)^d \leq \alpha \cdot i^*(H)^t \cdot (n \cdot |H'|)^d$.
Now, applying \cref{lem:decomposition} in a bottom-up fashion, we conclude that we can solve \lhomo{H} in time $\alpha \cdot i^*(H)^t \cdot (n \cdot |H|)^d = \Oh(i^*(H)^t \cdot (n\cdot |H|)^{\Oh(1)})$.
\end{proof}

\newpage
\section{Hardness for bipartite target graphs} \label{sec:hardness-prelim}
In this section we consider bipartite graphs $H$, which are not complements of circular-arc graphs.
Let us discuss two equivalent definitions of this class of graphs.

Let $X,Y$ be the bipartition classes of $H$ and consider a circle $C$ with two specified, distinct points $p$ and $q$.
A \emph{co-circular-arc representation} of $H$ is mapping of $V(H)$ to arcs of the circle $C$, such that (i) each arc corresponding to a vertex from $X$ contains $p$ but not $q$, (ii) each arc corresponding to a vertex from $Y$ contains $q$ but not $p$, and (iii) the vertices are adjacent if and only if their corresponding arc are disjoint.
It is known that the complement of $H$ is a circular-arc graph if and only if $H$ admits a co-circular-arc representation~\cite{DBLP:journals/jct/Spinrad88,DBLP:journals/gc/HellH97}.

However, for us it will be much more convenient to work with an equivalent definition of this class of graphs, discovered by Feder, Hell~\cite{DBLP:journals/combinatorica/FederHH99}.

\begin{definition}[Special edge asteroid] \label{def:asteroid}
Let $k\geq 1$ and consider $A_X=\{u_0,\ldots,u_{2k}\} \subseteq X$ and $A_Y=\{v_0,\ldots,v_{2k}\} \subseteq Y$.
The pair $(A_X,A_Y)$ is a \emph{special edge asteroid}  if, for each $i \in \{0,\ldots,2k\}$, the vertices $u_i$ and $v_i$ are adjacent, and there exists a $u_i$-$u_{i+1}$-path $\cW_{i,i+1}$ (indices are computed modulo $2k+1$), such that
\begin{compactenum}[(a)]
\item there are no edges between $\{u_i,v_i\}$ and $\{v_{i+k},v_{i+k+1}\}\cup \cW_{i+k,i+k+1}$ and
\item there are no edges between $\{u_0,v_0\}$ and $\{v_1,\ldots,v_{2k}\} \cup \bigcup_{i=1}^{2k-1} \cW_{i,i+1}$.
\end{compactenum}
\end{definition}

Feder, Hell, and Huang~\cite{DBLP:journals/combinatorica/FederHH99} proved the following characterization of bipartite graphs, whose complements are circular-arc graphs.

\begin{theorem}[Feder, Hell, and Huang~\cite{DBLP:journals/combinatorica/FederHH99}]\label{thm:cocircular}
A bipartite graph $H$ is the complement of a circular-arc graph if and only if $H$ does not contain an induced cycle with at least 6 vertices or a special edge asteroid.
\end{theorem}

We aim to prove \cref{thm:main-bipartite} b). Actually we will show a version, which gives the lower bound parameterized by the pathwidth of $G$. Clearly such statement will be stronger, as $\pw{G} \geq \tw{G}$. This corresponds to the pathwidth variant of \cref{thm:LMS}, also shown by Lokshtanov, Marx, Saurabh~\cite{DBLP:conf/soda/LokshtanovMS11a}.

\begin{cthm}{1'}[Lokshtanov, Marx, Saurabh~\cite{DBLP:conf/soda/LokshtanovMS11a}] \label{thm:LMS-pw}
For any $k \geq 3$, there is no algorithm solving $k \coloring$ on a graph with $n$ vertices and pathwidth $t$ in time $(k-\epsilon)^t \cdot n^{\Oh(1)}$, unless the SETH fails.
\end{cthm}

Thus we show the following strengthening of \cref{thm:main-bipartite} b).

\begin{cthm}{3' b)}\label{thm:main-hard-bipartite-pw}
Let $H$ be a fixed bipartite graph, whose complement is not a circular-arc graph.
Unless the SETH fails, there is no algorithm that solves the \lhomo{H} problem on instances with $n$ vertices and pathwidth $t$ in time $(i^*(H)-\epsilon)^t \cdot n^{\Oh(1)}$, for any $\epsilon >0$.
\end{cthm}

In order to prove \cref{thm:main-hard-bipartite-pw}, it is sufficient to show the following.

\begin{theorem} \label{thm:hardness-bipartite}
Let $H$ be a fixed connected bipartite undecomposable graph, whose complement is not a circular-arc graph.
Unless the SETH fails, there is no algorithm that solves the \lhomo{H} problem on instances with $n$ vertices and pathwidth $t$ in time $(i(H)-\epsilon)^t \cdot n^{\Oh(1)}$, for any $\epsilon >0$.
\end{theorem}

Let us show that \cref{thm:main-hard-bipartite-pw} and \cref{thm:hardness-bipartite} are equivalent.

\paragraph{(\cref{thm:hardness-bipartite} $\to$ \cref{thm:main-hard-bipartite-pw})} Assume the SETH and
suppose that \cref{thm:hardness-bipartite} holds and \cref{thm:main-hard-bipartite-pw} fails. So there is a bipartite graph $H$, whose complement is not a circular-arc graph, and an algorithm $A$ that solves $\lhomo{H}$ in time $(i^*(H) - \epsilon)^{\pw{G}} \cdot n^{\Oh(1)}$ for every input $(G,L)$, assuming that $G$ is given along with its optimal path decomposition.

Let $H'$ be an undecomposable connected induced subgraph of $H$, whose complement is not a circular-arc graph, and $i(H') = i^*(H)$. Let $(G,L')$ be an arbitrary instance of \lhomo{H'}.
Since $H'$ is an induced subgraph of $H$, the instance $(G,L')$ can be seen as an instance of \lhomo{H}, where no vertex from $V(H) \setminus V(H')$ appears in any list.
The algorithm $A$ solves this instance in time $(i^*(H) - \epsilon)^{\pw{G}} \cdot n^{\Oh(1)}=(i(H') - \epsilon)^{\pw{G}} \cdot n^{\Oh(1)}$, contradicting \cref{thm:hardness-bipartite}.

\paragraph{(\cref{thm:main-hard-bipartite-pw} $\to$ \cref{thm:hardness-bipartite}) }
Assume the SETH and suppose \cref{thm:main-hard-bipartite-pw} holds and \cref{thm:hardness-bipartite} fails. So there is a connected bipartite undecomposable graph $H$, whose complement is not a circular-arc graph, and an algorithm $A$ that solves $\lhomo{H}$ in time $(i(H) - \epsilon)^{\pw{G}} \cdot n^{\Oh(1)}$ for every input $(G,L)$.  But since $H$ is connected, bipartite, and undecomposable, we have $i^*(H) = i(H)$, so algorithm $A$ contradicts \cref{thm:main-hard-bipartite-pw}.

\subsection{Walks and obstructions: definitions and basic properties}
Let us start this section with introducing some definitions and properties that will be heavily used in our hardness proof.

We call a set of edges \emph{independent} if they form an induced matching.
A \emph{walk} $\cP$ is a sequence $\cP =p_1,\ldots, p_\ell$ of vertices of $H$, such that $p_ip_{i+1} \in E(H)$, for every $i \in [\ell-1]$. We call $|\cP|=\ell-1$ the \emph{length} of the walk $\cP$. By $\overline{\cP}$ we denote the walk $\cP$  reversed, i.e., $\overline{\cP} = p_\ell,\ldots, p_1$. If $|\cP| >0$, by $\cP\nobeg$ and $\cP\noend$, respectively, we denote the walks $p_2,\ldots, p_\ell$ and $p_1,\ldots, p_{\ell-1}$. If vertices $a$ and $b$ are, respectively, the first and the last vertex of a walk $\cP$, we say that $\cP$ is an $a$-$b$-walk and denote it by $\cP: a\to b$.  

Now let us define the relation of walks, which will be crucial in our hardness reductions.
\begin{definition}[Avoiding]
For walks $\cP = p_1,\ldots, p_\ell$ and $\cQ= q_1,\ldots, q_\ell$ of equal length, such that $p_1$ is in the same bipartition class as $q_1$, we say $\cP$ \emph{avoids $\cQ$} if $p_1 \neq q_1$ and for every $i \in [\ell-1]$ it holds that $p_iq_{i+1} \not\in E(H)$. 
\end{definition}

The following observation is immediate.

\begin{observation}\label{obs:priv-neighbours}
If $\cP= p_1,p_2\ldots,p_\ell$ avoids $\cQ=q_1,q_2,\ldots,q_\ell$ and $\ell \geq 2$, then $q_2 \in N(q_1) \setminus N(p_1)$ and  $p_{\ell-1} \in N(p_\ell) \setminus N(q_\ell)$. In particular,  $N(q_1) \not\subseteq N(p_1)$ and $N(p_\ell) \not\subseteq N(q_\ell)$. \qedhere
\end{observation}

For walks $\cP = p_1,\ldots ,p_\ell$ and $\cQ=q_1,\ldots ,q_k$ such that $p_\ell = q_1$, we define $\cP \circ \cQ:= p_1,\ldots, p_\ell,q_2, \ldots, q_k$ (in particular, $p_1,\ldots,p_\ell \circ p_\ell=p_1,\ldots,p_\ell$).
We observe the following.

\begin{observation}\label{obs:walks-composition}
If $\cP: x \to y$ avoids $\cQ:p \to q$ and $\cP': y \to z$ avoids $\cQ':q \to r$, then $\cP \circ \cP'$ avoids~$\cQ~\circ~\cQ'$. \qedhere
\end{observation}

We say that two disjoint sets $S_1$ and $S_2$ are \emph{non-adjacent} if there is no edge with one vertex in $S_1$ and the other in $S_2$.
If it does not lead to confusion, we will sometimes identify a walk $\cP$ with the set of vertices in $\cP$, so in particular we will say that walks $\cP$ and $\cQ$ are non-adjacent.
Clearly, if two walks of the same length are non-adjacent, they avoid each other.
For a walk $\cP= p_1,p_2\ldots,p_\ell$ and some $1\leq i\leq j \leq \ell$, we define the \emph{$p_i$-$p_j$-subwalk} of $\cP$, denoted by $\cP[p_i, p_j]$, as $\cP[p_i, p_j]:=p_i,p_{i+1},\ldots,p_{j-1},p_j$. (Note that the vertices $p_i$ and $p_j$ might appear on $\cP$ more than once, then $\cP[p_i, p_j]$ can be chosen arbitrarily.)  
For a set $\mathbb{P}$ of walks of equal length, we define the set $\mathbb{P}^{(i)} =\{v : \textrm{$v$ is the $i$-th vertex of some $\cP \in \mathbb{P}$}\}$.


\subsubsection{Special edge asteroids.}
Let us consider a special edge asteroid $(A_X,A_Y)$, where  $A_X=\{u_0,\ldots,u_{2k}\} \subseteq X$ and $A_Y=\{v_0,\ldots,v_{2k}\} \subseteq Y$. We use the notation from \cref{def:asteroid}.

Observe that  $\{u_0v_0,u_1v_1,u_{k+1}v_{k+1}\}$ is an independent set of edges.
Note that for $i \in \{0,\ldots,2k\}$, there might be many ways to choose $\cW_{i,i+1}$. We will select the actual path in the following way. First, among all possible paths that satisfy the conditions in the definition, we will only consider shortest ones (note that in particular they are induced). 
Second, if only possible, we will restrict our considerations to paths containing  $v_i$.
Finally, among all paths that are still in consideration, we will select one that contains $v_{i+1}$, if possible. 

In the following observation we show that edge asteroid is symmetric with respect to bipartition of $H$.

\begin{lemma}\label{obs:asteroid-dual}
Consider $A_X=\{u_0,\ldots,u_{2k}\} \subseteq X$ and $A_Y=\{v_0,\ldots,v_{2k}\} \subseteq Y$, such that $(A_X,A_Y)$ is a special edge asteroid. Let paths $\{\cW_{i,i+1}\}_{i=0}^{2k}$ be chosen as described above. 
Then for each $i \in \{0, \ldots, 2k\}$ there exists an induced $v_i$-$v_{i+1}$-path $\cW'_{i,i+1}$, such that
\begin{compactenum}[(a)]
\item there are no edges between $\{u_i,v_i\}$ and $\{u_{i+k},u_{i+k+1}\}\cup \cW'_{i+k,i+k+1}$,
\item there are no edges between $\{u_0,v_0\}$ and $\{u_1,\ldots,u_{2k}\} \cup \bigcup_{i=1}^{2k-1} \cW'_{i,i+1}$.
\end{compactenum}
Furthermore $A_X \cup A_Y \cup \bigcup_{i=0}^{2k} \cW_{i,i+1} = A_X \cup A_Y \cup \bigcup_{i=0}^{2k} \cW'_{i,i+1}$.
\end{lemma}
\begin{proof}
Let $i \in \{0, \ldots, 2k\}$. We obtain $\cW'_{i,i+1}$ by modifying the endpoints of $\cW_{i,i+1}$ as follows. 
If $v_i$ (respectively $v_{i+1}$) already belongs to $\cW_{i,i+1}$, then it must be the second (resp. last but one) vertex of $\cW_{i,i+1}$, as this path is induced. Thus we remove $u_i$ (resp. $u_{i+1}$).
In the other case, if $v_i$ (resp. $v_{i+1}$) does not appear on $\cW_{i,i+1}$, we append it as the first (resp. last) vertex. Clearly such constructed $\cW'_{i,i+1}$ is a $v_i$-$v_{i+1}$-path. Moreover, since $\cW'_{i,i+1} \subseteq \cW_{i,i+1} \cup \{v_i,v_{i+1}\}$, conditions (a) and (b) are satisfied by the definition of $\cW_{i,i+1}$. Moreover, note that $A_X \cup A_Y \cup \bigcup_{i=0}^{2k} \cW_{i,i+1} = A_X \cup A_Y \cup \bigcup_{i=0}^{2k} \cW'_{i,i+1}$.

So we only need to argue that $\cW'_{i,i+1}$ is an induced path. Clearly, if we only removed vertices $u_i$ or $u_{i+1}$ from $\cW_{i,i+1}$, then $\cW'_{i,i+1}$ is induced since $\cW_{i,i+1}$ is. So suppose that $v_i$ does not appear on $\cW_{i,i+1}$, but is adjacent to some vertex of $\cW_{i,i+1}$, other than $u_i$. It is straightforward to verify that 
in this case we could have chosen a shorter path $\cW_{i,i+1}$, or one that is equally long, but contains $v_i$, which contradicts the choice of $\cW_{i,i+1}$. The case if $v_{i+1}$ does not appear in $\cW_{i,i+1}$ is analogous.
\end{proof}

Note that \cref{obs:asteroid-dual} implies that if $(A_X,A_Y)$ is a special edge asteroid, then so is $(A_Y,A_X)$.
For $(A_X , A_Y)$ and paths $\{\cW_{i,i+1}\}_{i=0}^{2k}$ defined as above, we define the \emph{asteroidal subgraph} $\Ob$ to be the subgraph of $H$ induced by $A_X \cup A_Y \cup \bigcup_{i=0}^{2k} \cW_{i,i+1}$. To restore the symmetry of $(A_X,A_Y)$ and $(A_Y,A_X)$, we will only require that walks $\{\cW_{i,i+1}\}_{i=0}^{2k}$ and $\{\cW'_{i,i+1}\}_{i=0}^{2k}$ are induced. Observe that by \cref{obs:asteroid-dual}, if $\Ob$ is the asteroidal subgraph for $(A_X,A_Y)$, it is also the asteroidal subgraph for $(A_Y,A_X)$. Therefore, whenever we find in $H$ an asteroidal subgraph, we can freely choose the appropriate special edge asteroid, usually we will do it implicitly. Also, we will always use the notation introduced above.
In an asteroidal subgraph, in natural way we will treat paths $\cW_{i,i+1}$ and $\cW'_{i,i+1}$ as walks. Moreover, for each valid $i$ we define  $\cW_{i+1,i}:=\overline{\cW_{i,i+1}}$.

\subsubsection{Obstructions.}
An induced subgraph $\Ob$ of $H$ is an \emph{obstruction} if it is isomorphic to $C_6$ or $C_8$, or it is an asteroidal subgraph.
Note that each induced cycle with at least 10 vertices contains a special edge asteroid. Therefore we can restate~\cref{thm:cocircular} in the following way, which will be more useful for us.

\begin{corollary}
Let $H$ be a bipartite graph, which is not the complement of a circular-arc graph. Then $H$ contains an obstruction as an induced subgraph.
\end{corollary}

Let us show basic properties of obstructions. We denote the consecutive vertices of the cycle $C_k$ by $w_1,w_2, \ldots, w_k$ (with $w_1w_k \in E(C_k)$). For special edge asteroids, we use the notation from~\cref{def:asteroid}.

\begin{definition}[$C(\Ob)$]\label{def:corners}
For an obstruction $\Ob$, we define 
\begin{equation*}
C(\Ob):=\begin{cases}
\{(w_1,w_5), (w_2,w_4)\} & \text{if }\Ob\text{ is isomorphic to }C_6, \\
\{(w_1,w_5), (w_2,w_6)\} & \text{if }\Ob\text{ is isomorphic to }C_8, \\
\{(u_0,u_1), (v_0,v_1)\} & \text{if }\Ob\text{ is an asteroidal subgraph}.
\end{cases}
\end{equation*}
\end{definition}

For a pair $(\alpha, \beta) \in C(\Ob)$ we define $(\ovalp, \ovbet)$ to be the other element of $C(\Ob)$. Due to the symmetry of cycles and \cref{obs:asteroid-dual}, we will be able to assume symmetry between the elements of $C(\Ob)$, i.e., the properties of $(\alpha,\beta)$ that we will require will also be satisfied for $(\alpha',\beta')$. Thus in proofs it will usually be sufficient to consider one pair in $C(\Ob)$, as the proof for the other one will be analogous.

In the following two observations we show the existence of some useful walks in $\Ob$.

\begin{observation}\label{obs:asteroid-subwalks}
Let $\Ob$ be an obstruction in $H$ and let $C(\Ob)=\{(\alpha,\beta),(\ovalp, \ovbet)\}$. Then for every pair of vertices $x,y$ of $\Ob$, such that $x,y \not\in N[\alpha,\ovalp]$, there exists an $x$-$y$-walk $\cW[x,y]$, using only vertices of $\Ob$, which is non-adjacent to $\{\alpha,\ovalp\}$.
\end{observation} 
\begin{proof}
It is straightforward to observe that the graph $\Ob \setminus N[\alpha,\ovalp]$ is still connected, so the existence of claimed walks follows easily. The notation is justified by interpreting these walks in the case when $\Ob$ is an asteroidal subgraph as subwalks of 
\[
\cW:= \cW_{0,1}[p,u_1] \circ u_1,v_1,u_1 \circ \cW_{1,2} \circ u_2,v_2,u_2 \circ \cW_{2,3} \circ \ldots 
\circ \cW_{2k-1,2k} \circ u_{2k},v_{2k},u_{2k} \circ \cW_{2k,0}[u_{2k},q],
\]
where $p$ and $q$ are, respectively, the first vertex of $\cW_{0,1}$, which is not in $N[u_0,v_0]$, and the last vertex of $\cW_{2k,0}$, which is not in $N[u_0,v_0]$.
\end{proof}

\begin{observation}\label{obs:walks-between-corners}
Let $H$ be a bipartite graph with an obstruction $\Ob$, let $(\alpha, \beta) \in C(\Ob)$.
\begin{compactenum}[(a)]
\item There exist walks $\cX,\cX':\alpha \to \beta$ and $\cY,\cY':\beta \to \alpha$, such that $\cX$ avoids $\cY$ and $\cY'$ avoids $\cX'$.
\item If $\Ob$ is an asteroidal subgraph, let $\gamma$ be the vertex in $\{u_{k+1},v_{k+1}\}$ in the same bipartition class as $\alpha,\beta$, and let $c \in \{\alpha,\beta,\gamma\}$.
Then for any $a,b$, such that $\{a,b,c\} = \{\alpha,\beta,\gamma\}$, there exist walks $\cX_c:\alpha \to a \textrm{ and } \cY_c: \alpha \to b,$ and $\cZ_c: \beta \to c$, such that $\cX_c,\cY_c$ avoid $\cZ_c$ and $\cZ_c$ avoids $\cX_c,\cY_c$.
\end{compactenum}
All walks use only vertices of $\Ob$.
\end{observation}
\begin{proof}
Let us start with proving (a) in the case if $\Ob$ is either isomorphic to $C_6$ or to $C_8$, recall that consecutive vertices of such a cycle are denoted by $w_1,w_2,\ldots$. We can also assume that $\alpha = w_1$ and $\beta = w_5$, as the other case is symmetric.

If $\Ob \simeq C_6$, we set
\begin{alignat*}{12}
\cX&:=w_1,w_6,w_5,w_4,w_5 && \quad \cX'&:=w_1,w_2,w_3,w_4,w_5, \\
\cY&:=w_5,w_4,w_3,w_2,w_1 && \quad \cY'&:=w_5,w_6,w_1,w_2,w_1.
\end{alignat*}

If $\Ob \simeq C_8$, we set
\begin{alignat*}{12}
\cX=\cX'&:=w_1,w_2,w_3,w_4,w_5, \\
\cY=\cY'&:=w_5,w_6,w_7,w_8,w_1.
\end{alignat*}
It is straightforward to verify that these walks satisfy the statement (a).


So from now on let us consider the case that $\Ob$ is an asteroidal subgraph, and, by symmetry, assume that $\alpha = u_0$ and $\beta = u_1$. Then $\gamma = u_{k+1}$. Before we proceed to the proof of (a), let us define three auxiliary walks $\cA : u_0 \to u_0$, $\cB : u_0 \to u_{k+1}$, and $\cC : u_1 \to u_1$, as follows.
\begin{alignat*}{14}
\cA&:=\cU_0 &&~\circ~&&\cU_0 &&~\circ~&& \cU_0 &&~\circ~&& \cU_0 &&~\circ~&& \ldots &&~\circ~&& \cU_0 &&~\circ~&& \cU_0,\\
\cB&:=\cU_0 &&~\circ~&& \cW_{0,2k} &&~\circ~&& \cU_{2k} &&~\circ~&& \cW_{2k,2k-1} &&~\circ~&& \ldots &&~\circ~&& \cU_{k+2} &&~\circ~&& \cW_{k+2,k+1},\\
\cC&:=\cW[u_1,u_k] &&~\circ~&& \cU_k &&~\circ~&& \cW_{k,k-1} &&~\circ~&& \cU_{k-1} &&~\circ~&& \ldots &&~\circ~&& \cW_{2,1} &&~\circ~&& \cU_1,
\end{alignat*}
where $\cW[u_1,u_k]$ is obtained from \cref{obs:asteroid-subwalks} and by $\cU_i$ we mean a walk $u_i,v_i,\ldots,v_i,u_i$ of the appropriate length. Note that here we abuse the notation slightly, as the length of each $\cU_i$ might be different. By \cref{def:asteroid}, we have that $\cU_i$ and $\cW_{i+k,i+k+1}$ are non-adjacent, and so are sets $\{u_0,v_0\}$ and $\{v_1,\ldots,v_{2k}\} \cup \bigcup_{i=1}^{2k-1} \cW_{i,i+1}$. So it is straightforward to observe that $\cA,\cB$ avoid $\cC$ and $\cC$ avoids $\cA,\cB$.

In order to prove (a), we define
\begin{alignat*}{12}
\cX=\cX'&:=\cB &&~\circ~&& \cU_{k+1} &&~\circ~&&\cW[u_{k+1},u_1], \\
\cY=\cY'&:=\cC &&~\circ~&& \cW_{1,0}					  &&~\circ~&& \cU_0,
\end{alignat*}
where $\cW[u_{k+1},u_1]$ is given by \cref{obs:asteroid-subwalks}, and $\cU_i$ is defined as previously.
Observe that these walks avoid each other because $\cB$ and $\cC$ avoid each other, the walk $\cW_{1,0}$ is non-adjacent to $\{u_{k+1},v_{k+1}\}$, and $\cW[u_{k+1},u_1]$ is non-adjacent to $\{u_0,v_0\}$ by \cref{obs:asteroid-subwalks}.

Now let us show (b). Consider three cases.
If $c=u_0$, then by symmetry assume that $a=u_1,b=u_{k+1}$, and define 
\begin{alignat*}{12}
\cX_c&:= \cX &&~\circ~&& \cU_1, \\
\cY_c&:= \cX &&~\circ~&& \cW[u_{1},u_{k+1}], \\
\cZ_c&:= \cY &&~\circ~&&  \cU_0,
\end{alignat*}
where walks $\cX,\cY$ are given by statement (a) and $\cW[u_{1},u_{k+1}]$ is given by  \cref{obs:asteroid-subwalks}. Note that in case of asteroidal subgraph we have $\cX=\cX'$ and $\cY=\cY'$, and therefore $\cX,\cY$ avoid each other.

If $c=u_1$, then assume that $a=u_0,b=u_{k+1}$ and can observe that our auxiliary walks $\cA,\cB,\cC$ already satisfy the statement of the lemma. Indeed, is it sufficient to set
\begin{alignat*}{12}
\cX_c := \cA, \quad \cY_c := \cB, \quad \cZ_c := \cC.
\end{alignat*}
Finally, for $c=u_{k+1}$ and $a=u_0, b=u_1$, we can define
\begin{alignat*}{12}
\cX_c&:=  \cU_0	&&~\circ~&& \cU_0, \\
\cY_c&:=  \cU_0	&&~\circ~&& \cW_{0,1}, \\
\cZ_c&:= \cW[u_1,u_{k+1}]		&&~\circ~&& \cU_{k+1},
\end{alignat*}
where $\cW[u_1,u_{k+1}]$ is given by  \cref{obs:asteroid-subwalks}. This completes the proof.
\end{proof}
\subsection{Main ingredients}
In order to prove \cref{thm:hardness-bipartite} we will use several gadgets. 

\subsubsection{Building gadgets from walks.} First, let us show how we will use walks to create gadgets.
For a set $\DD = \{\cD_i\}_{i=1}^k$ of walks of equal length $\ell \geq 1$, let $P(\DD)$ be a path with $\ell+1$ vertices $p_1,\ldots,p_{l+1}$, such that the list of $p_i$ is $\DD^{(i)}$, i.e., the set of $i$-th vertices of walks in $\DD$.
The vertex $p_1$ will be called the \emph{input vertex} and $p_{\ell+1}$ will be called the \emph{output vertex}.

\begin{lemma}\label{lem:gadget-of-walks}
Let $\DD = \{\cD_i\}_{i=1}^k$ be a set of walks of equal length $\ell \geq 1$, such that $\cD_i$ is an $s_i$-$t_i$-walk.
Let $A,B$ be a partition of $\DD$ into two non-empty sets. Moreover, for $C \in \{A,B\}$, define $S(C) = \{s_i \colon \cD_i \in C\}$ and $T(C) = \{t_i \colon \cD_i \in C\}$. Suppose that $S(A) \cap S(B) = \emptyset$ and $T(A) \cap T(B) = \emptyset$, and every walk in $A$ avoids every walk in $B$. Then $P(\DD)$ with the input vertex $x$ and the output vertex $y$ has the following properties:
\begin{compactenum}[(a)]
\item $L(x) = S(A) \cup S(B)$ and $L(y) = T(A) \cup T(B)$,
\item for every $i \in [k]$ there is a list homomorphism $f_i \colon P(\DD) \to H$, such that $f_i(x) = s_i$ and $f_i(y)=t_i$,
\item for every $f \colon P(\DD) \to H$, if $f(x)\in S(A)$, then $f(y) \notin T(B)$.
\end{compactenum}
Furthermore, if every walk in $B$ avoids every walk in $A$, we additionally have 
\begin{compactenum}[(d)]
\item for every $f \colon P(\DD) \to H$, if $f(x)\in S(B)$, then $f(y) \notin T(A)$.
\end{compactenum}
\end{lemma}

\begin{proof}
For $i \in [k]$, let $d^i_j$ denote the $j$-th vertex of $\cD_i$, and denote consecutive vertices of $P(\DD)$ by $x=p_1,p_2,\ldots,p_{\ell+1}=y$.
The statement (a) follows directly from the definition of lists in $P(\DD)$.
For (b), consider any $i \in [k]$ and define $f_i(p_j):=d^i_j$ for $j\in [\ell+1]$. Observe that $f_i$ is indeed a list homomorphism $P(\DD) \to H$ since for every edge $p_jp_{j+1}$ of $P(\DD)$ it holds that $f_i(p_j)$ and $f_i(p_{j+1})$ are consecutive vertices $d^i_j$ and $d^i_{j+1}$ of the walk $\cD_i$ and thus they are adjacent in $H$. Moreover, clearly $f_i(x)=s_i$ and $f_i(y)=t_i$.

To show (c), suppose there exists a list homomorphism $f: P(\DD) \to H$ such that $f(p_1) \in S(A)$ and $f(p_{\ell+1}) \in T(B)$. Let $i \in [\ell+1]$ be the minimum integer such that there exists a walk $\cD_r$ in $B$ with $f(p_i)=d^r_i$. Note that it exists since $f(p_{\ell+1}) \in T(B)$. Moreover, $i>1$ since $f(p_1) \in S(A)$ and $S(A) \cap S(B) = \emptyset$. By minimality of $i$ we have that $f(p_{i-1})=d^s_{i-1}$ for some $\cD_s \in A$. Thus there is a walk $\cD_s \in A$ and a walk $\cD_r \in B$, such that $\cD_s$ does not avoid $\cD_r$, a contradiction. Similarly, we can show the property (d) by switching the roles of $A$ and $B$.
\end{proof}

For two gadgets $P = P(\DD)$ and $P'=P(\DD')$, such that the list of the output vertex of $P$ is the same as the list of the input vertex of $P'$, the \emph{composition} of $P$ and $P'$ is gadget $P''$ obtained by identifying the output vertex of $P$ and the input vertex of $P'$. The input and the output vertex of $P''$ are, respectively, the input vertex of $P$ and the output vertex of $P'$.

\subsubsection{Expressing basic relations.}
Let $H$ be an undecomposable, bipartite graph with an obstruction $\Ob$.
We will show that the structure of $\Ob$ is sufficiently rich to express some basic relations.

For a $k$-ary relation $\textrm{R}_k \subseteq V(H)^k$, by an $\textrm{R}_k$-gadget we mean a graph $F$ with $H$-lists $L$ and $k$ specified vertices $v_1,v_2, \ldots,v_k$, called \emph{interface}, such that for every $i \in [k]$ it holds that
\[
\{f(v_1)f(v_2)\ldots f(v_k) ~|~  f: (F,L) \to H \}  =\textrm{R}_k.
\]
In other words, the set of all possible colorings of the interface vertices that can be extended to a list $H$-coloring of the whole gadget is precisely the relation we are expressing. In the definition of an  $\textrm{R}_k$-gadget we do not insist that interface vertices are pairwise different.

Let $(\alpha, \beta) \in C(\Ob)$. To make the definitions more intuitive, let us assign logic values to vertices $\alpha,\beta$ in the following way: $\alpha$ will be interpreted as false, and $\beta$ will be interpreted as true.
First, we need to express a $k$-ary relation $\ork{k}=\{\alpha,\beta\}^k \setminus \{\alpha^k\}$ and a binary relation $\nand{2}=\{\alpha\alpha,\alpha\beta,\beta\alpha\}$.
Note that the relations $\ork{k}$ and $\nand{2}$ are symmetric with respect to interface vertices.

\begin{lemma}\label{lem:relationgadgets}
Let $H$ be a bipartite graph with an obstruction $\mathbb{O}$ and let $(\alpha,\beta) \in C(\Ob)$. For every $k \geq 2$ there exists an $\ork{k}$-gadget, and a $\nand{2}$-gadget.
\end{lemma}

\begin{proof}
First, notice that for any $k \ge 3$, the relation $\ork k$ can be expressed using a binary relation $\rneq=\{\alpha \beta, \beta \alpha \}$ and ternary $\ork 3$.
Indeed, note that an $\ork{k}$-gadget with interface vertices $x_1,x_2, \ldots,x_k$ can be easily constructed by introducing an $\ork{k-1}$-gadget with interface vertices $x_1,x_2,\ldots,x_{k-2},y$, an $\ork{3}$-gadget with interface vertices $y',x_{k-1},x_k$, and a $\rneq$-gadget with interface vertices $y,y'$ (this corresponds to the textbook NP-hardness reduction from \textsc{Sat} to \textsc{3-Sat}~\cite[Sec. 3.1.1.]{Garey:1979:CIG:578533}).

Let $\cX,\cY,\cX',\cY'$ be the walks given by \cref{obs:walks-between-corners}.
Let $P(\{\cX,\cY\})$ be the path obtained by applying \cref{lem:gadget-of-walks} for $A = \{\cX\}$ and $B=\{\cY\}$, and 
let $P(\{\cX',\cY'\})$ be the path obtained by applying \cref{lem:gadget-of-walks} for $A = \{\cY'\}$ and $B=\{\cX'\}$.
It is straightforward to observe that the graph obtained by identifying the first vertex of $P(\{\cX,\cY\})$ with the first vertex of $P(\{\cX',\cY'\})$, and the last vertex of $P(\{\cX,\cY\})$ with the last vertex of $P(\{\cX',\cY'\})$, is a $\rneq$-gadget (whose interface are the identified vertices).
So, in order to prove the lemma, we need to build gadgets for binary relations $\ork{2}$, $\nand{2}$, and the ternary relation $\ork{3}$.

Observe that an $\ork{2}$-gadget can be obtained by identifying two of three interface vertices of an $\ork{3}$-gadget.
Thus it is sufficient to show how to construct an $\ork{3}$-gadget and a $\nand{2}$-gadget. Let us first focus on constructing $\ork{3}$.

By symmetry of $\Ob$, we can assume that $\alpha =w_1, \beta = w_5$ if $\Ob$ is either an induced 6-cycle or an induced 8-cycle, or $\alpha=u_0,\beta=u_1$ if $\Ob$ is an asteroidal subgraph. Set $\gamma=w_3$ in the former case, or $\gamma = u_{k+1}$ in the latter one.

The high-level idea is to construct, for every triple $a,b,c$, such that $\{a,b,c\} =  \{\alpha,\beta,\gamma\}$, an $R(c)$-gadget for the relation $R(c)=\{\alpha a, \alpha b, \beta \alpha, \beta \beta, \beta \gamma \}$. 
This gadget will be called $P_{c}$, and its interface vertices will be $x$ and $y$, where $L(x)=\{\alpha,\beta\}$ and  $L(y)=\{\alpha,\beta,\gamma\}$.

Suppose we have constructed graphs $P_{\alpha},P_{\beta}$, and $P_{\gamma}$. Let $x_{1}$, $x_{2}$, $x_{3}$, respectively, be their $x$-vertices.
Consider a graph $G$ obtained by identifying the $y$-vertices of $P_{\alpha}, P_{\beta}$, and $P_{\gamma}$ to a single vertex $y$.
The lists of vertices remain unchanged, note that all $y$-vertices have the same list $\{\alpha,\beta,\gamma\}$, so identifying the vertices does not cause any conflict here.
Obviously we have $L(x_1)=L(x_2)=L(x_3)=\{\alpha,\beta\}$.

Observe it is not possible to map all $x_1,x_2,x_3$ to $\alpha$. However, every triple of colors from $\{\alpha,\beta\}^3 \setminus \{\alpha\alpha\alpha\}$ might appear on vertices $x_1,x_2,x_3$ in some list homomorphism from $G$ to $H$. Thus the graph $G$ satisfies the definition of an $\ork{3}$-gadget, whose interface vertices are $x_1,x_2,x_3$.
So let us show how to construct $P_{c}$ for every $c \in \{\alpha,\beta,\gamma\}$.

If $\Ob$ is isomorphic to $C_6$ or $C_8$, then the construction is straightforward and it is shown on \cref{fig:cycle-ors} (the picture shows  $\ork{3}$-gadgets, with $y$-vertices of $P_{c}$'s already identified).
For the case if $\Ob$ is an asteroidal subgraph, the construction will be quite similar, but a bit more involved.

\begin{figure}
\centering{\begin{tikzpicture}[every node/.style={draw,circle,fill=white,inner sep=0pt,minimum size=8pt},every loop/.style={}]

\node[fill=black!20,label=right:\footnotesize{$\{w_1,w_5\}$}] (a) at (0,0) {};
\node[label=right:\footnotesize{$\{w_4,w_6\}$}] (b) at (0,1) {};
\node[label=right:\footnotesize{$\{w_1,w_3\}$}] (c) at (0,2) {};
\node[label=right:\footnotesize{$\{w_2,w_4\}$}] (d) at (0,3) {};
\node[fill=black!20,label=right:\footnotesize{$\{w_1,w_5\}$}] (e) at (2,0) {};
\node[label=right:\footnotesize{$\{w_4,w_6\}$}] (f) at (2,1) {};
\node[label=right:\footnotesize{$\{w_1,w_3,w_5\}$}] (g) at (2,4) {};
\node[label=right:\footnotesize{$\{w_2,w_4\}$}] (h) at (4,3) {};
\node[label=right:\footnotesize{$\{w_3,w_5\}$}] (i) at (4,2) {};
\node[label=right:\footnotesize{$\{w_4,w_6\}$}] (j) at (4,1) {};
\node[fill=black!20,label=right:\footnotesize{$\{w_1,w_5\}$}] (k) at (4,0) {};
\draw (a) -- (b) -- (c) -- (d) -- (g) -- (f) -- (e);
\draw (g) -- (h) -- (i) -- (j) -- (k);
\end{tikzpicture}\hskip 2cm
\begin{tikzpicture}[every node/.style={draw,circle,fill=white,inner sep=0pt,minimum size=8pt},every loop/.style={}]

\node[fill=black!20,label=right:\footnotesize{$\{w_1,w_5\}$}] (a) at (4,0) {};
\node[label=right:\footnotesize{$\{w_6,w_8\}$}] (b) at (4,1) {};
\node[label=right:\footnotesize{$\{w_5,w_7\}$}] (c) at (4,2) {};
\node[label=right:\footnotesize{$\{w_4,w_6,w_8\}$}] (d) at (4,3) {};
\node[fill=black!20,label=left:\footnotesize{$\{w_1,w_5\}$}] (e) at (2,0) {};
\node[label=left:\footnotesize{$\{w_4,w_8\}$}] (f) at (2,1) {};
\node[label=left:\footnotesize{$\{w_1,w_3\}$}] (g) at (2,2) {};
\node[label=left:\footnotesize{$\{w_2,w_4\}$}] (h) at (2,3) {};
\node[label=right:\footnotesize{$\{w_1,w_3,w_5\}$}] (i) at (2,6) {};
\node[fill=black!20,label=left:\footnotesize{$\{w_1,w_5\}$}] (j) at (0,0) {};
\node[label=left:\footnotesize{$\{w_4,w_8\}$}] (k) at (0,1) {};
\node[label=left:\footnotesize{$\{w_5,w_7\}$}] (l) at (0,2) {};
\node[label=left:\footnotesize{$\{w_4,w_6\}$}] (m) at (0,3) {};
\node[label=left:\footnotesize{$\{w_3,w_5\}$}] (n) at (0,4) {};
\node[label=left:\footnotesize{$\{w_2,w_4\}$}] (o) at (0,5) {};
\draw (a) -- (b) -- (c) -- (d) -- (i) -- (h) -- (g) -- (f) -- (e);
\draw (i) -- (o) -- (n) -- (m) -- (l) -- (k) -- (j);
\end{tikzpicture}}
\caption{An $\ork 3$-gadget for $\Ob \simeq C_6$ (left) and for $\Ob \simeq C_8$ (right). Recall that the consecutive vertices of $\Ob$ are denoted by $(w_1,w_2,\ldots)$. The sets next to vertices indicate lists. Interface vertices are marked gray.}\label{fig:cycle-ors}
\end{figure}
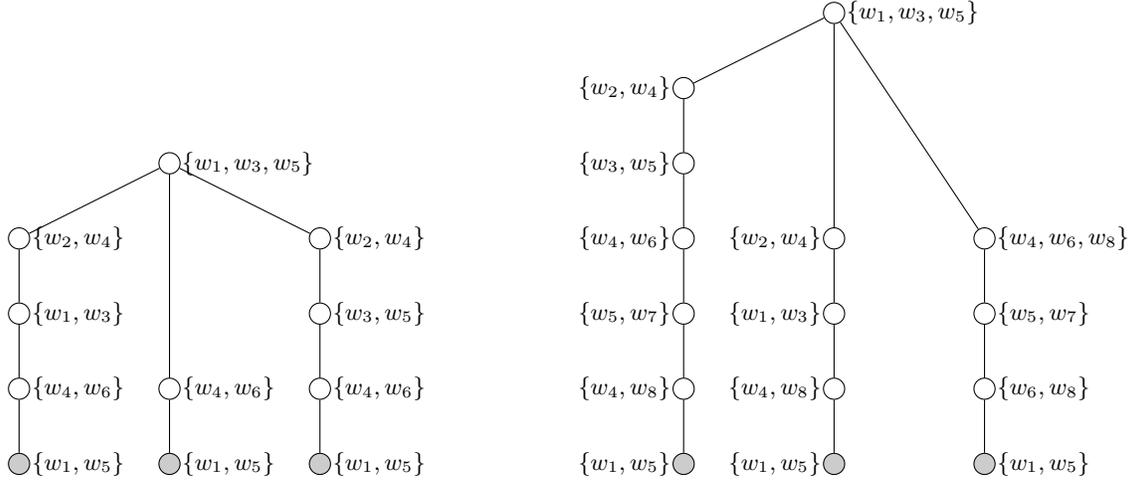

First, let us build an auxiliary gadget. Recall walks $\cW[u_1,u_{k+1}]$ and $\cW[u_{k+1},u_1]$ given by \cref{obs:asteroid-subwalks}, and consider the following walks of equal length:
\begin{alignat*}{12}
\cA&:= u_0,v_0,\ldots, u_0	&&~\circ~&& u_0,v_0,\ldots, u_0, \\
\cB_1&:= \cW[u_1,u_{k+1}]		&&~\circ~&& \cW[u_{k+1},u_1],    \\
\cB_2&:= \cW[u_1,u_{k+1}]		&&~\circ~&& u_{k+1},v_{k+1},\ldots,u_{k+1}.
\end{alignat*}
Note that $\cA$ is non-adjacent to $\cB_1,\cB_2$. Moreover, define 
\begin{alignat*}{12}
\cA_1&:=  u_0,v_0,\ldots, u_0, &&~\circ~&& u_0,v_0,\ldots, u_0, \\
\cA_2&:=  \cW_{1,0} &&~\circ~&& u_0,v_0,\ldots, u_0,\\
\cB&:=  u_{k+1},v_{k+1},\ldots,u_{k+1} &&~\circ~&& \cW[u_{k+1},u_1],
\end{alignat*}
and note that $\cA_1,\cA_2$ are non-adjacent to $\cB$. Thus, by \cref{lem:gadget-of-walks} we obtain that the graph $F$, which is a composition of $P(\{\cA,\cB_1,\cB_2\})$ and $P(\{\cA_1,\cA_2,\cB\})$, is a $\{u_0u_0,u_1u_0,u_1u_1 \}$-gadget.

Now, suppose we are given $a,b,c$, such that $\{a,b,c\} = \{u_0,u_1,u_{k+1}\}$.
Let $\cX_{c}: u_0 \to a$, $\cY_{c}: u_0 \to b$, and $\cZ_{c} : u_1 \to c$ be the walks given by \cref{obs:walks-between-corners}, recall that $\cX_c,\cY_c$ avoid $\cZ_c$ and $\cZ_c$ avoids $\cX_c,\cY_c$. Consider the graph $P(\{\cX_{c},\cY_{c},\cZ_{c}\})$, obtained by applying \cref{lem:gadget-of-walks} for $A = \{\cX_c,\cY_c\}$ and $B=\{\cZ_c\}$ and recall that it is a $\{u_0a,u_0b,u_1c\}$-gadget.
Now it is straightforward to observe that the composition of $F$ and $P(\{\cX_{c},\cY_{c},\cZ_{c}\})$ is an $R(c)$-gadget $P_c$ (see \cref{fig:r(c)-gadget}).
This concludes the construction of the $\ork{3}$-gadget.

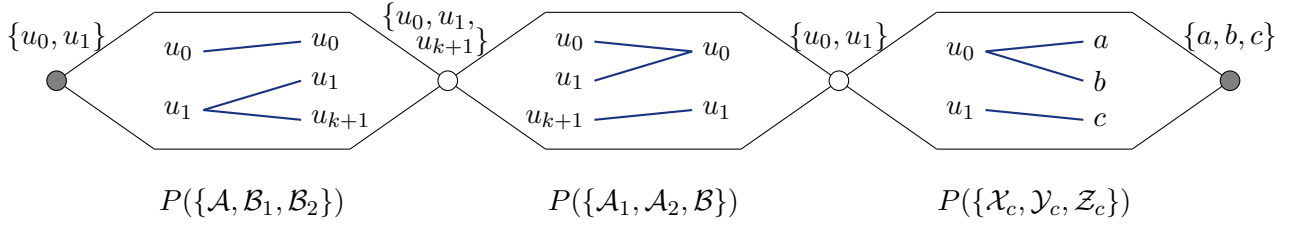
\begin{figure}
\centering{\begin{tikzpicture}[scale=1.3]
\foreach \i in {0,4,8}
{
\draw (\i,0)--(\i+1,0.7);
\draw (\i,0)--(\i+1,-0.7);
\draw (\i+1,0.7)--(\i+3,0.7);
\draw (\i+1,-0.7)--(\i+3,-0.7);
}
\foreach \i in {4,8,12}
{
\draw (\i,0)--(\i-1,0.7);
\draw (\i,0)--(\i-1,-0.7);
}
\draw (2,-1) node[below] {$P(\{\cA,\cB_1,\cB_2\})$};
\draw (6,-1) node[below] {$P(\{\cA_1,\cA_2,\cB\})$};
\draw (10,-1) node[below] {$P(\{\cX_c,\cY_c,\cZ_c\})$};
\draw[thick,color=blue] (1.5,0.3)--(2.5,0.4);
\draw[thick,color=blue] (1.5,-0.3)--(2.5,0);
\draw[thick,color=blue] (1.5,-0.3)--(2.5,-0.4);
\draw (1.5,0.3) node[left] {$u_0$};
\draw (1.5,-0.3) node[left] {$u_1$};
\draw (2.5,0.4) node[right] {$u_0$};
\draw (2.5,0) node[right] {$u_1$};
\draw (2.5,-0.4) node[right] {$u_{k+1}$};
\draw[thick,color=blue] (5.5,0.4)--(6.5,0.3);
\draw[thick,color=blue] (5.5,0)--(6.5,0.3);
\draw[thick,color=blue] (5.5,-0.4)--(6.5,-0.3);
\draw (5.5,0.4) node[left] {$u_0$};
\draw (5.5,0) node[left] {$u_1$};
\draw (5.5,-0.4) node[left] {$u_{k+1}$};
\draw (6.5,0.3) node[right] {$u_0$};
\draw (6.5,-0.3) node[right] {$u_1$};
\draw[thick,color=blue] (9.5,-0.3)--(10.5,-0.4);
\draw[thick,color=blue] (9.5,0.3)--(10.5,0);
\draw[thick,color=blue] (9.5,0.3)--(10.5,0.4);
\draw (9.5,-0.3) node[left] {$u_1$};
\draw (9.5,0.3) node[left] {$u_0$};
\draw (10.5,0.4) node[right] {$a$};
\draw (10.5,0) node[right] {$b$};
\draw (10.5,-0.4) node[right] {$c$};
\draw (0,0.2) node[above] {$\{u_0,u_1\}$};
\draw (3.8,0.4) node[above] {$\{u_0,u_1,$};
\draw (4.05,0.15) node[above] {$u_{k+1}\}$};
\draw (8,0.2) node[above] {$\{u_0,u_1\}$};
\draw (12,0.2) node[above] {$\{a,b,c\}$};
\draw[fill=white] (4,0) circle (0.1);
\draw[fill=white] (8,0) circle (0.1);
\draw[fill=gray] (0,0) circle (0.1);
\draw[fill=gray] (12,0) circle (0.1);
\end{tikzpicture}}
\caption{The $R(c)$-gadget $P_c$ as a composition of gadgets $P(\{ \cA,\cB_1,\cB_2\}), P(\{ \cA_1,\cA_2,\cB \})$, and $P(\{ \cX_c,\cY_c,\cZ_c\})$. Blue lines denote possible mappings of a pairs of colors that might appear in the input and the output vertex of each gadget.}
\label{fig:r(c)-gadget}
\end{figure}

Finally, observe that $\nand{2}$ can be obtained by composition of a $\rneq$-gadget, an $\ork{2}$-gadget (with one interface vertex as input and the other one as output), and a $\rneq$-gadget.
\end{proof}

\subsubsection{Distinguisher gadget.}
Now let us introduce another gadget, which will be the main tool used in our hardness proof.

\begin{definition}[Distinguisher]\label{def:distinguisher}
Let $S$ be an incomparable set in $H$ and let $(\alpha, \beta) \in \cC(\Ob)$, such that $\{\alpha,\beta\} \cup S$ is contained in one bipartition class of $H$.
Let $a,b \in S$.
A \emph{distinguisher} gadget is a graph $D_{a/b}$ with two specified vertices $x$ (called \emph{input}) and $y$ (called \emph{output}),
and $H$-lists $L$ such that:
\begin{compactenum}[(D1.)]
\item $L(x)=S$ and $L(y)=\{\alpha,\beta\}$,
\item there is a list homomorphism $\phi_a: (D_{a/b},L) \to H$, such that $\phi_a(x)=a$ and $\phi_a(y)=\alpha$,
\item there is a list homomorphism $\phi_b: (D_{a/b},L) \to H$, such that $\phi_b(x)=b$ and $\phi_b(y)=\beta$,
\item for any $c \in \cS \setminus \{a,b\}$ there is $\phi_c: (D_{a/b},L) \to H$, such that $\phi_c(x)=c$ and $\phi_c(y) \in \{\alpha,\beta\}$,
\item there is no list homomorphism $\phi: (D_{a/b},L) \to H$, such that $\phi(x)=a$ and $\phi(y)=\beta$.
\end{compactenum} 
\end{definition}

Distinguisher gadgets will be constructed using the following lemma, whose proof is postponed to \cref{sec:constructions}.

\begin{restatable*}{lemma}{lemWalksFromIncomp} \label{lem:walks-from-uncomp}
Let $H$ be a connected, bipartite, undecomposable graph, let $\Ob$ be an obstruction in $H$ and let $(\alpha,\beta) \in C(\Ob)$. Let $S$ be an incomparable set of $k\geq 2$ vertices of $H$, contained in the same bipartition class as $\alpha,\beta$.
Let $a$ and $b$ be two distinct vertices of $S$. For each $v \in S$ there exists a walk $\cD_v$, of length at least two, satisfying the following properties:
\begin{compactenum}[(1)]
\item for each $v \in S$, $\cD_v$ is a $v$-$\alpha$-walk or a $v$-$\beta$-walk, \label{prop:walks-lemma1}
\item $\cD_a$ is an $a$-$\alpha$-walk and $\cD_b$ is a $b$-$\beta$-walk, \label{prop:walks-lemma2}
\item for each $u,v \in S$, such that $\cD_u: u \to \alpha$ and $\cD_v: v \to \beta$, we have that $\cD_u$ avoids $\cD_v$.\label{prop:walks-lemma3}
\end{compactenum}
In particular, all walks have equal length.
\end{restatable*}

By \cref{lem:gadget-of-walks}, the existence of distinguisher gadgets follows directly from \cref{lem:walks-from-uncomp} applied for $A $ being the set of walks terminating at $\alpha$, and $B$ being the set of walks terminating at $\beta$.

\begin{corollary}\label{lem:distinguisher}
Let $H$ be an undecomposable bipartite graph with an obstruction $\Ob$ and let $(\alpha,\beta) \in C(\Ob)$. Let $S$ be an  incomparable set in $H$ contained in the same bipartition class as $\alpha,\beta$, such that $|S| \geq 2$. Then for every pair $(a,b)$ of distinct elements of $S$ there exists a distinguisher $D_{a/b}$. \qedhere
\end{corollary}

\subsection{Hardness reduction}\label{sec:hardness-bipartite}
Suppose we are able to construct all gadgets mentioned above.
The following lemma shows how to construct the main gadget, which will encode the inequality relation on an incomparable set $S$.
Formally, by $\rneq(S) \subseteq S^2$ we denote the relation on $S^2$, such that $ss' \in \rneq(S)$ if and only if $s \neq s'$.


\inequality*

\begin{proof}
Let $\Ob$ be an obstruction in $H$ and consider $(\alpha,\beta) \in C(\Ob)$, such that $\alpha, \beta$ and the elements of $S$ are in the same bipartition class of $H$. We denote the vertices of $S$ by $v_1,v_2,\ldots,v_k$.
We will construct our gadget in three steps.

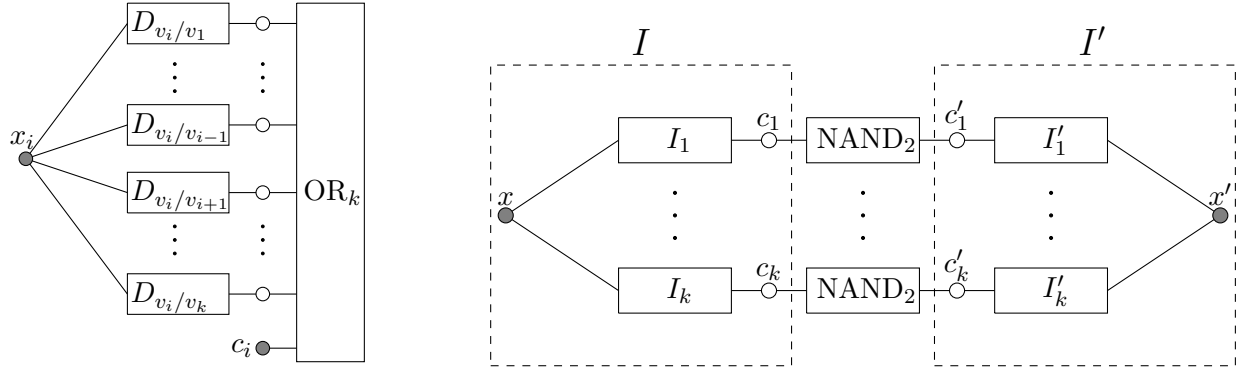
\begin{figure}
\centering{\begin{tikzpicture}[scale=0.9]
\draw (0,0)--(1.5,2);
\draw (0,0)--(1.5,-2);
\draw (0,0)--(1.5,0.5);
\draw (0,0)--(1.5,-0.5);
\draw[fill=gray] (0,0) circle (0.1);
\draw (-0.05,0) node[above] {$x_i$};
\draw (1.5,2.3)--(3,2.3)--(3,1.7)--(1.5,1.7)--(1.5,2.3);
\draw (1.5,-2.3)--(3,-2.3)--(3,-1.7)--(1.5,-1.7)--(1.5,-2.3);
\draw (1.5,-0.8)--(3,-0.8)--(3,-0.2)--(1.5,-0.2)--(1.5,-0.8);
\draw (1.5,0.8)--(3,0.8)--(3,0.2)--(1.5,0.2)--(1.5,0.8);
\draw (3,2)--(4,2);
\draw (3,-2)--(4,-2);
\draw (3,0.5)--(4,0.5);
\draw (3,-0.5)--(4,-0.5);
\draw[fill=white] (3.5,2) circle (0.1);
\draw[fill=white] (3.5,-2) circle (0.1);
\draw[fill=white] (3.5,0.5) circle (0.1);
\draw[fill=white] (3.5,-0.5) circle (0.1);
\draw (4,2.3)--(5,2.3)--(5,-3)--(4,-3)--(4,2.3);
\draw (3.5,-2.8)--(4,-2.8);
\draw[fill=gray] (3.5,-2.8) circle (0.1);
\draw (3.5,-2.8) node[left] {$c_i$};
\draw (1.4,1.95) node[right] {$D_{v_i/v_{1}}$};
\draw (1.4,-2.05) node[right] {$D_{v_i/v_{k}}$};
\draw (1.4,0.45) node[right] {$D_{v_i/v_{i-1}}$};
\draw (1.4,-0.55) node[right] {$D_{v_i/v_{i+1}}$};
\draw (3.95,-0.5) node[right] {$\ork k$};
\foreach \i in {1.4,1.2,1,-1,-1.2,-1.4}
{
\draw [fill=black] (2.2,\i) circle (0.02);
\draw [fill=black] (3.5,\i) circle (0.02);
}
\end{tikzpicture}} \hfill
\begin{tikzpicture}[scale=1]
\draw (0,0)--(1.5,1);
\draw (0,0)--(1.5,-1);
\draw (3,1)--(4,1);
\draw (5.5,1)--(6.5,1);
\draw (8,1)--(9.5,0);
\draw (3,-1)--(4,-1);
\draw (5.5,-1)--(6.5,-1);
\draw (8,-1)--(9.5,0);
\draw[fill=gray] (0,0) circle (0.1);
\draw[fill=gray] (9.5,0) circle (0.1);
\draw[fill=white] (3.5,1) circle (0.1);
\draw[fill=white] (3.5,-1) circle (0.1);
\draw[fill=white] (6,1) circle (0.1);
\draw[fill=white] (6,-1) circle (0.1);
\foreach \i in {0.3,0,-0.3}
{
\foreach \j in {2.25,4.75,7.25}
\draw [fill=black] (\j,\i) circle (0.02);
}
\draw (0,0) node[above] {$x$};
\draw (9.5,0) node[above] {$x'$};
\draw (3.5,1) node[above] {$c_1$};
\draw (6,1) node[above] {$c_1'$};
\draw (3.5,-1) node[above] {$c_k$};
\draw (6,-1) node[above] {$c_k'$};
\draw (1.8,2) node[above] {\Large{$I$}};
\draw (7.8,2) node[above] {\Large{$I'$}};
\draw (1.5,1.3)--(3,1.3)--(3,0.7)--(1.5,0.7)--(1.5,1.3);
\draw (1.5,-1.3)--(3,-1.3)--(3,-0.7)--(1.5,-0.7)--(1.5,-1.3);
\draw (4,1.3)--(5.5,1.3)--(5.5,0.7)--(4,0.7)--(4,1.3);
\draw (4,-1.3)--(5.5,-1.3)--(5.5,-0.7)--(4,-0.7)--(4,-1.3);
\draw (6.5,1.3)--(8,1.3)--(8,0.7)--(6.5,0.7)--(6.5,1.3);
\draw (6.5,-1.3)--(8,-1.3)--(8,-0.7)--(6.5,-0.7)--(6.5,-1.3);
\draw[dashed] (-0.2,-2) --++ (0,4) --++ (4,0) --++ (0,-4) --cycle;
\draw[dashed] (5.7,-2) --++ (0,4) --++ (4,0) --++ (0,-4) --cycle;
\draw (2,1) node[right] {$I_1$};
\draw (2,-1) node[right] {$I_k$};
\draw (7,1) node[right] {$I'_1$};
\draw (7,-1) node[right] {$I'_k$};
\draw (4,1) node[right] {$\nand 2$};
\draw (4,-1) node[right] {$\nand 2$};
\end{tikzpicture}
\caption{The graph $I_i$ with special vertices $x_i$ and $c_i$ (left) and $\rneq(S)$-gadget with interface vertices $x,x'$ (right).}
\label{fig:neq-gadget}
\end{figure}

\paragraph*{Step I.} In this step we will show that for every $i \in [k]$ we can construct a graph $I_i$ with two special vertices $x_i$ and $c_i$ and $H$-lists $L$, satisfying the following properties. 
\begin{compactitem}
\item $L(x_i) = S$ and $L(c_i)=\{\alpha,\beta\}$,
\item for every list homomorphism $\phi: (I_i,L) \to H$, if $\phi(x_i)=v_i$, then $\phi(c_i)=\beta$,
\item for every $j \neq i$ there exists a list homomorphism $\phi_j: (I_i,L) \to H$ such that $\phi_j(x_i)=v_j$ and $\phi_j(c_i)=\alpha$.
\end{compactitem}

Let us fix any $i \in [k]$.
For every $j \in [k] \setminus\{i\}$ we call \cref{lem:distinguisher} for $S,(\alpha,\beta$), and $a=v_i,b=v_j$ to construct a distinguisher gadget $D_{v_i/v_j}$ with the input $x_{i,j}$ and the output $y_{i,j}$.
We identify the vertices $x_{i,j}$, for all feasible $j$, to a single vertex $x_i$, and  introduce a new vertex $c_i$. Then we use \cref{lem:relationgadgets} to construct an $\ork k$ gadget and identify its $k$ interface vertices with distinct elements of $\{c_i\} \cup \{y_{i,j}\}_{j \neq i}$. This completes the construction of $I_i$ (see \cref{fig:neq-gadget}, left).

Recall that the properties of the $\ork{k}$-gadget imply that every list homomorphism from $I_i$ to $H$ maps at least one of vertices in $\{c_i\} \cup \{y_{i,j}\}_{j \neq i}$ to $\beta$.
By the property (D5.) in \cref{def:distinguisher},  for any $\phi : (I_i,L) \to H$ with  $\phi(x_i)=v_i$, and for every $j\neq i$ it holds that $\phi(y_{i,j})=\alpha$. This in turn forces $\phi(c_i)=\beta$.

On the other hand, by properties (D2.), (D3.), and (D4.), for any $j \neq i$ there is a homomorphism $\phi_j : (I_i,L) \to H$, such that $\phi_j(x_i)=v_j$ and $\phi(y_{i,j})=\beta$, which allows us to set $\phi_j(c_i)=\alpha$. So $I_i$ satisfies all desired properties.

\paragraph{Step II.} In this step we will construct a graph $I$ with $H$-lists $L$ and special vertices $x,c_1,\ldots,c_k$, satisfying the following properties.
\begin{compactitem}
\item $L(x)=S$ and $L(c_i)=\{\alpha,\beta\}$ for every $i\in[k]$,
\item for every list homomorphism $\phi: (I,L) \to H$, if $\phi(x)=v_i$, then $\phi(c_i)=\beta$.
\item for every $i \in [k]$ there exists a list homomorphism $\phi_i : (I,L) \to H$, such that $\phi_i(x)=v_i$ and $\phi_i(c_i)=\beta$, and $\phi_i(c_j)=\alpha$ for every $j \in [k] \setminus \{i\}$.
\end{compactitem}

The graph $I$ is constructed by introducing $k$ gadgets $I_1,\ldots,I_k$, and identifying the vertices $x_1,\ldots,x_k$ into a single vertex $x$ (see \cref{fig:neq-gadget}, right). The desired properties of $I$ follow directly from properties of $I_i$'s.

\paragraph{Step III.}
Finally, we can construct a $\rneq(S)$-gadget.
We introduce two copies of the gadget from the previous step, call them $I$ and $I'$ (we will use primes to denote the appropriate vertices of $I'$). For each $i \in [k]$, we introduce a $\nand{2}$-gadget and identify its interface vertices with vertices $c_i$ and $c'_i$ (see \cref{fig:neq-gadget}, right). Let us call such constructed graph $F$. We claim that $F$ is a $\rneq(S)$-gadget, whose interface vertices are $x$ and $x'$.

Clearly, $L(x)=L(x')=S$. Suppose that $\phi : (F,L) \to H$ is a list homomorphism such that $\phi(x)=\phi(x')=v_i$. Then, by definition of $I$, we have $\phi(c_i)=\phi(c_i')=\beta$, but this is impossible due to the properties of the $\nand{2}$-gadget joining $c_i$ and $c'_i$.

On the other hand, let us choose any distinct $v_i,v_j \in S$. We can color $I$ according to the homomorphism $\phi_i$, and $I'$ according to the homomorphism $\phi_j$ (both $\phi_i$ and $\phi_j$ are defined in Step II). In particular this means that $x$ is mapped to $v_i$, $x'$ is mapped to $v_j$,  $c_i$ and $c'_j$ are mapped to $\beta$, and all other vertices in $\{c_1,\ldots,c_k\} \cup \{c'_1,\ldots,c'_k\}$ are mapped to $\alpha$. Since $i \neq j$, by the definition of a $\nand{2}$-gadget, we can extend such defined mapping to all vertices of $F$. This completes the proof of the lemma.
\end{proof}

Now, equipped with \cref{lem:edge-gadget}, we can easily prove \cref{thm:hardness-bipartite}.

\begin{proof}[Proof of \cref{thm:hardness-bipartite}]
Recall that $H$ is an undecomposable bipartite graph, whose complement is not a circular-arc graph.
Let $S$ be the largest incomparable set in $H$, contained in one bipartition class. Let $k = |S|$, i.e., $k = i(H)$. 

Observe that since the complement of $H$ is not a circular-arc graph, we have that $k \ge 3$.
Indeed, recall that $H$ contains an obstruction, which is either an induced $C_6$, an induced $C_8$, or an asteroidal subgraph. Observe that all vertices from one bipartition class of $C_6$ or $C_8$ form an incomparable set of size at least 3. On the other hand, recall that a special edge asteroid contains at least three independent edges, so their appropriate endvertices form the desired incomparable set.

We reduce from $k \coloring$, let $G$ be an instance. Clearly we can assume that $G$ is connected and has at least 3 vertices. We will construct a graph $G^*$ with $H$-lists $L$ and the following properties:
\begin{compactenum}[(a)]
\item $(G^*,L) \to H$ if and only if $G$ is $k$-colorable,
\item the number of vertices of $G^*$ is at most $g(H) \cdot |E(G)|$ for some function $g$ of $H$,
\item the pathwidth of $G^*$ is at most $g(H)+\pw{G}$,
\item $G^*$ can be constructed in time $(|V(G)|)^{\Oh(1)} \cdot g'(H)$ for some function $g'$.
\end{compactenum}

Observe that this will be sufficient to prove the theorem. Indeed, suppose that for some $\epsilon >0$ we can solve $\lhomo{H}$ in time $\Oh^*((k-\epsilon)^t)$ on instances of pathwidth $t$.
Let us observe that applying this algorithm to $G^*$ gives an algorithm solving the $k \coloring$ problem on $G$ in time
\[
(k-\epsilon)^{\pw{G^*}} \cdot |V(G^*)|^{\Oh(1)} \leq (k-\epsilon)^{\pw{G} + g(H)} \cdot \left( g(H) \cdot |E(G)| \right)^{\Oh(1)} = (k-\epsilon)^{\pw{G}} \cdot |V(G)|^{\Oh(1)},
\]
where the last step follows since $|H|$ is a constant. Recall that by \cref{thm:LMS-pw}, the existence of such an algorithm for $k \coloring$ contradicts the SETH.

We start the construction of $G^*$ with the vertex set of $G$. The lists of these vertices are set to $S$.
Then, for each edge $uv$ of $G$, we introduce a copy $F_{uv}$ of the $\rneq(S)$-gadget introduced in \cref{lem:edge-gadget}. We identify the interface vertices of this gadget with $u$ and $v$, respectively. This completes the construction of $G^*$. Let us show that it satisfies the properties (a)--(d).

Note that (a) follows directly from \cref{lem:edge-gadget}. Indeed, consider an edge $uv$ of $G$.
On one hand, for every list homomorphism $f : (G^*,L) \to H$ we have that $f(u) \neq f(v)$.
On the other hand, mapping $u$ and $v$ to any distinct vertices from $S$ can be extended to a homomorphism of the whole graph $F_{uv}$.

To show (b), recall that the number of vertices of each $F_{uv}$ depends only on $H$, let it be $g(H)$.
Every original vertex of $G$ belongs to some gadget in $G^*$, so $G^*$ contains at most $g(H) \cdot |E(G)|$ vertices.

Next, consider a path decomposition $\cT$ of $G$ of width $\pw{G}$, let the consecutive bags be $X_1,X_2,\ldots,X_\ell$.
We extend $\cT$ to a path decomposition $\cT^*$ of $G^*$ as follows: for every edge $uv$ in $G$ we choose one bag $X_{i}$ such that $u,v \in X_{i}$, and we add a new bag $X'_{i} = X_{i} \cup V(F)$, which becomes the immediate successor of $X_{i}$. We repeat this for every edge, making sure that for $X_{i}$ we can only choose the original bags coming from $\cT$.
Note that it might happen that we will insert several new bags in a row, if the same $X_i$ was chosen for different edges, but this is not a problem.
Is it straightforward to observe that $\cT^*$ is a path decomposition of $G^*$, and the width of $\cT^*$ is at most $g(H) + \pw{G}$. This proves (c).

Finally, it is straightforward to observe that the construction of $G^*$ was performed in time polynomial in $G$ (recall that we treat $H$ as a constant-size graph).
\end{proof}

\subsection{Technical details of \cref{lem:walks-from-uncomp}}\label{sec:constructions}
In this section we show how to construct distinguisher gadgets, which were the main technical tool used in the proof of \cref{thm:hardness-bipartite}.

\subsubsection{Constructing avoiding walks from incomparable vertices to a specified vertex.}

Intuitively, we aim to show that for any two incomparable vertices $s,v$, and a third vertex $t$, we can either construct walks from $s$ to $t$ and  from $v$ to $v$, or walks from $v$ to $t$  and from $s$ to $s$, which satisfy certain avoiding conditions. However, due to some corner cases, the actual statement of the lemma is much more complicated.

\begin{lemma}\label{lem:two-walks-near}
Let $H$ be a bipartite graph with bipartition classes $X$ and $Y$. Let $s,v \in X$ be incomparable vertices and let $s' \in N(s) \setminus N(v)$, $v' \in N(v) \setminus N(s)$. Let $S \subseteq N(s,s',v,v')$ be non-empty. Let $U$ be the set of vertices reachable from $\{s,v\}$ in $H - S$. We also assume that $H$ has no decomposition $(D,N,R)$ such that $\{s,v,s',v'\} \subseteq D$ and $S \subseteq N$. Then there is $y \in S$ and $x \in U$ (both in the same bipartition class) and two pairs of walks of length at least one: 
\begin{compactenum}
\item $\cA, \cA': s \to y$ and $\cB, \cB': v \to x$, or \label{two-walks-near-statement1}
\item $\cA, \cA': v \to y$ and $\cB, \cB': s \to x$,\label{two-walks-near-statement2}
\end{compactenum}
such that $\cA$ avoids $\cB$, $\cB'$ avoids $\cA'$. Moreover, for every $i$ it holds that $\{\cA,\cB\}^{(i)} \not\subseteq S$ and all four walks are entirely contained in $S \cup U$.

Furthermore one of the following holds:
\begin{compactenum}[a)]
\item $x \in \{v,v'\}$ in case  (\ref{two-walks-near-statement1}.) and $x \in \{s,s'\}$ in case (\ref{two-walks-near-statement2}.), or \label{two-walks-near-additional-a}
\item walks $\cB, \cB'$ are entirely contained in $U$.\label{two-walks-near-additional-b}
\end{compactenum}
\end{lemma}

\begin{proof} We split the proof into three cases.


\paragraph{Case 1: There exists $a \in N(s,v)$  and $b \in S_X$ such that $ab \notin E(H)$.} By symmetry assume that $a \in N(s)$. If $b \in N(v')$, we set $x=s, y=b$, and define $\cA=\cA'=v,v',b$ and $\cB=s,s',s, \cB'=s,a,s$, obtaining walks as in the statement (\ref{two-walks-near-statement2}.\ref{two-walks-near-additional-a}). On the other hand, if $b \notin N(v')$, then necessarily $b \in N(s')$. Then we take $x=v, y=b$, and define $\cA=\cA'=s,s',b,$ and $\cB=\cB'=v,v',v$, which satisfy statement (\ref{two-walks-near-statement1}.\ref{two-walks-near-additional-a}). 


\paragraph{Case 2: There exists $a \in S_Y$ and $b \in N(s',v')$ such that $ab \notin E(H)$.}
By symmetry assume that $b \in N(s')$. If $a \in N(v)$, then we set $x=s', y=a$, and define $\cA=v,a$, $\cB = s,s'$, $\cA' = v,v',v,a$, and $\cB' = s,s',b,s'$, obtaining statement (\ref{two-walks-near-statement2}.\ref{two-walks-near-additional-a}).
On the other hand, if $a \notin N(v)$, then $a \in N(s)$. In this case we take $x=v', y=a$, and define $\cA=\cA'=s,a$ and $\cB=\cB'=v,v'$, obtaining statement (\ref{two-walks-near-statement1}.\ref{two-walks-near-additional-a}).


\paragraph{Case 3: $S$ is bipartite-complete to $N[s,s',v,v']$.} Note that this implies that $H[S]$ is a biclique.

Let $K_{Y}$ (resp. $K_X$) be the set of all vertices $w$ in $U_{Y}$ (resp. $U_X$) such that there exists $s_w \in S_{X}$ (resp. $S_Y$) such that $ws_w \notin E(H)$. Since $H$ has no decomposition $(D,N,R)$ such that $s,v,s',v' \in D$ and $S \subseteq N$, we have that $K := K_{X} \cup K_{Y} \neq \emptyset$.
Otherwise, $U$ is bipartite-complete to $S$ and there are no edges between $U$ and $H-(S\cup U)$, so $\bigl(U,S, V(H) \setminus (S \cup U)\bigr)$ is a decomposition of $H$, a contradiction. Moreover, since we are in Case 3, we observe that $K \cap N[s,v,s',v'] = \emptyset$. Let $\widetilde{S}$ be a minimal $\{s,s',v,v'\}$-$K$-separator in $H[U]$, contained in $N(s,v,s',v')$, and let $\widetilde{U}$ be the set of vertices reachable from $\{s,s',v,v'\}$ in $H[U]-\widetilde{S}$. Note that $\widetilde{U} \cup \widetilde{S}$ is bipartite-complete to $S$.

For contradiction, suppose that the lemma does not hold and $H$ is a counterexample with the minimum number of vertices, i.e., for any graph $H'$ with strictly fewer vertices and any choice of four vertices and a set, satisfying the assumptions of the lemma, the desired walks exist.

\begin{claim}
The graph $H[U]$, vertices $s,s',v,v'$, and the set $\widetilde{S}$ satisfy the assumptions of the lemma (where the role of $S$ is played by $\widetilde{S}$).
\end{claim}
\begin{claimproof}
Note that $s,s',v,v' \in U$, which in particular implies that the vertices $s,v$ are incomparable in $H[U]$.
Also, the set $\widetilde{S}$ is contained in $N(s,s',v,v') \cap U$. 

Suppose that $H[U]$ has a decomposition $(D,N,R)$, where $\{s,s',v,v'\} \subseteq D$ and $\widetilde{S} \subseteq N$.
Since $N$ is disjoint with $K$, we know that $N$ is bipartite-complete to $S$, implying that $S \cup N$ is a biclique.

Hence, $D$ is bipartite-complete to $S \cup N$ and there are no more vertices in $H$ adjacent to $D$, so we obtained a decomposition $(D,N \cup S, V(H) \setminus (D \cup N \cup S))$ of $H$, such that $s,s',v,v' \in D$ and $S \subseteq (N \cup S)$, a contradiction. 
\end{claimproof}

Since $S \neq \emptyset$ and $S \subseteq V(H) - U$, we obtain that $H[U]$ has strictly fewer vertices than $H$. So, by minimality of $H$, there exist $\widetilde{x} \in \widetilde{U}, \widetilde{y} \in \widetilde{S}$  and one of the following quadruples of walks:
\begin{alignat*}{4}
\cC,\cC': s \to \widetilde{y} & \quad \text{ and } \quad & \cD,\cD': v \to \widetilde{x} & \quad \text{ (statement (\ref{two-walks-near-statement1}.)), \quad or}\\
\cC,\cC': v \to \widetilde{y} & \quad \text{ and } \quad & \cD,\cD': s \to \widetilde{x} & \quad \text{ (statement (\ref{two-walks-near-statement2}.))},
\end{alignat*}
where $\cC$ avoids $\cD$ and $\cD'$ avoids $\cC'$. 

Let us consider the first situation (statement (\ref{two-walks-near-statement1}.)), i.e., that we obtained walks $\cC,\cC': s \to \widetilde{y}$ and $\cD,\cD': v \to \widetilde{x}$ (the other one is symmetric).

By minimality of $\widetilde{S}$ there exists a walk $\cP: \widetilde{y} \to w$, where $w \in K$ and no vertex from $\cP$ except $\widetilde{y}$ is in $\widetilde{S}$. This means that no vertex from $\cP$, except for the first vertex, is adjacent to any vertex from $\widetilde{U}$, so, in particular, $w$ is non-adjacent to $\widetilde{x}$. 
Observe that since $K \cap \widetilde{S} = \emptyset$ we must have $|\cP| \ge 1$.
Let $w^\bullet$ be the last vertex of $\mathcal{P}\noend$ (recall that by $\cP\noend$ we denote the walk $\cP$ with the last vertex removed); note that we may have $w^{\bullet}=\widetilde{y}$.

Let $\widetilde{x}^\bullet$ be the last vertex of $(\cD')\noend$. By \cref{obs:priv-neighbours} applied to $\cD',\cC'$ we know that $\widetilde{x}^\bullet$ is adjacent to $\widetilde{x}$ and non-adjacent to $\widetilde{y}$. Moreover, since all vertices of $\cD'$ are in $\widetilde{U} \cup \widetilde{S}$, we note that $\widetilde{x}^\bullet$ is adjacent to all vertices from the appropriate bipartition class of $S$.

Now we will separately consider two subcases, depending whether the call for $H[U]$ resulted in the statement (\ref{two-walks-near-statement1}.\ref{two-walks-near-additional-a}.) or (\ref{two-walks-near-statement1}.\ref{two-walks-near-additional-b}.)

\paragraph*{Case 3~a): the call for $H[U]$ resulted in the statement (\ref{two-walks-near-statement1}.\ref{two-walks-near-additional-a}.).} This means that $\widetilde{x} \in \{v,v'\}$ and $\{\cC,\cD\}^{(i)} \not\subseteq S$ for every $i$.

If $|\cP|=1$, i.e., $\cP = \widetilde{y},w$, then we set $x = \widetilde{y}$, $y = s_w$, and define: 
\begin{equation}\label{eq:case3-three}
\begin{alignedat}{6}
\mathcal{A}&=\mathcal{D'} &&\circ \widetilde{x},\widetilde{x}^\bullet,s_w &\quad\quad&  \mathcal{A'}&&=\mathcal{D}&&\circ \widetilde{x},\widetilde{x}^\bullet,s_w\\
\mathcal{B}&=\mathcal{C'} &&\circ \widetilde{y},w,\widetilde{y} &\quad\quad&  \mathcal{B'}&&=\mathcal{C}&&\circ \widetilde{y},w,\widetilde{y}.
\end{alignedat}
\end{equation}
Note that $s_w \in S$ is adjacent to $\widetilde{x}^\bullet \in \widetilde{U} \cup \widetilde{S}$, as in this case $\widetilde{y},\widetilde{x},s_w$ are in one bipartition class, while $\widetilde{x}^\bullet$ is in the other one, and $\widetilde{U} \cup \widetilde{S}$ is bipartite-complete to $S$. Observe that the walks satisfy the condition in the statement (\ref{two-walks-near-statement2}.\ref{two-walks-near-additional-b}).

So let us assume that $|\mathcal{P}|>1$, i.e., $w^\bullet \neq \widetilde{y}$. Define $\widetilde{x}'$ in a way that $\{\widetilde{x},\widetilde{x}'\} = \{v,v'\}$. We define walks
\begin{alignat*}{5}
&& \mathcal{R'}=~&& \widetilde{x},\widetilde{x}^\bullet,\widetilde{x} &~\circ~&& \widetilde{x}, \widetilde{x}',\widetilde{x},\ldots,\widetilde{x}'' \\
&& \mathcal{R}=~&& \widetilde{x},\widetilde{x}',\widetilde{x} &~\circ~&& \widetilde{x}, \widetilde{x}',\widetilde{x},\ldots,\widetilde{x}'',
\end{alignat*}
such that $|\cP|=|\cR|=|\cR'|$, where $\widetilde{x}''$ is either $\widetilde{x}$ or $\widetilde{x}'$ (i.e., either $v$ or $v'$), depending on the parity of $\cP$. Note that $\widetilde{x}''$ is in the same bipartition class as $w$.
We observe that $\cP$ avoids $\cR'$: recall that $\widetilde{y}$ is non-adjacent to $\widetilde{x}^\bullet$, and all other vertices of $\cP$ are in $U - (\widetilde{S} \cup \widetilde{U})$, and thus they are non-adjacent to $\{v,v'\}=\{\widetilde{x},\widetilde{x}'\}$.
Moreover, the latter implies that $\cR$ avoids $\cP$. We set $x=w^\bullet$, $y=s_w$ and:

\begin{equation}\label{eq:case3-two}
\begin{alignedat}{6}
\mathcal{A}&=\mathcal{D'} &&\circ \mathcal{R} \circ \widetilde{x}'',s_w & \quad\quad &  \mathcal{A'}=&&\mathcal{D} &&\circ \mathcal{R'} &&\circ \widetilde{x}'',s_w\\
\mathcal{B}&=\mathcal{C'} &&\circ \mathcal{P} \circ w,w^\bullet & \quad\quad  & \mathcal{B'}=&&\mathcal{C} &&\circ \mathcal{P} &&\circ w,w^\bullet.
\end{alignedat}
\end{equation}
Recall that $\widetilde{x}'' \in \{v,v'\} \subseteq \widetilde{U}$ and as $S$ is bipartite-complete to $\widetilde{U}$, we have that $\widetilde{x}''$ is adjacent to $s_w \in S$. Moreover, we observe that $\cA$ avoids $\cB$ and $\cB'$ avoids $\cA'$, as $w^\bullet$ is in $U \setminus (\widetilde{S} \cup \widetilde{U})$ and thus is non-adjacent to $\widetilde{x}'' \in \widetilde{U}$.
It is straightforward to verify that the constructed walks satisfy the statement (\ref{two-walks-near-statement2}.\ref{two-walks-near-additional-b}).

\paragraph*{Case 3~b): the call for $H[U]$ resulted in the statement (\ref{two-walks-near-statement1}.\ref{two-walks-near-additional-b}.).} 

If $\widetilde{x} \notin \{v,v'\}$, then $\mathcal{D,D'}$ are entirely contained in $\widetilde{U}$, in particular $\widetilde{x},\widetilde{x}^\bullet \in \widetilde{U}$. We define $\mathcal{R}=\widetilde{x},\widetilde{x}^\bullet,\widetilde{x},\ldots,\widetilde{x}''$ so that $|\mathcal{R}|=|\mathcal{P}|$, where $\widetilde{x}''$ is either $\widetilde{x}$ or $\widetilde{x}^\bullet$, depending on the parity of $|\cP|$ . Note that $\cP$ and $\cR$ avoid each other, as $\widetilde{y}$ is non-adjacent to $\widetilde{x}^\bullet$, and no vertex from $\cP$, except for $\widetilde{y}$, is adjacent to any vertex of $\widetilde{U}$.

We set $x=w^\bullet$, $y=s_w$ and:
\begin{equation}\label{eq:case3-one}
\begin{alignedat}{6}
\mathcal{A}&=\mathcal{D'}&& \circ \mathcal{R} &&\circ \widetilde{x}'', s_w, && \quad \mathcal{A'}=&&\mathcal{D} &&\circ \mathcal{R} \circ \widetilde{x}'',s_w, \\
\mathcal{B}&=\mathcal{C'}&& \circ \mathcal{P} &&\circ w, w^\bullet, && \quad \mathcal{B'}=&&\mathcal{C} &&\circ \mathcal{P} \circ w,w^\bullet.
\end{alignedat}
\end{equation}
Similarly to the case of walks constructed in \eqref{eq:case3-two}, we can observe that $\cA$ avoids $\cB$ and $\cB'$ avoids $\cA'$.
Note that this works also for the case $w^\bullet=\widetilde{y}$, that is, if $|\cP|=1$. Then the walks from \cref{eq:case3-one} are as follows:

\begin{equation*}
\begin{alignedat}{6}
\mathcal{A}&=\mathcal{D'}&& \circ \widetilde{x},\widetilde{x}^\bullet,s_w, && \quad \mathcal{A'}=&&\mathcal{D} \circ \widetilde{x},\widetilde{x}^\bullet,s_w,\\
\mathcal{B}&=\mathcal{C'}&& \circ \widetilde{y},w,\widetilde{y}, && \quad \mathcal{B'}=&&\mathcal{C} \circ \widetilde{y},w,\widetilde{y},
\end{alignedat}
\end{equation*}

Finally, observe that in \cref{eq:case3-one} we used only vertices from $U$ (except $s_w$) and thus walks $\cB,\cB'$ are entirely contained in $U$, and for every $i$ it holds that $\{\cA,\cB\}^{(i)} \not\subseteq S$. Thus we obtain the statement (\ref{two-walks-near-statement2}.\ref{two-walks-near-additional-b}). of the lemma.
\end{proof}

\begin{lemma}\label{lem:two-walks-incomp-set}
Let $H$ be an undecomposable, bipartite graph with bipartition classes $X,Y$ and let $s,v \in X$ be incomparable vertices. Let  $T \subseteq X$ be a nonempty set of some vertices reachable from at least one of $s,v$, and incomparable with both $s,v$. Then there exist $t \in T$ and two pairs of walks:
\begin{compactenum}
\item $\mathcal{A,A'}: s \to t$ and $\mathcal{B,B'}: v \to v$, or \label{two-walks-incomp-statement1}
\item $\mathcal{A,A'}: s \to s$ and $\mathcal{B,B'}: v \to t$, \label{two-walks-incomp-statement2}
\end{compactenum} 
such that $\mathcal{A}$ avoids $\mathcal{B}$ and $\mathcal{B'}$ avoids $\mathcal{A'}$. 
\end{lemma}

\begin{proof}
Since $s$ and $v$ are incomparable, there are $v' \in N(v) \setminus N(s)$ and $s' \in N(s) \setminus N(v)$.

\paragraph*{Case 1: there exists $t \in T$ and a vertex $t_v' \in N(t) \setminus N(v)$, which is adjacent to $s$.}
Let $v'_t \in N(v) - N(t)$, it exists since $v$ and $t$ are incomparable. Then we can set: $\cA = \cA'  =s,t'_v,t, \cB=v,v',v$ and $\mathcal{B'}=v,v'_t,v$. It is straightforward to verify that $\mathcal{A}$ avoids $\mathcal{B}$ and $\mathcal{B'}$ avoids $\mathcal{A'}$.  
The case that there is $t \in T$ and $t_s' \in N(t) \setminus N(s)$, such that $t_s'v \in E(H)$, is symmetric.

\paragraph*{Case 2: for every $t \in T$, every neighbor of $t$ is either adjacent to both $s,v$ or non-adjacent to both.}
Observe that this implies that $T \cap N(s',v') = \emptyset$: otherwise, if there is any $t \in T$ adjacent to one of $s',v'$, say $s'$, then $s'$ must be adjacent to both $s,v$, which contradicts its definition.
Let $S$ be a minimal $\{s,v,s',v'\}$-$T$-separator contained in $N(s,v,s',v')$ and let $U$ be the set of vertices reachable from $\{s,s',v,v'\}$ in $H-S$.
Observe that the graph $H$, vertices $s,s',v,v'$, and the set $S$ satisfy the assumptions of \cref{lem:two-walks-near}. So, by \cref{lem:two-walks-near}, there are $y \in S$ and $x \in U$, and walks $\cC,\cC': s \to y, \cD,\cD': v \to x$ (statement (\ref{two-walks-near-statement1}.)) or $\mathcal{C,C'}: v \to y, \mathcal{D,D'}: s \to x$ (statement (\ref{two-walks-near-statement2}.)), such that $\mathcal{C}$ avoids $\mathcal{D}$ and $\mathcal{D'}$ avoids $\mathcal{C'}$.

Suppose that calling \cref{lem:two-walks-near} resulted in statement (\ref{two-walks-near-statement1}.), i.e.,  the obtained walks are $\mathcal{C,C'}: s \to y$ and $\mathcal{D,D'}: v \to x$ (the other case is symmetric). Let $x^\bullet$ be the last vertex on $(\mathcal{D'})\noend$, note that it is adjacent to $x$ and, by \cref{obs:priv-neighbours}, non-adjacent to $y$.
By minimality of $S$ there is $t \in T$ and a $y$-$t$-path $\mathcal{P}$ of length at least 1, whose every vertex, except for $y$, is in $V(H) \setminus (S \cup U)$. This implies that $y$ is the only vertex of $\cP$ with a neighbor in $U$.
Let us consider two subcases.

\paragraph*{Case 2a: calling \cref{lem:two-walks-near} resulted in statement (\ref{two-walks-near-statement1}.\ref{two-walks-near-additional-a}).}
This means that $x \in \{v,v'\}$. Let us define $x'$, such that $\{x,x'\}=\{v,v'\}$.
If $|\cP|>1$, we define $\mathcal{R}=x,x^\bullet,x,x',x,\ldots,v$ and $\mathcal{R'}=x,x',\ldots,v$, such that $|\cR|= |\cR'|=|\cP|$. Recall that $v$ and $t$ are in the same bipartition class, so the last vertex of both walks is indeed $v$.  Observe that $\cP$ avoids $\cR$ and $\cR'$ avoids $\cP$, as $x^\bullet$ is non-adjacent to $y$, and every vertex of $\cP$, except $y$, is non-adjacent to every vertex of $U$, so, in particular, to $\{x,x'\}=\{v,v'\}$.
We set:
\begin{alignat*}{4}
\mathcal{A}= &\mathcal{C} \circ \mathcal{P}, & \quad \mathcal{A'} &=\mathcal{C'} \circ \mathcal{P},\\
\mathcal{B}= & \mathcal{D} \circ \mathcal{R}, & \quad \mathcal{B'} &=\mathcal{D'} \circ \mathcal{R'},
\end{alignat*}
which are walks as in the statement (\ref{two-walks-incomp-statement1}.).
And if $|\cP|=1$, i.e. $\cP=y,t$, then $x=v'$ since $x,y$ are in the same bipartition class. We can define $\cR=v',x^\bullet,v',v$, $\cR'=v',v,v',v$ and $\cP'=y,t,t',t$, where $t'$ is a vertex in $N(t)-N(s,v)$ (note that it exists, since $t$ is incomparable with $s,v$ and every vertex in $N(t)\cap N(s,v)$ is adjacent to both $s,v$ by the definition of Case 2). Note that $\cR'$ avoids $\cP'$ and $\cP'$ avoids $\cR$.
We set:
\begin{alignat*}{4}
\mathcal{A}= &\mathcal{C} \circ \mathcal{P'}, & \quad \mathcal{A'} &=\mathcal{C'} \circ \mathcal{P'},\\
\mathcal{B}= & \mathcal{D} \circ \mathcal{R}, & \quad \mathcal{B'} &=\mathcal{D'} \circ \mathcal{R'},
\end{alignat*}
and obtain walks as in the statement (\ref{two-walks-incomp-statement1}.).

\paragraph*{Case 2b: calling \cref{lem:two-walks-near} resulted in statement (\ref{two-walks-near-statement1}.\ref{two-walks-near-additional-b}).}
This means that walks $\cD,\cD'$ are entirely contained in $U$.
Let $t'$ be a neighbor of $t$ such that $t' \notin S_Y$, again it exists in this case (recall that $S_Y \subseteq N(s,v)$). We set:
\begin{alignat*}{4}
\mathcal{A}=&\mathcal{C} \circ \mathcal{P}^*, & \quad \mathcal{A'}=&\mathcal{C'} \circ \mathcal{P}^*,\\
\mathcal{B}=&\mathcal{D} \circ \mathcal{D}^*, & \quad \mathcal{B'}=&\mathcal{D'} \circ \mathcal{D}^*,
\end{alignat*}
where $\mathcal{P}^*:= \mathcal{P} \circ t,t',\ldots,t$ and $\mathcal{D}^*:=\overline{\mathcal{D'}} \circ v,v',\ldots,v$ are defined in a way that $|\mathcal{P}^*|=|\mathcal{D}^*|:=\max(|\mathcal{P}|, |\mathcal{D'}|)$. Basically these walks play the same role as $\cP$ and $\overline{\cD}'$, but extra padding added to one of them ensures they have the same length.
Observe that $\cD^*$ is entirely contained in $U$ and thus it avoids $\cP^*$ whose only vertex adjacent to $U$ is the first one. Moreover $\cP^*$ avoids $\cD^*$, since the first vertex of $\cP$ is $y$ and it is non-adjacent to $x^\bullet$, which is the second vertex on $\cD^*$. Thus $\cA$ avoids $\cB$ and $\cB'$ avoids $\cA'$ and we have obtained statement (\ref{two-walks-incomp-statement1}.).
\end{proof}

\begin{lemma}\label{lem:walks-s-v}
Let $H$ be a bipartite, undecomposable graph with bipartition classes $X,Y$, let $s,v \in X$ be incomparable vertices and let $t \in X$ be a vertex reachable from at least one of $s,v$. Then there exist a vertex $q \in X$ and two pairs of walks: \begin{compactenum}
\item $\mathcal{P,P'}: s \to t$ and $\mathcal{Q,Q'}: v \to q$, or
\item $\mathcal{P,P'}: s \to q$ and $\mathcal{Q,Q'}: v \to t$,
\end{compactenum}
such that $\mathcal{P}$ avoids $\mathcal{Q}$ and $\mathcal{Q'}$ avoids $\mathcal{P'}$. Moreover, if $t$ is incomparable with at least one of $s,v$ then $q=v$ in the first case and $q=s$ in the other.

Finally, given a bipartite graph $H$, in time polynomial in $|H|$ we can either find the desired walks or a decomposition of $H$.
\end{lemma}

\begin{proof}
Observe that if $t$ is incomparable with both $s,v$, we can obtain the desired walks by applying \cref{lem:two-walks-incomp-set} for $s,v$, and $T=\{t\}$.

Now let us assume that $t$ is incomparable with exactly one of $s,v$, say $v$ (the other case is symmetric) and let $t_v'$ be a vertex in $N(t) \setminus N(v)$ and let $v'_t$ be a vertex in $N(v) \setminus N(t)$. Since $s$ and $v$ are incomparable, there are $v_s' \in N(v) \setminus N(s)$ and $s_v' \in N(s) \setminus N(v)$.
Since $t$ is comparable with $s$, we either have $N(t) \subseteq N(s)$ or $N(s) \subseteq N(t)$.
If $N(t) \subseteq N(s)$, we observe that $t_v'$ must be adjacent to $s$, and $v_s'$ must be non-adjacent to $t$ since $v_s' \notin N(s)$.
Then we set $q=v$, $\cP = \cP' = s,t_v',t$, and $\mathcal{Q=Q'}=v,v_s',v$.
On the other hand, if $N(s) \subseteq N(t)$, we observe that $s_v'$ must be adjacent to $t$ and $v_t'$ must be non-adjacent to $s$. Then we set $q=v$, $\cP = \cP' = s,s_v',t$, and $\mathcal{Q=Q'}=v,v_t',v$. 

So from now we can assume that $t$ is comparable with both $s,v$. Observe that it implies that either $N(s,v) \subseteq N(t)$, or $N(t) \subseteq N(s) \cap N(v)$. Note that the other cases are not possible: if we have, say, $N(v) \subseteq N(t)$ and $N(t) \subseteq N(s)$, then $N(v) \subseteq N(s)$, a contradiction.
Observe that if there is $s' \in N(s) \setminus N(v)$ and $v' \in N(v) \setminus N(s)$, such that $t$ is adjacent to at least one of $s',v'$, then we can set $S :=\{t\} \subseteq N(s,s',v,v')$ and, by the assumption of $H$, it has no decomposition. Thus we can call \cref{lem:two-walks-near} for vertices $s,s',v,v'$ and the set $S=\{t\}$. Since $t$ is the only vertex of $S$, we obtain a vertex $q$, reachable from $\{s,v\}$ in $H-\{t\}$, and  walks $\mathcal{P,P'}: s \to t$ and $\mathcal{Q,Q'}: v \to q$, or $\mathcal{Q,Q'}: v \to t$ and $\mathcal{P,P'}: s \to q $, such that $\cP$ avoids $\cQ$ and $\cQ'$ avoids $\cP'$.

So it only remains to consider the case in which $t$ is non-adjacent to every vertex in $\bigl(N(s) \setminus N(v)\bigr) \cup \bigl(N(v) \setminus N(s)\bigr)$.  Choose arbitrary $s' \in N(s) \setminus N(v)$ and $v' \in N(v) \setminus N(s)$. Observe that by the case we are considering, we have $N(t) \subseteq N(s) \cap N(v)$, which implies that $N(t) \subseteq N(s,s',v,v')$.
We call \cref{lem:two-walks-near} for $H$, vertices $s,s',v,v'$, and the set $S=N(t)$. We obtain vertices $y \in N(t)$ and $x \in U$, where $U$ is the set of vertices reachable from $\{ s,v \}$  in $H-S$, and walks $\mathcal{A,A'}: s \to y, \mathcal{B,B'}: v \to x$ or $\mathcal{A,A'}: v \to y, \mathcal{B,B'}: s \to x$ of length at least one, such that $\cA$ avoids $\cB$ and $\cB'$ avoids $\cA'$. Let $q$ be the last vertex  on $(\mathcal{B'})\noend$ -- it is adjacent to $x$ and non-adjacent to $y$. We also have that $x$ is non-adjacent to $t$, since all neighbors of $t$ are in $S$. If we obtain walks $\mathcal{A,A'}: s \to y, \mathcal{B,B'}: v \to x$  then we set:
\begin{align*}
\mathcal{P}=\mathcal{A} \circ y,t, & \quad \mathcal{P'}=\mathcal{A'} \circ y,t, \\
\mathcal{Q}=\mathcal{B} \circ x,q, & \quad \mathcal{Q'}=\mathcal{B'} \circ x,q,
\end{align*}
and if we obtain walks $\mathcal{A,A'}: v \to y, \mathcal{B,B'}: s \to x$, then we set:
\begin{align*}
\mathcal{P}=\mathcal{B'} \circ x,q, & \quad \mathcal{P'}=\mathcal{B} \circ x,q, \\
\mathcal{Q}=\mathcal{A'} \circ y,t, & \quad \mathcal{Q'}=\mathcal{A} \circ y,t.
\end{align*}
It is straightforward to verify that in both cases $\mathcal{P}$ avoids $\mathcal{Q}$ and $\mathcal{Q'}$ avoids $\mathcal{P'}$.
The proof, as well as the proofs of \cref{lem:two-walks-near} and \cref{lem:two-walks-incomp-set}, are clearly constructive, so the claim about the polynomial-time algorithm follows easily.
\end{proof}
\subsubsection{Constructing walks from incomparable vertices to an obstruction.}
Now we will introduce several lemmas about how our walks reach the obstruction. 

\begin{lemma}\label{lem:two-walks-to-obs}
Let $H$ be an undecomposable, bipartite graph with an obstruction $\Ob$. Let $(\alpha, \beta) \in C(\Ob)$, and let $s$ be a vertex reachable from $\Ob$, which belongs to the same bipartition class as $\beta$ but is incomparable with it. Then there exist walks $\cA,\cA':s \to \alpha$ and $\cB,\cB':\beta \to \beta$, such that $\cA$ avoids $\cB$ and $\cB'$ avoids $\cA'$.
\end{lemma}
\begin{proof}
We use \cref{lem:walks-s-v} for $s=\alpha$, $v=\beta$, and $t=s$ to obtain walks $\cP, \cP': \alpha \to s$ and $\cQ,\cQ':\beta \to \beta$ or $\cP, \cP': \alpha \to \alpha$ and $\cQ,\cQ':\beta \to s$, such that $\cP$ avoids $\cQ$ and $\cQ'$ avoids $\cP'$. In the first case, it is enough to take $\cA=\overline{\cP'}$, $\cB=\overline{\cQ'}$, $\cA'=\overline{\cP}$, $\cB'=\overline{\cQ}$, and we are done.
In the second case, we use \cref{lem:walks-s-v} again, now for $s=s$, $v=\beta$ and $t=\alpha$ and obtain walks $\cR, \cR': s \to \alpha$ and $\cS,\cS':\beta \to \beta$ or $\cR, \cR': s \to s$ and $\cS,\cS':\beta \to \alpha$, such that $\cR$ avoids $\cS$ and $\cS'$ avoids $\cR'$. 
In the first case, we are done, as these are the walks we are looking for. In the second case we set
\begin{equation*}
\begin{alignedat}{7}
\cA =& \cR &&\circ \overline{\cQ} &&\circ \cY', \quad &&\cA' = \cR' &&\circ \overline{\cQ'} &&\circ \cY, \\
\cB =& \cS &&\circ \overline{\cP} &&\circ \cX', \quad &&\cB' = \cS' &&\circ \overline{\cP'} &&\circ \cX,
\end{alignedat}
\end{equation*} 
where $\cX,\cX' : \alpha \to \beta$ and $\cY,\cY' : \beta \to \alpha$ are obtained by \cref{obs:walks-between-corners}.
\end{proof} 

In the next proof some constructions will depend on the bipartition class of a particular vertex. To avoid considering cases separately, for two non-necessarily distinct vertices $u$ and $u'$, such that $u' \in N[u]$, we introduce a walk
\begin{equation*}
\cE[u,u']:=\begin{cases}
u,u' & $ if $uu' \in E(H), \\
u & $ if $u=u'.
\end{cases}
\end{equation*}

Note that if sets $\{u,u'\}$ and $\{v,v'\}$ are non-adjacent (again, it might be that $u=u'$ or $v=v'$), then the walks $\cE[u,u']$ and $\cE[v,v']$ are clearly non-adjacent.

\begin{lemma}\label{lem:walks-to-corners-near}
Let $H$ be a bipartite graph with an obstruction $\Ob$. Let $(\alpha, \beta) \in C(\Ob)$ and let $s \in V(H)$ be a vertex reachable from $\Ob$, which belongs to the same bipartition class as $\alpha,\beta$. Assume that $N(s) \setminus N(\beta) \neq \emptyset$ and $\dist(N[s] \setminus N(\beta),\Ob)\leq~1$. Then one of the following exists:
\begin{compactenum}
\item walks $\cP: \alpha \to \alpha, \cQ: s \to \alpha$ and $\cR: \beta \to \beta$, such that $\cP, \cQ$ avoid $\cR$, \label{lem17-statement1}
\item walks $\cP: \alpha \to \alpha, \cQ: s \to \beta$ and $\cR: \beta \to \beta$, such that $\cP$ avoids $\cQ, \cR$.
\label{lem17-statement2}
\end{compactenum}
\end{lemma} 
\begin{proof}
Let $\alpha',\beta'$ be the pair of vertices, such that $C(\Ob)=\{(\alpha, \beta) ,(\alpha',\beta')\}$.

By \cref{def:corners} recall that edges $\alpha\ovalp$ and $\beta\ovbet$ are independent. First, we consider some simple special cases separately. If $s\ovbet\in E(H)$, we obtain the statement (\ref{lem17-statement2}.) by setting
\begin{alignat*}{4}
\cP=&\alpha,\ovalp \alpha, \\
\cQ=&s,\ovbet,\beta, \\
\cR=&\beta, \ovbet, \beta.
\end{alignat*}
If now $s\ovbet \not \in E(H)$ and there exists $s' \in N(s) \setminus N(\beta)$ such that $\alpha s' \in E(H)$ (in particular when $s$ is adjacent to $\ovalp$), then we obtain the statement (\ref{lem17-statement1}.) by setting
\begin{alignat*}{2}
\cP=& \alpha,\ovalp,\alpha,\\
\cQ=&s,s',\alpha,\\
\cR=&\beta,\ovbet,\beta.
\end{alignat*}
So assume that the cases above do not apply. Let us fix any $s' \in N(s) \setminus N(\beta)$ for which $\dist(\{s,s'\},\Ob)\leq 1$. Note that the edges $\alpha\ovalp, \beta\ovbet$, and $ss'$ are independent.

\paragraph{Case 1: there exists $p \in N[s,s'] \cap V(\Ob)$ such that $p \not\in N(\alpha,\alpha')$.} 
Clearly $p \not\in\{\alpha,\ovalp,\beta,\ovbet\}$, as otherwise edges $\alpha\ovalp, \beta\ovbet, ss'$ would not be independent. Note that this implies that this case cannot occur if $\Ob$ is isomorphic to $C_6$, see \cref{fig:cycle-corners} (left).
Let $\overline{s}$ be an element of $\{s,s'\}$ which is a neighbor of $p$. Recall that since $p,\beta \not\in N[\alpha,\ovalp]$, by \cref{obs:asteroid-subwalks}, there exists a $p$-$\beta$-walk $\cW[p,\beta]$, which is non-adjacent to $\{\alpha,\ovalp\}$. We define:
\begin{align*}
\cP&:= \alpha,\ovalp,\ldots,\alpha, \\
\cQ&:= \cE[s,\overline{s}] \circ \overline{s},p \circ \cW[p,\beta], \\
\cR&:= \beta,\ovbet,\ldots,\beta,
\end{align*}
in a way that $|\cP|=|\cR|= |\cQ|$. Clearly $\cP$ avoids $\cR$. Moreover, $\cP$ avoids $\cQ$ since $s, s', p \not \in N[\alpha,\ovalp]$ and $\cW[p,\beta]$ is non-adjacent to $\{\alpha,\ovalp\}$ . Thus we obtain the statement (\ref{lem17-statement2}.).

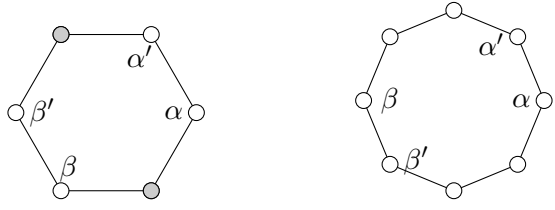
\begin{figure}[t]
\centering{\begin{tikzpicture}[every node/.style={draw,circle,fill=white,inner sep=0pt,minimum size=6pt},every loop/.style={}]
\def\n{6}
\foreach \i in {1,...,\n}
	\draw (360/\n*\i-360/\n:1.2) -- (360/\n*\i:1.2);
\foreach \i in {1,...,\n}
	\node at (360/\n*\i:1.2) {};
\node[fill=none,draw=none] at (0.9,0) {$\alpha$};
\node[fill=none,draw=none] at (-0.85,0) {$\beta'$};
\node[fill=none,draw=none] at (0.45,0.75) {$\alpha'$};
\node[fill=none,draw=none] at (-0.5,-0.75) {$\beta$};
\node[fill=black!20] at (360/\n*2:1.2) {};
\node[fill=black!20] at (360/\n*5:1.2) {};
\end{tikzpicture}\hskip 2cm
\begin{tikzpicture}[every node/.style={draw,circle,fill=white,inner sep=0pt,minimum size=6pt},every loop/.style={}]
\def\n{8}
\foreach \i in {1,...,\n}
	\draw (360/\n*\i-360/\n:1.2) -- (360/\n*\i:1.2);
\foreach \i in {1,...,\n}
	\node at (360/\n*\i:1.2) {};
	\node[fill=none,draw=none] at (0.9,0) {$\alpha$};
\node[fill=none,draw=none] at (-0.85,0) {$\beta$};
\node[fill=none,draw=none] at (0.5,0.75) {$\alpha'$};
\node[fill=none,draw=none] at (-0.5,-0.8) {$\beta'$};
\end{tikzpicture}}
\caption{The position of elements of $C(\Ob)$ when $\Ob$ is isomorphic to $C_6$ (left) of $C_8$ (right). Gray vertices indicate the possible position of vertices $p$ and $r$ in the Case 2 of the proof of \cref{lem:walks-to-corners-near}.} \label{fig:cycle-corners}
\end{figure}

\paragraph{Case 2: $N[s,s'] \cap V(\Ob) \subseteq N(\alpha,\ovalp)$.}
It implies that $\dist(\{s,s'\},\Ob)= 1$, as if $\{s,s'\} \cap V(\Ob) \neq \emptyset$, then $\{s,s'\}$ and $\{\alpha,\ovalp\}$ would be adjacent.
So there exists $p \in N(s,s') \cap V(\Ob) \cap N(\alpha,\ovalp)$. Again, let $\overline{s}$ be an element of $\{s,s'\}$ which is a neighbor of $p$ and let $\{\overline{\alpha},\overline{\alpha}'\}:=\{\alpha,\ovalp\}$ such that $\overline{\alpha}p \in E(H)$. 

Note that if $p \not\in N(\beta,\ovbet)$, we can set
\begin{align*}
\cP&:=\alpha,\ovalp,\ldots,\alpha, \\
\cQ&:=\cE[s,\overline{s}] \circ \overline{s},p,\overline{\alpha} \circ \cE[\overline{\alpha}, \alpha] \\
\cR&:=\beta,\ovbet,\ldots,\beta.
\end{align*}
in a way that $|\cP|=|\cR|= |\cQ|$. Then $\cP, \cQ$ avoid $\cR$ and we get the statement (\ref{lem17-statement2}.).
Observe that if $\Ob$ is isomorphic to $C_8$, then the above case applies; see \cref{fig:cycle-corners} (right).

So we can assume that $p \in N(s,s') \cap N(\alpha,\ovalp) \cap N(\beta,\ovbet)$.  Let $\{\overline{\beta},\overline{\beta}'\}:=\{\beta, \ovbet\}$, such that $\overline{\beta}p \in E(H)$. If $\Ob$ is isomorphic to $C_6$, then let $r$ be the other vertex of $\Ob$ which belongs to $N(\alpha,\ovalp) \cap N(\beta,\ovbet)$ ($r$ is uniquely determined, see \cref{fig:cycle-corners} (left)).
As $pr \not\in E(H)$, we define the walks that satisfy the statement (\ref{lem17-statement1}.):
\begin{equation*}
\begin{alignedat}{10}
\cP&:=\cE[\alpha,\overline{\alpha}] &&\circ \overline{\alpha},p,\overline{\alpha}, \overline{\alpha}' &&\circ \cE[\overline{\alpha}',\alpha] \\ 
\cQ&:=\cE[s,\overline{s}] &&\circ \overline{s},p,\overline{\alpha}, \overline{\alpha}' &&\circ \cE[\overline{\alpha}',\alpha], \\
\cR&:=\cE[\beta,\overline{\beta}] &&\circ \overline{\beta},\overline{\beta}',r,\overline{\beta}' &&\circ \cE[\overline{\beta}',\beta].
\end{alignedat}
\end{equation*}

So finally consider the case that $\Ob$ is an asteroidal subgraph. We use the notation introduced in \cref{def:asteroid}.
Then $p \in \cW_{0,i}$ for some $i \in \{1,2k\}$, as all other vertices of $\Ob$ are non-adjacent to $\alpha$. Let $\{\kappa,\kappa'\}  =\{u_{k+1},v_{k+1}\}$ if $i=1$ or $\{\kappa,\kappa'\}  =\{u_{k},v_{k}\}$ if $i=2k$, where $\kappa'$ and $p$ are in the same bipartition class.
Note that by the definition of a special edge asteroid we know that the walk $\cW_{0,i}$, and in particular $p$, is non-adjacent to $\{\kappa,\kappa'\}$. Recall that the walk $\cW[\beta,\kappa]$ given by \cref{obs:asteroid-subwalks} is non-adjacent to $\{\alpha,\alpha'\}=\{u_0,v_0\}$. So, since we are in Case 2, we observe that $s$ and $s'$ are non-adjacent to $\{\kappa,\kappa'\}$ and to $\cW[\beta,\kappa]$, as otherwise we would have chosen $p \in \{\kappa,\kappa'\}\cup \cW[\beta,\kappa]$, such that $p \in N[s,s'] \cap V(\Ob)$ and $p \notin N(\alpha,\alpha')$, ending up in Case 1. We set 
\begin{alignat*}{10}
\cP &:=\alpha,\ovalp,\ldots,\alpha \circ \cE[\alpha,\overline{\alpha}] &&~\circ~&& \overline{\alpha}, \overline{\alpha}', \overline{\alpha} &&~\circ~&& \cE[\overline{\alpha},\alpha] \circ \alpha , \ovalp, \ldots,\alpha, \\
\cQ &:=s,s', \ldots,s \ \circ \cE[s,\overline{s}] &&~\circ~&& \ \overline{s},p, \overline{\alpha} &&~\circ~&& \cE[\overline{\alpha},\alpha] \circ  \alpha, \ovalp, \ldots, \alpha, \\
\cR &:=\cW[\beta,\kappa] &&~\circ~&&   \kappa, \kappa',\kappa &&~\circ~&& \cW[\kappa,\beta],
\end{alignat*}
so that the lengths of the subwalks in each aligned column are equal.
Clearly, $\cR$ is non-adjacent with both $\cP$ and $\cQ$, so we obtain the statement \eqref{lem17-statement1}.
\end{proof}

\begin{lemma}\label{lem:walks-to-corners-far}
Let $H$ be a connected, undecomposable, bipartite graph with an obstruction $\Ob$, let $(\alpha, \beta) \in C(\Ob)$ and let $s \in V(H)$ be a vertex reachable from $\Ob$, which belongs to the same bipartition class as $\beta$, but incomparable with it. Then at least one of the following~exists:
\begin{compactenum}
\item walks $\cP: \alpha \to \alpha, \cQ: s \to \alpha$ and $\cR: \beta \to \beta$, such that $\cP, \cQ$ avoid $\cR$,
\item walks $\cP: \alpha \to \alpha, \cQ: s \to \beta$ and $\cR: \beta \to \beta$, such that $\cP$ avoids $\cQ, \cR$.
\end{compactenum}  
\end{lemma} 
\begin{proof}
Let $\alpha',\beta'$ be the pair of vertices, such that $C(\Ob)=\{(\alpha, \beta) ,(\alpha',\beta')\}$.
Let $X,Y$ be the bipartition classes of $H$, such that $\alpha,\beta \in X$ and $\ovalp,\ovbet \in Y$.
If $\dist(N[s] \setminus N(\beta),\Ob) \leq 1$, we are done by \cref{lem:walks-to-corners-near}.
Thus let us assume that $\dist(N[s] \setminus N(\beta), \Ob) \geq 2$. Fix any $s' \in N(s) \setminus N(\beta)$ and observe that edges $\alpha\ovalp, \beta\ovbet$ and $ss'$ are independent.
Let $P$ be the minimal $s$-$\Ob$-separator contained in $N(\Ob)$. Let $R$ be the set of vertices reachable from $s$ in $H \setminus P$ and let $Q$ be the set of vertices reachable from $\alpha, \beta$ in $H \setminus P$. Note that $V(\Ob) \subseteq Q$ and $s' \in R$, as it does not belong to $N(\Ob)$.

\paragraph{Case 1: there exists $p \in P$ such that $p \not\in N(\alpha,\beta,\ovalp,\ovbet)$.} Denote by $\cS_p$ the $s$-$p$-walk, such that $\cS_p\noend$ is contained in $R$ (recall that $\cS_p\noend$ denotes the walk $\cS_p$ with the vertex $p$ removed). The walk $\cS_p$ exists, because $P$ is a minimal $s$-$\Ob$-separator. Let $p^\bullet \in R$ be the last vertex of $\cS_p\noend$ (it is possible that $p^\bullet=s$). 

If $p \in P_X$, then clearly $p \in N(\Ob) \setminus N(\beta)$, so $\dist(N[p] \setminus N(\beta),\Ob)\leq 1$. Moreover, $N(p) \setminus N(\beta)$ is non-empty, because $p^\bullet \in N(p) \setminus N(\beta)$.
So by \cref{lem:walks-to-corners-near} we can obtain walks $\cT: \alpha \to \alpha, \cU:p \to \lambda$ for $\lambda \in \{\alpha,\beta\}$, and $\cV: \beta \to \beta$ such that if $\cU$ is a $p$-$\alpha$-walk, then $\cT, \cU$ avoid $\cV$, and otherwise $\cT$ avoids $\cU,\cV$.
Then we define $\cP:=\alpha, \ovalp, \ldots, \alpha \circ  \cT, \cQ:= \cS_p \circ \cU,$ and $\cR:= \beta, \ovbet, \ldots, \beta \circ \cV$ in a way that $|\cP|=|\cQ|=|\cR|$.
Clearly, $\cS_p$ is non-adjacent to $\{\alpha, \ovalp, \beta, \ovbet\}$, because $\cS_p\noend \subseteq R$ and $p \not\in N(\ovalp,\ovbet)$. Thus walks $\cP,\cQ,\cR$ satisfy statement (1.) or (2.), depending on the statement in the call of \cref{lem:walks-to-corners-near}.

Similarly, if $p \in P_Y$, we observe that $\dist(N[p] \setminus N(\ovbet),\Ob)\leq \dist(p,\overline{\Ob})= 1$ and $p^\bullet \in N(p) \setminus N(\ovbet)$. Thus again we can use \cref{lem:walks-to-corners-near}, but now for $(\ovalp,\ovbet) \in C(\Ob)$ instead of $(\alpha,\beta)$.
We obtain walks $\cT: \ovalp \to \ovalp,\cU:p \to \lambda'$ for $\lambda' \in \{\ovalp,\ovbet\}$ and $\cV: \ovbet \to \ovbet$, such that if $\cU$ is a $p$-$\ovalp$-walk, then $\cT, \cU$ avoid $\cV$ and otherwise $\cT$ avoids $\cU,\cV$.
We define
\begin{equation*}
\begin{alignedat}{10}
\cP= \ & \alpha, \ovalp, \ldots, \ovalp &&\circ \cT &&\circ \ovalp,\alpha, \\
\cQ= \ & \cS_p &&\circ \cU &&\circ \lambda', \lambda, \\
\cR= \ & \beta, \ovbet, \ldots, \beta' &&\circ \cV  &&\circ \ovbet,\beta,
\end{alignedat}
\end{equation*}
where $\lambda$ is the vertex in $\{\alpha,\beta\} \cap N(\lambda')$.
Analogously like in the subcase when $p \in P_X$, we can verify that the statement of the lemma holds.

\paragraph{Case 2: for every $p \in P$ it holds that $p \in N(\alpha,\beta,\ovalp,\ovbet)$.} We use \cref{lem:two-walks-near} for $s=\alpha, s'=\alpha' v= \beta, v'=\beta'$ and $S=P$ to obtain vertices $y \in P$ and $x\in Q$ and walks $\cA,\cA': \alpha \to y, \cB,\cB': \beta \to x$ or $\cA,\cA': \beta \to y, \cB,\cB': \alpha \to x$, such that $\cA$ avoids $\cB$, $\cB'$ avoids $\cA'$.
Furthermore, for every $i$ we have $\{\cA,\cB\}^{(i)} \not\subseteq P$ and  $\cA,\cA',\cB,\cB' \subseteq P \cup Q$.

Observe that it is enough to consider the first case, i.e., $\cA,\cA': \alpha \to y, \cB,\cB': \beta \to x$. Indeed, observe that in the second case walks $\cX \circ \cA, \cX' \circ\cA': \alpha \to y$ and $\cY \circ\cB,\cY' \circ\cB': \beta \to x$, where $\cX, \cX', \cY, \cY'$ are given by \cref{obs:walks-between-corners}, satisfy the assumptions of the first case.

Now let $\cS_y$ be the $s$-$y$-walk, such that $\cS_y\noend$ is contained in $R$. Define $\cA^* := \alpha,\alpha',\alpha \circ \cA$, $\cB^* := \beta,\beta',\beta \circ \cB$, and $\cS^* := s,s',s \circ \cS_y$. Observe that $|\cA^*|=|\cB^*|$ and we have $\alpha,\alpha' \in \cA^*$, $\beta,\beta' \in \cB^*$, and $s,s' \in \cS^*$.
Denote the consecutive vertices of these walks by $\cA^* = a_1,\ldots,a_{\ell}$,  $\cB^* = b_1,\ldots,b_{\ell}$, and $\cS^*=s_1,\ldots,s_{m}$. Note that $a_1=\alpha, s_1=s, b_1=\beta$ and $a_{\ell} = s_{m} = y$, $b_{\ell} = x$, and $\cA^*$ avoids $\cB^*$.

Observe that there exist $i \in [\ell]$ and $j \in [m-1]$ such that $a_is_j \in E(H)$ or $b_is_j \in E(H)$ (for example we have $a_{\ell} = s_{m} = y$, so $a_{\ell} s_{m-1} \in E(H)$). Take minimum such $j$ and for that $j$ take minimum $i$. 

Define $\cC_\alpha:=\alpha, \ovalp, \ldots, \alpha$ and $\cC_\beta:=\beta, \ovbet, \ldots, \beta$, such that $|\cC_\alpha|=|\cC_\beta|=\max(j-i+1, 0)$, and $\cC_s:=s, s', \ldots, s$, such that $|\cC_s|:=\max(i-1-j, 0)$.
Note that then three walks $\cC_\alpha \circ a_1, \ldots, a_{i-1}$, $\cC_\beta \circ b_1, \ldots, b_{i-1}$, and $\cC_s \circ s_1, \ldots, s_j$ have the same length.

We claim that only one of the edges $a_is_j$ and $b_is_j$ exists. Indeed, recall that for every $j \in [m-1]$ we have $s_j \in R$. On the other hand, walks $\cA^*$ and $\cB^*$ are contained in $Q \cup P$. So $a_i$ (or $b_i$) can only be adjacent to $s_j$ if $a_i \in P$ ($b_i \in P$). However, by the properties of $\cA,\cB$ we know that at most one of $a_i,b_i$ may be in $P$.

If $a_is_j \in E(H)$ and $b_is_j \not\in E(H)$, we set 
\begin{equation*}
\begin{alignedat}{6}
\cP&=\cC_\alpha \circ a_1, \ldots, a_{i-1} \ & \ \circ & \ a_{i-1}, & \ a_i & \circ a_{i}, \ldots, a_\ell && \circ \overline{\cA'} && \\
\cQ&=\cC_s \circ s_1, \ldots, s_j \ & \ \circ & \ s_j,& \ a_i & \circ a_{i}, \ldots, a_\ell && \circ \overline{\cA'} && \\
\cR&=\cC_\beta \circ b_1, \ldots, b_{i-1} \ & \ \circ & \ b_{i-1},& \ b_i & \circ b_{i}, \ldots, b_\ell && \circ \overline{\cB'} && \\
\end{alignedat}
\end{equation*}
Note that $\cP = \cC_\alpha \circ \cA^* \circ \overline{\cA'}$ and $\cR = \cC_\beta \circ \cB^* \circ \overline{\cB'}$.
Observe that since $s,s' \in \cS^*$, $\beta,\beta' \in \cB$, and by the definition of $j$ and $i$, the subwalk $\cC_s \circ s_1,\ldots,s_j$ of $\cQ$ is non-adjacent to the subwalk $\cC_\beta \circ b_1,\ldots,b_{i-1}$ of $\cR$.

By analogous arguments we can see that if $b_is_j \in E(H)$ and $a_is_j \not\in E(H)$, then for
\begin{equation*}
\begin{alignedat}{6}
\cP&=\cC_\alpha \circ a_1, \ldots, a_{i-1} \ & \ \circ & \ a_{i-1}, & \ a_i & \circ a_{i+1}, \ldots, a_\ell && \circ \overline{\cA'} && \\
\cQ&=\cC_s \circ s_1, \ldots, s_j \ & \ \circ & \ s_j,& \ b_i & \circ b_{i+1}, \ldots, b_\ell && \circ \overline{\cB'} && \\
\cR&=\cC_\beta \circ b_1, \ldots, b_{i-1} \ & \ \circ & \ b_{i-1},& \ b_i & \circ b_{i+1}, \ldots, b_\ell && \circ \overline{\cB'} && \\
\end{alignedat}
\end{equation*} 
we have that $\cP$ avoids $\cQ$ and $\cR$.
\end{proof}

\subsubsection{Proof of \cref{lem:walks-from-uncomp}.}

All lemmas introduced so far will be used to prove \cref{lem:walks-from-uncomp}.

\lemWalksFromIncomp

Before we prove \cref{lem:walks-from-uncomp} in full generality, we consider a special case. We say that an incomparable set $S$ is \emph{strongly incomparable} if for every $v \in S$ there exists $v' \in N(v)$ such that for every $u \in S - \{v\}$ it holds that $uv' \not\in E(H)$. Observe that then the edges $\{vv' \}_{v \in S}$ are independent.
 
\begin{lemma}\label{lem:walks-from-uncomp-special}
\cref{lem:walks-from-uncomp} holds if we additionally assume that $S$ is strongly incomparable.
\end{lemma}
\begin{proof}Let $X,Y$ be the bipartition classes of $H$ such that $S = \{x_1,\ldots,x_k \} \subseteq X$, where $k \geq 2, x_1=a, x_2=b$. To simplify the notation, we will write $\cD_i$ instead of $\cD_{x_i}$.
Let $(\ovalp,\ovbet)$ be the pair, such that $C(\Ob) = \{(\alpha,\beta),(\ovalp,\ovbet)\}$.

First, observe that all constructed walks must have even length, as they start and end in the same bipartition class. Thus the only possibility to have walks satisfying conditions (1)-(3), which are of length less than two, is if $|S|=2$ and $a=\alpha$ and $b=\beta$. To avoid having walks of length 0, in this case we return walks $\cD_a=\alpha,\ovalp,\alpha$ and $\cD_b=\beta,\ovbet,\beta$. So from now on we do not need to worry about the length of the walks.

We prove the lemma by induction on $k$. Consider the base case that $S=\{a,b\}$. 
By \cref{lem:walks-s-v} used for $s=a, v=b$ and $t=\beta$, we obtain a vertex $q$ and walks $\cP,\cP': a \to \beta,\cQ,\cQ': b \to q$ or $\cP,\cP': a \to q,\cQ,\cQ': b \to \beta$, such that $\cP$ avoids $\cQ$ and $\cQ'$ avoids $\cP'$.
Note that in both cases $q$ must be incomparable with $\beta$ (recall \cref{obs:priv-neighbours}), so we can use \cref{lem:two-walks-to-obs} for $s=q$ to obtain $\cA,\cA':q \to \alpha$ and $\cB,\cB': \beta \to \beta$, such that $\cA$ avoids $\cB$ and $\cB'$ avoids $\cA'$.

If $\cP: a \to q$ and $\cQ:b \to \beta$, then clearly we can define $\cD_a = \cP \circ \cA$ and $\cD_b = \cQ \circ \cB$.
On the other hand, if $\cP: a \to \beta,\cQ: b \to q$, we define $\cD_a = \cP \circ \cB' \circ \cY'$ and $\cD_b = \cQ \circ \cA' \circ \cX'$, where $\cX', \cY'$ are given by \cref{obs:walks-between-corners}.

So now assume that $S=\{x_1,\ldots,x_k\}$ for $k \geq 3$ and $x_1=a$, $x_2=b$, and the lemma holds for $k-1$.
Let $\{ \widetilde{\cD_i} \}_{i=1}^{k-1}$ be the set of walks given by the inductive call for the set $S \setminus \{x_k\}$.
The consecutive vertices of $\widetilde{\cD_i}$ are denoted by $d_1^i, \ldots, d_\ell^i$.
Let $A$ be the set of walks $\widetilde{\cD_i}$ terminating at $\alpha$ and let $B$ be the set of walks $\widetilde{\cD_i}$ terminating at $\beta$. Clearly $\widetilde{\cD_a} \in A$ and $\widetilde{\cD_b} \in B$.

As $S$ is strongly incomparable, for every $x_i \in S$ there exists $x'_i \in N(x_i)$, such that the edges in the set $\{x_ix'_i\}_{i \in [k]}$ are independent.

\paragraph{Case 1: There is an edge between $\{x_k,x'_k\}$ and $\bigcup_{i \in [k-1]} \widetilde{\cD_i}$.}
This means that there is some $j \in [\ell]$ and $p \in [k-1]$, such that one of edges $x_kd^p_j$ or $x'_kd^p_j$ exists (note that both edges cannot exist since $H$ is bipartite).
We choose the minimum such $j$, and for this $j$, if only possible, we choose $p$ such that $\widetilde{\cD_p}\in B$.

Observe that $j >1$, because otherwise $x_px'_k \in E(H)$, a contradiction with the choice of $x'_k$.
Let $\overline{x}_k \in \{x_k,x'_k\}$ be the vertex for which $\overline{x}_kd_j^p \in E(H)$. Then we define 
\[
\cD_i := 
\begin{cases}
x_k,x'_k,\ldots,\overline{x}_k \circ \overline{x}_k,d_j^p,\ldots,d_\ell^p & \text{ if } i=k,\\
\widetilde{\cD_i} & \text{ if } i \in [k-1],
\end{cases}
\]
so that they have equal length. Denote the consecutive vertices of $\cD_k$ by $d^k_1,\ldots,d^k_\ell$.
It is clear that these walks satisfy conditions (1) and (2), so we only need to prove the condition (3).
Assume that there exist walks $\cD_q:x_q \to \alpha$ and $\cD_r:x_r \to \beta$, such that $\cD_q$ does not avoid $\cD_r$.
Note that by the inductive assumption this cannot happen if $q,r \in [k-1]$. Thus either $r=k$ or $q=k$.

Assume that $r=k$, i.e., $\cD_k$ is an $x_k$-$\beta$-walk. Note that this means that $\widetilde{\cD_p} \in B$.
Let $c \geq 2$ be the minimum index for which $d_{c-1}^qd_c^k \in E(H)$.
If $c < j$, it means $d^q_{c-1}x_k \in E(H)$ or $d^q_{c-1}x'_k \in E(H)$, which contradicts the minimality of $j$.
If $c \geq j$, then $d^k_{c}=d^p_c$, so $d_{c-1}^q$ is adjacent to $d_c^p$. Thus the $x_q$-$\alpha$ walk $\widetilde{D_q}$ does not avoid the $x_p$-$\beta$ walk $\widetilde{D_p}$, a contradiction.

So assume $q=k$, i.e., $\cD_k$ is an $x_k$-$\alpha$-walk and $\widetilde{\cD_p}$ is an $x_p$-$\alpha$-walk. Let $c \geq 2$ be the smallest index for which $d_{c-1}^kd_c^r \in E(H)$. The argument is analogous: if $c < j$, then we have a contradiction with the minimality of $j$. If $c >j$, then the $x_p$-$\alpha$-walk $\widetilde{\cD_p}$ does not avoid the $x_r$-$\beta$ walk $\widetilde{\cD_r}$, a contradiction.
Finally, if $c=j$, recall that we would choose $r$ instead of $p$, as $\widetilde{\cD_r} \in B$ and $\widetilde{\cD_p} \notin B$. This completes the proof of this case.

\paragraph{Case 2: There are no  edges between $\{x_k,x'_k\}$ and $\bigcup_{i \in [k-1]} \widetilde{\cD_i}$.} Then we use \cref{lem:walks-to-corners-far} for $\alpha,\beta$ and $s=x_k$ to obtain walks $\cP, \cQ$ and $\cR$ and define

\[\cD_i :=
\begin{cases}
\begin{aligned}
&\widetilde{\cD_i}	 &~\circ~& \cP && \textrm{ if } i \in  A, \\ 
&\widetilde{\cD_i}	&~\circ~& \cR && \textrm{ if } i \in  B, \\
&x_k,x'_k,\ldots,x_k &~\circ~& \cQ && \textrm{ if } i=k,
\end{aligned}
\end{cases}\]
so that all walks have the same length $\ell'-1$. We extend the naming of vertices of walks $\cD_i$ by denoting their consecutive vertices by $d^i_1,\ldots, d^i_{\ell'}$, note that this is consistent with previous notation, as for every $i \in [k-1]$ the walk $\widetilde{\cD_i}$ is the prefix of $\cD_i$.

Again, properties (1) and (2) are straightforward, let us verify the property (3). 
As $\cP$ avoids $\cR$ and for every $q \in A$ and $r \in B$ the walk $\cD_q$ avoids $\cD_r$, by \cref{obs:walks-composition} we know that the property (3) holds for all $q,r \in [k-1]$.

So we only need to consider two cases. First, assume that $\cQ$ is an $x_k$-$\alpha$-walk (and thus so it $\cD_k$), and $\cP, \cQ$ avoid $\cR$.
Suppose that there exists $\cD_r \in  B$ such that $\cD_k$ does not avoid $\cD_r$, i.e., there exists $c \geq 2$ such that $d_{c-1}^kd_c^r \in E(H)$. Recall that the number of vertices in $\widetilde{\cD_r}$ is $\ell$.
If $c-1 \geq \ell$, then $\cQ$ does not avoid $\cR$, a contradiction. And if $c-1 < \ell$, then $d_{c-1}^k \in \{x_k,x'_k\}$, so $\widetilde{\cD_r}$ is adjacent to $x_k$ or $x'_k$, a contradiction with the assumption of the case.

Now assume that $\cQ$ is an $x_k$-$\beta$-walk and there exists $q \in A$ such that $\cD_q$ does not avoid $\cD_k$, i.e., there exists $c \geq 2$ such that $d^q_{c-1}d^k_c \in E(H)$. The arguments which lead to the contradiction are analogous -- if $c-1 \geq \ell$, then $\cP$ does not avoid $\cQ$, and if $c-1 < \ell$, then $\widetilde{\cD_q}$ is adjacent to $x_k$ or $x'_k$.
\end{proof}

Now we are ready to prove \cref{lem:walks-from-uncomp}.

\begin{proof}[Proof of \cref{lem:walks-from-uncomp}]
Let $X,Y$ be the bipartition classes of $H$ such that $S = \{x_1,\ldots,x_k \} \subseteq X$, where $k \geq 2, x_1=a, x_2=b$. Again, to simplify the notation, we will write $\cD_i$ instead of $\cD_{x_i}$.
Let $a' \in N(a) \setminus N(b)$ and $b' \in N(b) \setminus N(a)$. We define $x'_1:=a'$, $x'_2:=b'$, and for every $i \geq 3$ such that $x_i \notin N(a')$ we choose $x'_i$ to be any vertex from $N(x_i) \setminus N(a)$; they exist, since $S$ is incomparable. 

Let $U=\{x_i,x'_i : x_i \notin N(a')\} \cup \{a,a'\}$, and let $\cC$ be the set of all connected components of $H \setminus N(a,a')$. Note that $x_i$ and $x'_i$ are always in the same component from $\cC$. Moreover, there is a component $C_a \in \cC$, such that $V(C_a)=\{a,a'\}$, and a component $C_b \in \cC$, such that $b,b' \in V(C_b)$.
For each $C \in \cC$ containing at least one vertex from $U$, we choose one vertex $u \in U_X \cap V(C)$ and call it the \emph{representative} of $C$. The representatives are chosen arbitrarily, except that we choose $b$ as the representative of $C_b$. Note that necessarily $a$ is the representative of $C_a$. 
Let $R \subseteq S$ be the set of all vertices that are representatives of components in $\cC$, clearly $a,b \in R$.
For every vertex of $C \in \cC$, its representative is the representative of $C$.

We claim that $R$ is strongly incomparable. Indeed, note that for every $x_i \in R$, the vertex $x'_i \in N(x_i)$ is non-adjacent to every $x_j \in R\setminus \{x_i\}$. This is because if $x'_i$ is adjacent to some $x_j$, then, since $x'_i$ is non-adjacent to $a$, both $x_i,x_j$ must be in the same component in $\cC$, so they cannot both belong to $R$.
So calling \cref{lem:walks-from-uncomp-special} for the set $R$ and $a,b,\alpha,\beta$ gives us the family of walks $\{ \widetilde{\cD_i}\}_{x_i \in R}$. 

Recall that the only vertices $x_i \in S$, for which $x'_i$ is not defined yet, are in $(S \cap N(a')) \setminus \{a\}$. 
Let us consider such $x_i$, clearly it is adjacent to some vertices in $U_Y$ (at least $a'$). 
If there is some $x_j \in U$, such that:
\begin{compactitem}
\item $x_i$ is adjacent to $x'_j$, and
\item $\widetilde{\cD_r}$ terminates at $\beta$, where $x_r$ is the representative of $x_j$,
\end{compactitem}
then we set $x'_i := x'_j$. Otherwise, we set $x'_i := a'$ (note that $\widetilde{\cD_a}$ terminates at $\alpha$).
Now for every $x_i$ we have defined $x'_i$, and always $x'_i \in U_Y$. In particular, no $x'_i$, except for $a'$, is adjacent to $a$.

Consider a vertex $x_i \in S$ and let $C \in \cC$ be the component containing $x'_i$.
Let $x_r$ be the representative of $C$, and let $\cZ_i$ be a $x'_i$-$x_r$-walk, contained in $C$.
We define \[\cD_i := \underbrace{x_i, x'_i \circ \cZ_i ~\circ~ x_r,x'_r,\ldots,x_r}_{t \textrm{ vertices}} ~\circ~ \widetilde{\cD_r},\] where $t$ is chosen so that all $\cD_i$'s are of equal length.
Let us denote by $d^i_1,\ldots,d^i_\ell$ the consecutive vertices of $\cD_i$, note that $d^i_t,d^i_{t+1},\ldots,d^i_\ell=\widetilde{\cD_r}$ and $d^i_t=x_r$.

It is clear that all walks $\cD_i$ terminate at $\alpha$ or $\beta$, and in particular $\cD_a$ is an $a$-$\alpha$-walk and $\cD_b$ is a $b$-$\beta$-walk, so the properties (1) and (2) hold.
To prove the property (3), suppose that there are $p,q \in [k]$ such that $\cD_p: x_p \to \alpha$ does not avoid $\cD_q: x_q \to \beta$.
So there exists $c\geq 2$ such that $d^p_{c-1}$ is adjacent to $d^q_c$.
Let $x_{p'},x_{q'}$ be, respectively, the representatives of $x_p$ and $x_q$. Clearly $x_p',x_q' \in R$.
Note that if $c-1 \geq t$, then the $x_{p'}$-$\alpha$-walk $\widetilde{\cD_{p'}}$ does not avoid the $x_{q'}$-$\beta$-walk $\widetilde{\cD_{q'}}$, a contradiction with the properties of the walks $\{ \widetilde{\cD_i}\}_{x_i \in R}$ ensured by \cref{lem:walks-from-uncomp-special}.

If $2 \leq c-1 < t$, then there exists an $x_{p'}$-$x_{q'}$-path in $H - N(a,a')$, so they are in the same connected component in $\cC$. Since each component in $\cC$ has exactly one representative, we obtain that $p'=q'$ and thus and $\cD_p$ and $\cD_q$ both terminate in $\alpha$ or in $\beta$.

Finally, consider $c=2$, which means that $d^p_1=x_p$ is adjacent to $d^q_2=x'_q$.
There are three possibilities: (i) $x_p,x'_q \notin N(a,a')$, or (ii) $x_p \in N(a')$, or (iii) $x'_q \in N(a)$.
In case (i) we observe that $x_p$ and $x'_q$ are in the same connected component in $\cC$, which means that $x_{p'}=x_{q'}$, so both walks $\cD_p,\cD_q$ terminate at the same vertex.
In case (ii), recall that when choosing $x'_p$, we gave preference to vertices whose representative's walk terminates at $\beta$ (see the second condition in definition). Thus we would have chosen $x'_p=x'_q$, a contradiction.
Finally, in case (iii), recall that $x'_q=a'$. But since the representative of $a'$ is $a$, the walk $\cD_q$ terminates at $\alpha$, a contradiction.
This completes the proof of the lemma.
\end{proof}

\newpage
\section{Algorithm for general target graphs}
In this section we will generalize the invariant $i^*(H)$ and extend \cref{thm:main-bipartite-algo} to all targets relevant target graphs $H$.
Let us start with a simple observation, which is an analogue of \cref{prop:lists-bipartite}.

\begin{observation}\label{prop:lists-general}
Let $(G,L)$ be an instance of \lhomo{H}. Without loss of generality we might assume the following.
\begin{compactenum}[(1)]
\item The graph $G$ is connected,
\item for each $x \in V(G)$, the set $L(x)$ is incomparable, \label{it:lists-general-incomparable}
\item for each edge $xy \in E(G)$, for every $u \in L(x)$ there is $v \in L(y)$, such that $uv \in E(H)$.
\end{compactenum}
\end{observation}
\begin{proof}
The first two items are analogous to the corresponding ones in \cref{prop:lists-bipartite}.
For the last item observe that if $u \in L(x)$ is non-adjacent (in $H$) to every element of $L(y)$, then we can safely remove $u$ from $L(x)$.
\end{proof}

The high-level idea is to reduce the general case of \lhomo{H} to the case, when the target is bipartite, and then use \cref{thm:main-bipartite-algo}. For this, we will consider the so-called \emph{associated instances}, introduced by Feder, Hell, and Huang~\cite{DBLP:journals/jgt/FederHH03} (see \cref{sec:associated}).

We will also separately consider some special graphs that we call \emph{strong split graphs}. A graph $H$ is a \emph{strong split graph}, if its vertex set can be partitioned into two sets $B$ and $P$, where $B$ is independent and $P$ induces a reflexive clique. We call the pair $(B,P)$ \emph{the partition of $H$}. Note that the partition is unique: all vertices without loops must belong to $B$ and all vertices with loops must belong to $P$.

\subsection{Associated instances and clean homomorphisms}\label{sec:associated}
For a graph $G=(V,E)$, by $G^*$ we denote the \emph{associated bipartite graph}, defined as follows.
The vertex set of $G^*$ is the union of two independent sets: $\{x' \colon x \in V\}$ and $\{x'' \colon x \in V\}$. 
The vertices $x'$ and $y''$ are adjacent if and only if $xy \in E$. Note that the edges of type $x'x''$ in $G^*$ correspond to loops in $G$. 
The vertices $x'$ and $x''$ are called \emph{twins}.

Let $(G,L)$ be an instance of \lhomo{H}. An \emph{associated instance} is the instance $(G^*,L^*)$ of $\lhomo{H^*}$, where $L^*$ are \emph{associated lists} defined as follows. For $x \in V(G)$, we set $L^*(x') = \{u' \colon u \in L(x)\}$ and $L^*(x'') = \{u'' \colon u \in L(x)\}$. Note that in the associated lists, the vertices  appearing in the list of $x'$ are precisely the twins of the vertices appearing in the list of $x''$. A homomorphism $f \colon (G^*,L^*)  \to H^*$ is \emph{clean} if it maps twins to twins, i.e., $f(x') = u'$ if and only if $f(x'') = u''$. 
The following simple observation was the crucial step of the proof of the complexity dichotomy for list homomorphisms, shown by Feder, Hell, and Huang~\cite{DBLP:journals/jgt/FederHH03}. We state it using slightly different language, which is more suitable for our purpose.

\begin{proposition}[Feder, Hell, Huang~\cite{DBLP:journals/jgt/FederHH03}] \label{prop:associated-equiv}
Let $(G,L)$ be an instance of \lhomo{H}. Then it is a yes-instance if and only if $(G^*,L^*)$ admits a clean homomorphism to $H^*$. 
\end{proposition}

Let us point out that the restriction to clean homomorphisms is necessary for the equivalence.
Indeed, consider for example $G=K_3$ and $H=C_6$, so clearly $G \not\to H$. However, we have $G^* \simeq C_6$ and $H^* \simeq 2C_6$, so $G^* \to H^*$.

Recall that if $H$ is bipartite, then \lhomo{H} is polynomial-time solvable if $H$ is the complement of a circular-arc graph~\cite{DBLP:journals/combinatorica/FederHH99}, and NP-complete otherwise.
Feder, Hell, and Huang~\cite{DBLP:journals/jgt/FederHH03} proved the following dichotomy theorem.

\begin{theorem}[Feder, Hell, Huang~\cite{DBLP:journals/jgt/FederHH03}]\label{thm:bi-arcs}
Let $H$ be an arbitrary graph (with loops allowed).
\begin{compactenum}
\item The \lhomo{H} problem is polynomial-time solvable if $H$ is a bi-arc graph, and NP-complete otherwise.
\item The graph $H$ is a bi-arc-graph if and only if $H^*$ is the complement of a circular-arc graph.
\end{compactenum}
\end{theorem}

So for our problem the interesting graphs $H$ are those, for which $H^*$ is not the complement of a circular-arc graph.
In the observation below we summarize some properties of associated instances.

\begin{observation}\label{obs:associated-properties}
Consider an instance $(G,L)$ of \lhomo{H} and the associated instance $(G^*,L^*)$ of \lhomo{H^*}.
Suppose that $G$ is given along with a tree decomposition of width $t$.
\begin{compactenum}[(1)]
\item For each $v \in V(G)$ and $u \in V(H)$, we have $u \in L(v)$ if and only if $u' \in L^*(v')$ if and only if $u'' \in L^*(v'')$. In particular, each list is contained in one bipartition class of $H^*$. \label{it:associated-twins}
\item In polynomial time we can construct a tree decomposition $\cT^*$ of $G^*$ of width at most $2t$ with the property that for each $x \in V(G)$, each bag of $\cT$ either contains both $x',x''$ or none of them. \label{it:associated-treedecomp}
\end{compactenum}
\end{observation}
\begin{proof}
The first item follows directly from the definition of $(G^*,L^*)$.
To see the second item, consider a tree decomposition $\cT$ of $G$. We construct $\cT^*$ by taking the same tree structure as for $\cT$, and replacing each vertex $v$ of $G$ in each bag of $\cT$ by the vertices $v',v''$ of $G^*$.
It is straightforward to verify that this way we obtain a tree decomposition of $\cT^*$ with the desired properties.
\end{proof}

Observe that for bipartite $H$, the graph $H^*$ consists of two disjoint copies of $H$, so clearly $i^*(H) = i^*(H^*)$.
This motivates the following definition, generalizing \cref{def:i_star}.
\begin{definition}[$i^*(H)$ for general $H$]
Let $H$ be a connected non-bi-arc graph. Define
\[i^*(H) := i^*(H^*).\]
\end{definition}

\subsection{Decompositions of generals target graphs} \label{sec:general-decompositions}
In this section we generalize the notion of decompositions of bipartite graphs, introduced in \cref{sec:decomposition}, to all graphs (with possible loops).
The high-level idea is to define decompositions of $H$, so that they will correspond to bipartite decompositions of $H^*$. We consider the following three types of decompositions of a graph $H$ (see \cref{fig:decompos}). Note that unless stated explicitly, we do not insist that any of the defined sets is non-empty.

\begin{definition}[$F$-decomposition]\label{def:f-decomposition}
A partition of $V(H)$ into an ordered triple of sets $(F,K,Z)$ is an \emph{$F$-decomposition} if the following conditions are satisfied (see \cref{fig:decompos}, left).
\begin{compactenum}
\item $K$ is non-empty and it separates $F$ and $Z$, \label{it:fdecomp-separator}
\item $|F| \geq 2$, \label{it:fdecomp-geq2}
\item $K$ induces a reflexive clique,\label{it:fdecomp-cliqueis}
\item $F$ is complete to $K$. \label{it:fdecomp-complete}
\end{compactenum}
\end{definition}

\begin{definition}[$BP$-decomposition]\label{def:bp-decomposition}
A partition of $V(H)$ into an ordered five-tuple of sets $(B,P,M,K,Z)$ is a \emph{$BP$-decomposition} if the following conditions are satisfied (see \cref{fig:decompos}, middle).
\begin{compactenum}
\item $K \cup M$ is non-empty and there are no edges between $(P \cup B)$ and $Z$,
\item $|P| \geq 2$ or $|B| \geq 2$,
\item $K \cup P$ induces a reflexive clique and $B$ is an independent set,
\item $M$ is complete to $P \cup K$ and $B$ is complete to $K$,
\item $B$ is non-adjacent to $M$.
\end{compactenum}
\end{definition}

\begin{definition}[$B$-decomposition] \label{def:b-decomposition}
A partition of $V(H)$ into an ordered six-tuple of sets $(B_1,B_2,K,M_1,M_2,Z)$ is a \emph{$B$-decomposition} if the following conditions are satisfied (see \cref{fig:decompos}, right).
\begin{compactenum}
\item $K \cup M_1 \cup M_2$ is non-empty and it separates $(B_1 \cup B_2)$ and  $Z$, \label{it:bdecomp-separator}
\item $|B_1| \geq 2$ or $|B_2| \geq 2$, \label{it:bdecomp-geq2}
\item $K$ induces a reflexive clique and each of $B_1,B_2$ is an independent set, \label{it:bdecomp-cliqueis}
\item $K$ is complete to $M_1 \cup M_2 \cup B_1 \cup B_2$, and $M_2$ is complete to $M_1 \cup B_1$, and $M_1$ is complete to $B_2$, \label{it:bdecomp-complete}
\item $B_1$ is non-adjacent to $M_1$ and $B_2$ is non-adjacent to $M_2$. \label{it:bdecomp-nonadjacent}
\end{compactenum}
\end{definition}

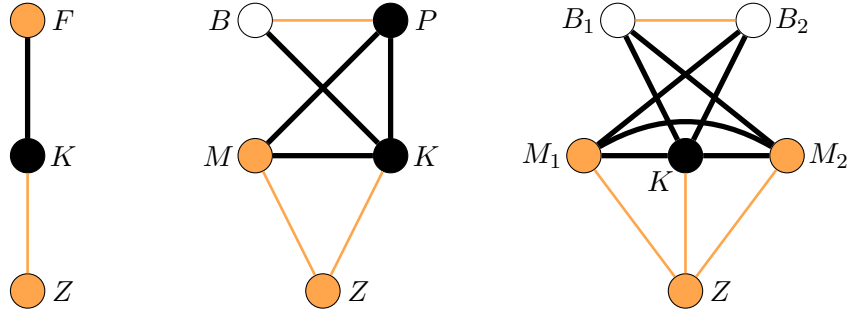
\begin{figure}[h]
\begin{center}
\begin{tikzpicture}[every node/.style={draw,circle,fill=white,inner sep=0pt,minimum size=13pt},every loop/.style={},scale=0.9]
\node[fill=black, label=right:$K$] (K) at (2,0) {};
\node[fill=\yourfavouritecolor,label=right:$F$] (F) at (2,2) {};
\node[fill=\yourfavouritecolor,label=right:$Z$] (Z) at (2,-2) {};

\draw[line width=2] (F) -- (K);
\draw[\yourfavouritecolor,line width=1] (K) -- (Z);
\end{tikzpicture} \hskip 1.5cm
\begin{tikzpicture}[every node/.style={draw,circle,fill=white,inner sep=0pt,minimum size=13pt},every loop/.style={},scale=0.9]
\node[fill=\yourfavouritecolor, label=left:$M$] (MO) at (0,0) {};
\node[fill=black, label=right:$K$] (K) at (2,0) {};
\node[label=left:$B$] (BO) at (0,2) {};
\node[fill=black,label=right:$P$] (P) at (2,2) {};

\node[fill=\yourfavouritecolor,label=right:$Z$] (Z) at (1,-2) {};
\draw[line width=2] (P) -- (K);
\draw[line width=2] (BO) -- (K);
\draw[line width=2] (P) -- (MO);
\draw[line width=2] (K) -- (MO);

\draw[\yourfavouritecolor,line width=1] (MO) -- (Z);
\draw[\yourfavouritecolor,line width=1] (K) -- (Z);
\draw[\yourfavouritecolor,line width=1] (BO) -- (P);
\end{tikzpicture}\hskip 1cm
\begin{tikzpicture}[every node/.style={draw,circle,fill=white,inner sep=0pt,minimum size=13pt},every loop/.style={},scale=0.9]
\node[fill=\yourfavouritecolor, label=left:$M_1$] (MO) at (0.5,0) {};
\node[fill=black, label=225:$K$] (K) at (2,0) {};
\node[fill=\yourfavouritecolor, label=right:$M_2$] (MT) at (3.5,0) {};
\node[label=left:$B_1$] (BO) at (1,2) {};
\node[label=right:$B_2$] (BT) at (3,2) {};

\node[fill=\yourfavouritecolor,label=right:$Z$] (Z) at (2,-2) {};
\draw[line width=2] (BO) -- (K);
\draw[line width=2] (BT) -- (K);
\draw[line width=2] (BT) -- (MO);
\draw[line width=2] (BO) -- (MT);
\draw[line width=2] (K) -- (MT);
\draw[line width=2] (K) -- (MO);
\draw[line width=2] (MO) to[bend left] (MT);

\draw[\yourfavouritecolor,line width=1] (MO) -- (Z) -- (MT);
\draw[\yourfavouritecolor,line width=1] (K) -- (Z);
\draw[\yourfavouritecolor,line width=1] (BT) -- (BO);
\end{tikzpicture}
\end{center}

\caption{A schematic view of an $F$-decomposition (left), a $BP$-decomposition (middle), and a $B$-decomposition of $H$ (right).
Disks correspond to sets of vertices: white ones depict independent sets, black ones depict reflexive cliques, and orange ones depict arbitrary subgraphs. Similarly, thick black lines indicate that all possible edges between two sets exist, and thin orange lines depict edges that might exist, but do not have to. The lack of a line means that there are no edges between two sets.}
\label{fig:decompos}
\end{figure}

Observe that a graph $H$ can have more than one decomposition: for example, if $(B,P,M,K,Z)$ is an $BP$-decomposition of $H$, but $M=\emptyset$, then $(B \cup P, K, Z)$ is an $F$-decomposition of $H$. 

For each kind of decomposition, we define its \emph{factors} as the following pair of graphs $(H_1, H_2)$.
\begin{compactdesc}
\item[for an $F$-decomposition:] $H_1=H[F]$ and $H_2$ is obtained from $H$ by contracting $F$ to a vertex $f$. It has a loop if and only if $F$ is not an independent set.
\item[for a $BP$-decomposition:] $H_1=H[B \cup P]$ and $H_2$ is obtained from $H$ by contracting $P$ and $B$ respectively (if they are non-empty), to vertices $p$ and $b$, such that $p$ has a loop and $b$ does not. Also, $pb \in E(H_2)$ if and only if there is any edge between $P$ and $B$ in $H$.
\item[for a $B$-decomposition:] $H_1=H[B_1 \cup B_2]$ and $H_2$ is obtained from $H$ by contracting $B_1$ and $B_2$ respectively (if they are non-empty), to vertices $b_1$ and $b_2$ (without loops). Also, $b_1b_2 \in E(H_2)$ if and only if there is any edge between $B_1$ and $B_2$ in $H$.
\end{compactdesc}

Let us prove that if $H$ is not a strong split graph, then the three types of decompositions defined above precisely correspond to bipartite decompositions of the associated bipartite graph $H^*$.
For any $W \subseteq V(H)$, we define two subsets of $V(H^*)$ as follows: $W':=\{x' : x \in W\}$ and $W'':=\{x'' : x \in W\}$.

\begin{lemma} \label{lem:decompositions-equivalent}
Let $H$ be be a connected, non-bi-arc graph, which is not a strong split graph.
Then $H^*$ admits a bipartite decomposition if and only if $H$ admits a $B$-, a $BP$-, or an $F$-decomposition.
\end{lemma}
\begin{proof}
First, observe that a subset $W \subseteq V(H)$ induces a reflexive clique in $H$ if and only if $W' \cup W''$ induces a biclique in $H^*$.

Let us show that if $H$ has a decomposition, then $H^*$ has a bipartite decomposition $(D,N,R)$. We consider three cases, depending on the type of a decomposition of $H$.

If $H$ has an $F$-decomposition $(F,K,Z)$, we define $D,N,R$ as follows (see \cref{fig:decompos-all} a)):
\begin{align*}
D := & F' \cup F'',\\
N := & K' \cup K'',\\
R := & Z' \cup Z''.
\end{align*}
Now the fact that $(D,N,R)$ is a bipartite decomposition of $H^*$ (recall \cref{def:bipartite-decomposition}) follows directly from the definition of a $F$-decomposition (recall \cref{def:f-decomposition}): each property in \cref{def:bipartite-decomposition} follows from the corresponding property in \cref{def:f-decomposition}.

The other two cases are analogous.
If $H$ has a $BP$-decomposition, then we define $D,N,R$ as follows (see \cref{fig:decompos-all} b)):
\begin{align*}
D := & B' \cup P'',\\
N := & K'\cup M' \cup P' \cup K'',\\
R := & Z' \cup Z'' \cup M'' \cup B''.
\end{align*}
Finally, if $H$ has a $B$-decomposition $(B_1,B_2,K,M_1,M_2,Z)$, then we define $D,N,R$ as follows (see \cref{fig:decompos-all}~c)):
\begin{align*}
D := & B_1' \cup B_2'',\\
N := & K'\cup M_1' \cup K'' \cup M_2'',\\
R := &  Z' \cup M_2' \cup B_2' \cup Z'' \cup M_1'' \cup B_1''.
\end{align*}
It is straightforward to verify that in both cases $(D,N,R)$ is a bipartite decomposition of $H^*$.

\begin{figure}[h]
\begin{center}
\begin{tikzpicture}[every node/.style={draw,circle,fill=white,inner sep=0pt,minimum size=11pt},every loop/.style={},scale=0.7]
\node[draw=none, fill=none] at (.5,0) {a)};
\node[fill=black, label=right:$K$] (K) at (2,0) {};
\node[fill=\yourfavouritecolor,label=right:$F$] (F) at (2,1) {};
\node[fill=\yourfavouritecolor,label=right:$Z$] (Z) at (2,-1) {};

\draw[line width=2] (F) -- (K);
\draw[\yourfavouritecolor,line width=1] (K) -- (Z);

\node[label=left:$K'$] (K1) at (5,0) {};
\node[label=left:$F'$] (F1) at (5,1) {};
\node[label=left:$Z'$] (Z1) at (5,-1) {};

\node[label=right:$K''$] (K2) at (8,0) {};
\node[label=right:$F''$] (F2) at (8,1) {};
\node[label=right:$Z''$] (Z2) at (8,-1) {};

\draw[line width=2] (F1) -- (K2);
\draw[line width=2] (F2) -- (K1);
\draw[\yourfavouritecolor,line width=1] (K1) -- (Z2);
\draw[\yourfavouritecolor,line width=1] (K2) -- (Z1);
\draw[\yourfavouritecolor,line width=1] (Z2) -- (Z1);
\draw[black,line width=2] (K2) -- (K1);
\draw[\yourfavouritecolor,line width=1] (F2) -- (F1);
\draw[dashed,red,line width=1.5, rounded corners=10pt] (6,.35) --(4.65,.35) -- (4.65,-.35) -- (8.35, -.35) -- (8.35,.35) -- (6,.35);

\end{tikzpicture}
\end{center}
\begin{tikzpicture}[every node/.style={draw,circle,fill=white,inner sep=0pt,minimum size=11pt},every loop/.style={},scale=0.7]
\node[draw=none, fill=none] at (-1.3,1.5) {b)};
\node[draw=none, fill=none] at (0,-2) {};
\node[fill=\yourfavouritecolor, label=left:$M$] (MO) at (0,1) {};
\node[fill=black, label=right:$K$] (K) at (2,1) {};
\node[label=left:$B$] (BO) at (0,3) {};
\node[fill=black,label=right:$P$] (P) at (2,3) {};

\node[fill=\yourfavouritecolor,label=right:$Z$] (Z) at (1,-1) {};
\draw[line width=2] (P) -- (K);
\draw[line width=2] (BO) -- (K);
\draw[line width=2] (P) -- (MO);
\draw[line width=2] (K) -- (MO);

\draw[\yourfavouritecolor,line width=1] (MO) -- (Z);
\draw[\yourfavouritecolor,line width=1] (K) -- (Z);
\draw[\yourfavouritecolor,line width=1] (BO) -- (P);

\draw[\yourfavouritecolor,line width=1] (K) -- (Z);

\node[label=left:$K'$] (K1) at (5,2) {};
\node[label=left:$B'$] (B1) at (5,3) {};
\node[label=left:$M'$] (M1) at (5,1) {};
\node[label=left:$P'$] (P1) at (5,0) {};
\node[label=left:$Z'$] (Z1) at (5,-1) {};

\node[label=right:$K''$] (K2) at (8,2) {};
\node[label=right:$B''$] (B2) at (8,0) {};
\node[label=right:$M''$] (M2) at (8,1) {};
\node[label=right:$P''$] (P2) at (8,3) {};
\node[label=right:$Z''$] (Z2) at (8,-1) {};

\draw[\yourfavouritecolor,line width=1] (B1) -- (P2);
\draw[\yourfavouritecolor,line width=1] (B2) -- (P1);
\draw[line width=2] (B1) -- (K2);
\draw[line width=2] (B2) -- (K1);
\draw[line width=2] (M1) -- (P2);
\draw[line width=2] (M2) -- (P1);
\draw[line width=2] (K1) -- (P2);
\draw[line width=2] (K2) -- (P1);
\draw[line width=2] (K1) -- (M2);
\draw[line width=2] (K2) -- (M1);
\draw[line width=2] (P1) -- (P2);
\draw[line width=2] (K2) -- (K1);
\draw[\yourfavouritecolor,line width=1] (M2) -- (M1);
\draw[\yourfavouritecolor,line width=1] (Z2) -- (Z1);
\draw[\yourfavouritecolor,line width=1] (Z2) -- (M1);
\draw[\yourfavouritecolor,line width=1] (M2) -- (Z1);
\draw[\yourfavouritecolor,line width=1] (Z2) -- (K1);
\draw[\yourfavouritecolor,line width=1] (K2) -- (Z1);

\draw[dashed,red,line width=1.5, rounded corners=10pt] (6,2.35) --(4.65,2.35) -- (4.65,-.4)--(5.35,-.4) -- (8.35,1.75) -- (8.35,2.35) -- (6,2.35);
\end{tikzpicture}\hskip .5cm
\begin{tikzpicture}[every node/.style={draw,circle,fill=white,inner sep=0pt,minimum size=11pt},every loop/.style={},scale=0.7]
\node[draw=none, fill=none] at (-1,1.5) {c)};
\node[fill=\yourfavouritecolor, label=left:$M_1$] (MO) at (0.5,1) {};
\node[fill=black, label=225:$K$] (K) at (2,1) {};
\node[fill=\yourfavouritecolor, label=right:$M_2$] (MT) at (3.5,1) {};
\node[label=left:$B_1$] (BO) at (1,3) {};
\node[label=right:$B_2$] (BT) at (3,3) {};

\node[fill=\yourfavouritecolor,label=right:$Z$] (Z) at (2,-1) {};
\draw[line width=2] (BO) -- (K);
\draw[line width=2] (BT) -- (K);
\draw[line width=2] (BT) -- (MO);
\draw[line width=2] (BO) -- (MT);
\draw[line width=2] (K) -- (MT);
\draw[line width=2] (K) -- (MO);
\draw[line width=2] (MO) to[bend left] (MT);

\draw[\yourfavouritecolor,line width=1] (MO) -- (Z) -- (MT);
\draw[\yourfavouritecolor,line width=1] (K) -- (Z);
\draw[\yourfavouritecolor,line width=1] (BT) -- (BO);

\node[label=left:$K'$] (K1) at (7,2) {};
\node[label=left:$B_1'$] (B11) at (7,3) {};
\node[label=left:$B_2'$] (B12) at (7,0) {};
\node[label=left:$M_1'$] (M11) at (7,1) {};
\node[label=left:$M_2'$] (M12) at (7,-1) {};
\node[label=left:$Z'$] (Z1) at (7,-2) {};

\node[label=right:$K''$] (K2) at (10,2) {};
\node[label=right:$B_1''$] (B21) at (10,0) {};
\node[label=right:$B_2''$] (B22) at (10,3) {};
\node[label=right:$M_1''$] (M21) at (10,-1) {};
\node[label=right:$M_2''$] (M22) at (10,1) {};
\node[label=right:$Z''$] (Z2) at (10,-2) {};

\draw[\yourfavouritecolor,line width=1] (B11) -- (B22);
\draw[\yourfavouritecolor,line width=1] (B21) -- (B12);
\draw[line width=2] (B11) -- (K2);
\draw[line width=2] (B12) -- (K2);
\draw[line width=2] (B22) -- (K1);
\draw[line width=2] (B21) -- (K1);
\draw[line width=2] (M11) -- (K2);
\draw[line width=2] (M12) -- (K2);
\draw[line width=2] (M22) -- (K1);
\draw[line width=2] (M21) -- (K1);
\draw[\yourfavouritecolor,line width=1] (M22) -- (M12);
\draw[\yourfavouritecolor,line width=1] (M21) -- (M11);
\draw[line width=2] (B22) -- (M11);
\draw[line width=2] (B21) -- (M12);
\draw[line width=2] (M21) -- (B12);
\draw[line width=2] (M22) -- (B11);
\draw[line width=2] (M22) -- (M11);
\draw[line width=2] (M21) -- (M12);
\draw[\yourfavouritecolor,line width=1] (Z2) -- (M12);
\draw[\yourfavouritecolor,line width=1] (Z2) -- (M11);
\draw[\yourfavouritecolor,line width=1] (M22) -- (Z1);
\draw[\yourfavouritecolor,line width=1] (M21) -- (Z1);
\draw[\yourfavouritecolor,line width=1] (K1) -- (Z2);
\draw[\yourfavouritecolor,line width=1] (K2) -- (Z1);
\draw[\yourfavouritecolor,line width=1] (Z2) -- (Z1);
\draw[black,line width=2] (K2) -- (K1);

\draw[dashed,red,line width=1.5, rounded corners=10pt] (8,2.35) --(6.65,2.35) -- (6.65,.6) -- (10.35,.6) -- (10.35,2.35) -- (8,2.35);
\end{tikzpicture}
\caption{Decompositions of a graph $H$ (left) and their corresponding bipartite decompositions $(D,N,R)$ of $H^*$ (right):  a) an $F$-decomposition, b) a $BP$-decomposition and c) a $B$-decomposition. Dashed lines mark the set $N$.}
\label{fig:decompos-all}
\end{figure}
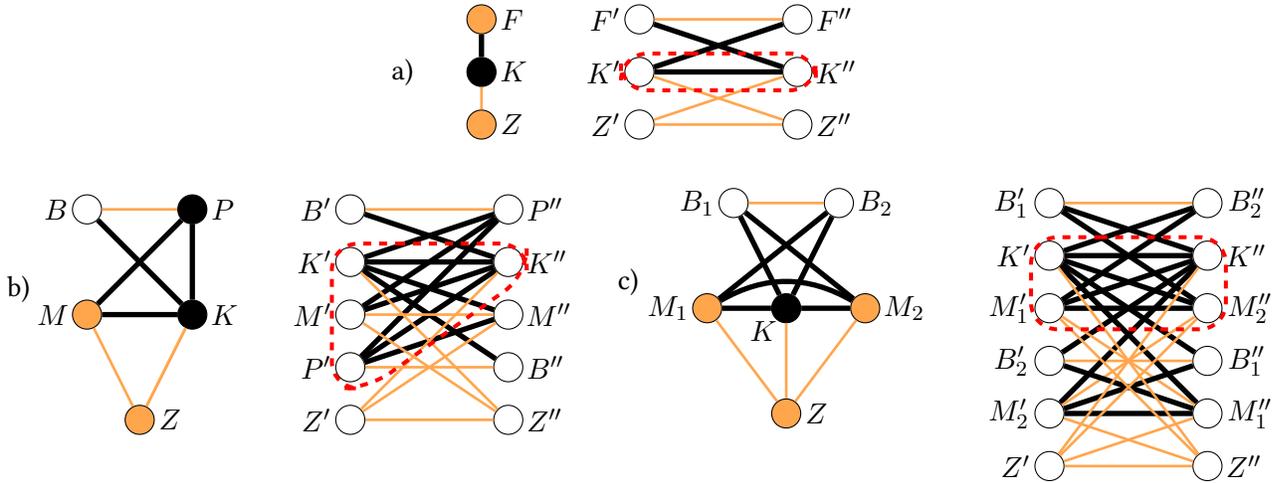

Now we suppose that $(D,N,R)$ is a bipartite decomposition of $H^*$ and recall that $N \neq \emptyset$ and one bipartition class of $D$ has at least two elements.
We aim to show that $H$ has one of the decompositions defined above.
We partition the vertices of $H$ into nine sets as follows (see \cref{fig:vertices-partition}):

\begin{figure}[h]
\begin{center}
\begin{tikzpicture}[scale=0.8]
\draw[line width=2] (0,5)--(3,2.5)--(0,2.5)--(3,5);
\draw[orange, line width=1] (0,2.5)--(3,0)--(0,0)--(3,2.5);
\draw[orange, line width=1] (0,5)--(3,5);
\foreach \i in {0,2.5,5}
{
\foreach \j in {0,3}
{
\draw[fill=white] (\j,\i) circle (1);
\draw (\j,\i)--(\j,\i+1);
\draw (\j,\i)--++(-30:1);
\draw (\j,\i)--++(210:1);
}
}
\draw (0.45,-0.1) node[above] {$M'_2$};
\draw (-0.5,-0.1) node[above] {$B'_2$};
\draw (0,-0.8) node[above] {$Z'$};
\draw (3.4,-0.1) node[above] {$M''_1$};
\draw (2.5,-0.1) node[above] {$B''_1$};
\draw (3,-0.8) node[above] {$Z''$};
\draw (0.5,2.5) node[above] {$K'$};
\draw (-0.5,2.5) node[above] {$Q'$};
\draw (0,1.6) node[above] {$M'_1$};
\draw (3.5,2.5) node[above] {$K''$};
\draw (2.5,2.5) node[above] {$P''$};
\draw (3,1.6) node[above] {$M''_2$};
\draw (0.5,5) node[above] {$P'$};
\draw (-0.5,5) node[above] {$F'$};
\draw (0,4.1) node[above] {$B'_1$};
\draw (3.5,5) node[above] {$Q''$};
\draw (2.5,5) node[above] {$F''$};
\draw (3,4.1) node[above] {$B''_2$};

\draw[dashed] (-1.5,6.2) --++ (6,0) --++ (0,-2.4) --++ (-6,0) --cycle;
\draw[dashed] (-1.5,3.7) --++ (6,0) --++ (0,-2.4) --++ (-6,0) --cycle;
\draw[dashed] (-1.5,1.2) --++ (6,0) --++ (0,-2.4) --++ (-6,0) --cycle;
\draw (-2,5) node[above] {$D$};
\draw (-2,2.5) node[above] {$N$};
\draw (-2,0) node[above] {$R$};
\end{tikzpicture}
\caption{Schematic definition of sets in the proof of \cref{lem:decompositions-equivalent}. E.g. $P$ is the set of those $x \in V(G)$, for which $x' \in D$  and $x'' \in N$.
}
\label{fig:vertices-partition}
\end{center}
\end{figure}
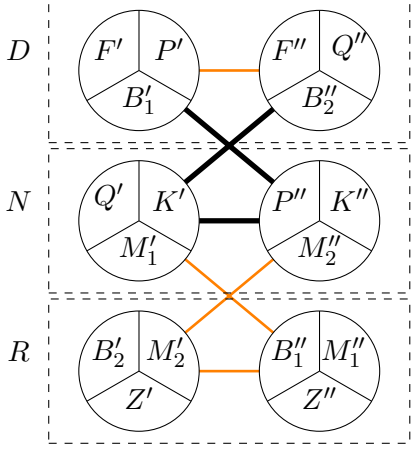

\begin{equation*}
\begin{alignedat}{10}
F:=&&\{x: x' \in D, x'' \in D\}, \ && Q:=&&\{x: x' \in N, x'' \in D\}, \ && B_2:=&&\{x: x' \in R, x'' \in D\}, \\ 
P:=&&\{x: x' \in D, x'' \in N\}, \ && K:=&&\{x: x' \in N, x'' \in N\}, \ && M_2:=&&\{x: x' \in R, x'' \in N\}, \\ 
B_1:=&&\{x: x' \in D, x'' \in R\}, \ && M_1:=&&\{x: x' \in N, x'' \in R\}, \ && Z:=&&\{x: x' \in R, x'' \in R\}.
\end{alignedat}
\end{equation*}
Clearly, from the definition of the bipartite decomposition it follows that some pairs of these sets cannot be both non-empty (e.g. $F$ and $M_2$, because $F'$ must be complete to $M_2''$ and in the same time $F''$ is non-adjacent to $M_2'$). Observe that  $M_1 \cup Q \cup K$ is complete to $Q \cup B_2 \cup F \cup P \cup M_2 \cup K$ and $P \cup M_2 \cup K$ is complete to $B_1 \cup P \cup F \cup M_1 \cup Q \cup K$. Also, $B_1 \cup P \cup F$ is non-adjacent to $M_1 \cup B_1 \cup Z$ and $B_2 \cup Q \cup F$ is non-adjacent to $M_2 \cup B_2 \cup Z$. In particular, it implies that $P, Q,$ and $K$ must be reflexive cliques and $B_1,B_2$ are independent sets. Finally, at least one of the sets $B_1 \cup P \cup F$ and $B_2 \cup Q \cup F$ has at least two vertices and at least one of the sets $M_1 \cup Q \cup K$ and $P \cup M_2 \cup K$ is non-empty.

\paragraph{Case 1: $F \neq \emptyset$.} It implies that $M_1, M_2 = \emptyset$. If $K \neq \emptyset$, then it is straightforward to observe that $(F \cup P \cup Q \cup B_1 \cup B_2, K, Z)$ is an $F$-decomposition of $H$.
If $K = \emptyset$, then also $Z=\emptyset$, as we assumed that $H$ is connected.
Moreover, since $M_1 \cup M_2 \cup P \cup Q \cup K \neq \emptyset$, at least one of $P,Q$ is non-empty. Assume that $P \neq \emptyset$, as the other case is symmetric.
Recall that this means that $B_1=\emptyset$.

If $Q \neq \emptyset$, then $B_2 = \emptyset$ and $(F \cup Q,P,\emptyset)$ is an $F$-decomposition. So let $Q = \emptyset$.
If $|F| \geq 2$, then $(F,P,B_2)$ is an $F$-decomposition.
If $|B_2| \geq 2$ or $|P| \geq 2$, then $(B_2, P, F, \emptyset, \emptyset)$ is a $BP$-decomposition.
So in the last case we have $|F|=|P|=1$ and $|B_2| \leq 1$. It is easy to verify that then $H$ is a bi-arc graph (or, equivalently, $H^*$ does not contain an induced cycle on at least 6 vertices or an edge asteroid). This contradicts our assumption on $H$.

\paragraph{Case 2: $F = \emptyset$.} We consider three subcases: either $B_1,B_2 \neq \emptyset$, or $B_1=\emptyset$ and $B_2\neq\emptyset$ (the case $B_1\neq\emptyset$ and $B_2=\emptyset$ is symmetric), or $B_1,B_2 = \emptyset$.
The first case implies that $P,Q = \emptyset$, so we immediately obtain a $B$-decomposition $(B_1,B_2,K,M_1,M_2,Z)$.

In the second one we have $Q=\emptyset$.
If $P \neq \emptyset$, then $M_1=\emptyset$.
If additionally $M_2 \cup K \neq \emptyset$, then there exists a $BP$-decomposition $(B_2,P,M_2,K,Z)$.
On the other hand, if $M_2 \cup K = \emptyset$, then $Z = \emptyset$ by connectivity of $H$.
Thus the only non-empty sets are $P$ and $B_2$, and therefore $H$ is a strong split graph.
Finally, if $P=\emptyset$, then $H$ admits a $B$-decomposition $(\emptyset,B_2,K,M_1,M_2,Z)$.

So let us assume that $B_1,B_2 = \emptyset$. We will consider further subcases:  
either $M_1,M_2\neq \emptyset$, or $M_1=\emptyset$ and $M_2\neq\emptyset$ (the other case, i.e., $M_1\neq\emptyset$ and $M_2=\emptyset$, is symmetric), or $M_1,M_2=\emptyset$.
Note that if $M_1,M_2\neq \emptyset$, then $P,Q=\emptyset$, so $D = \emptyset$, which is a contradiction with $(D,N,R)$ being a bipartite decomposition of $H^*$.
If $M_1=\emptyset$ and $M_2\neq\emptyset$, then $Q=\emptyset$ and $|P|\geq 2$ and we have a $BP$-decomposition $(\emptyset,P,M_2,K,Z).$ Lastly, when $M_1,M_2=\emptyset$, then without loss of generality $|P|\geq 2$ and either $(P,Q\cup K,Z)$ is an $F$-decomposition, or $Q \cup K=\emptyset$, then by connectivity of $H$ the set $Z$ must be empty and $H$ is just a reflexive clique induced by $P$. However, such a graph is a bi-arc graph, a contradiction.
\end{proof}

Let us point out that one application of a $BP$-decomposition or a $B$-decomposition corresponds to two consecutive applications of a bipartite decomposition in $H^*$.
Consider the case of a $BP$-decomposition and the bipartite decomposition $(B' \cup P'', K' \cup M' \cup P' \cup K'', Z' \cup M'' \cup B'' \cup Z'')$ of $H'$. Note that after contracting $B'$ to $b'$ and $P''$ to $p''$, we still have a decomposition $(P' \cup B'', K' \cup K'' \cup M'' \cup \{p''\}, M' \cup \{b'\} \cup Z' \cup Z'')$ of the second factor $(H^*)_2$ of $H^*$.
Similarly, in the case of a $B$-decomposition, we still have another bipartite decomposition of $(H^*)_2$, where the new set $D$ is $B_1'' \cup B_2'$.

Note that in a $BP$-decomposition, a $B$-decomposition, and an $F$-decomposition, when $F$ is an independent set or contains a vertex with a loop, the factors are always induced subgraphs of $H$.
Indeed, we can equivalently obtain $H_2$ by removing certain vertices from $H$. In the last case of an $F$-decomposition, when $F$ contains at least one edge and has only vertices without loops, $H_2$ is not an induced subgraph of $H$.
We can equivalently define $H_2$ as the graph obtained by removing from $F$ all but two vertices that are adjacent to each other, and then replacing them with a vertex with a loop. 

In the following lemma we consider a graph $H'$ that was obtained from $H$ by a series of decompositions (i.e., it is a factor of $H$, or a factor of a factor of $H$ etc.). We show that even if $H'$ is not an induced subgraph of $H$, the associated bipartite graph $H'^*$ is still an induced subgraph of $H^*$.

\begin{lemma} \label{lem:induced-hstar}
Let $H'$ be a graph obtained from $H$ by a series of decompositions.
Then $H'^*$ is an induced subgraph of $H^*$.
\end{lemma}

\begin{proof}
Recall that by the definitions of decompositions, $H'$ was obtained from $H$ by a sequence of two types of operations:
\begin{compactenum}[(O1)]
\item removing some vertices or, equivalently, taking an induced subgraph ($BP$-decomposition, $B$-decomposition, and $F$-decomposition when the set $F$ contains a vertex with a loop or is independent),
\item removing some vertices and then contracting two adjacent vertices $a$, $b$, such that $N(a) \setminus \{b\} = N(b) \setminus \{a\}$ and none of $a,b$ has a loop, to a vertex $c$ with a loop; vertices $a$ and $b$ are removed from the graph and the new vertex $c$ becomes adjacent to all vertices in $N(a) \setminus \{b\} = N(b) \setminus \{a\}$  ($F$-decomposition, when $F$ has no vertex with a loop and contains at least one edge, see \cref{fig:decomp-operation-o2}).
\end{compactenum}

Observe that if only the first type of operation was applied, then $H'$ is an induced subgraph of $H$, which implies that $H'^*$ is an induced subgraph of $H^*$.
Let us analyze the case when some operations of the second type were applied as well.

In this case $H'$ might not be an induced subgraph of $H$, but each newly created vertex $c$ uniquely corresponds to two adjacent vertices, i.e., $a,b$. Moreover, when the operation was applied, $a$ and $b$ had the same neighborhoods in the current graph, except of being adjacent to each other. Note that both $a$ and $b$ are vertices of $H$: they do not have loops, while we only add vertices with loops.
Furthermore, when $c$ is created, $a$ and $b$ are removed from a graph, so each vertex is used in the operation of the second type at most once. 

Let us consider $H'^*$. Note that the edge $ab$ of $H$ corresponds to two edges $a'b''$ and $a''b'$ in $H^*$.
Note that we can map the vertex $c'$ to $a'$ and $c''$ to $b''$.
Since $N(a) \setminus \{b\} = N(b) \setminus \{a\}$, when the operation was applied, the associated bipartite graph of the obtained graph is indeed an induced subgraph of $H^*$, and the mapping mentioned above defines the isomorphism from $H'^*$ to a subgraph of $H^*$.
\end{proof}

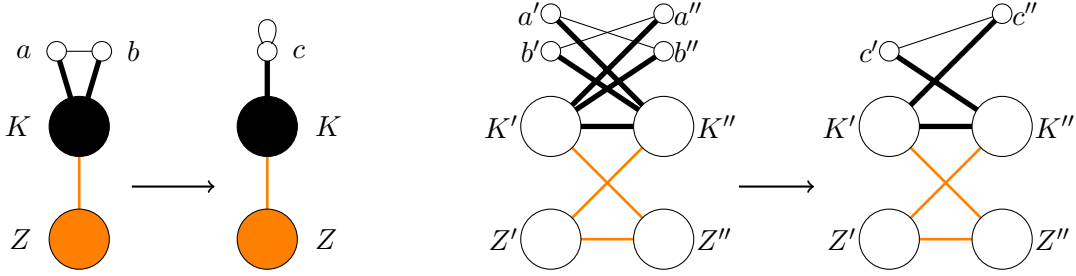
\begin{figure}
\begin{center}
\begin{tikzpicture}[every loop/.style={}]
\draw[line width=2] (-0.3,2.5)--(0,1.5)--(0.3,2.5);
\draw[orange, line width=1] (0,0)--(0,1.5);
\draw (-0.3,2.5)--(0.3,2.5);
\draw[line width=2] (2.5,2.5)--(2.5,1.5);
\draw[orange, line width=1] (2.5,0)--(2.5,1.5);
\draw[fill=orange] (0,0) circle (0.4);
\draw[fill=black] (0,1.5) circle (0.4);
\draw[fill=white] (-0.3,2.5) circle (0.13);
\draw[fill=white] (0.3,2.5) circle (0.13);
\draw[fill=orange] (2.5,0) circle (0.4);
\draw[fill=black] (2.5,1.5) circle (0.4);
\draw[fill=white] (2.5,2.5) circle (0.13);
\draw[->,thick] (0.7,0.7)--(1.8,0.7);
\draw (-0.5,0) node[left] {$Z$};
\draw (-0.5,1.5) node[left] {$K$};
\draw (-0.5,2.5) node[left] {$a$};
\draw (0.5,2.5) node[right] {$b$};
\draw (3,0) node[right] {$Z$};
\draw (3,1.5) node[right] {$K$};
\draw (2.7,2.5) node[right] {$c$};

\node[minimum size=5.5pt,inner sep=0pt,outer sep=0pt] at (2.5,2.5) {} edge [in=60,out=120,loop] ();
\end{tikzpicture} \qquad \qquad
\begin{tikzpicture}
\draw (0,3)--(1.5,2.5);
\draw (0,2.5)--(1.5,3);
\draw[line width=2] (0,3)--(1.5,1.5)--(0,1.5)--(1.5,3);
\draw[line width=2] (0,2.5)--(1.5,1.5);
\draw[line width=2] (0,1.5)--(1.5,2.5);
\draw[orange, line width=1] (0,1.5)--(1.5,0)--(0,0)--(1.5,1.5);
\draw[fill=white] (0,0) circle (0.4);
\draw[fill=white] (1.5,0) circle (0.4);
\draw[fill=white] (0,1.5) circle (0.4);
\draw[fill=white] (1.5,1.5) circle (0.4);
\draw[fill=white] (0,2.5) circle (0.13);
\draw[fill=white] (1.5,2.5) circle (0.13);
\draw[fill=white] (0,3) circle (0.13);
\draw[fill=white] (1.5,3) circle (0.13);
\draw (-0.3,0) node [left] {$Z'$};
\draw (-0.3,1.5) node [left] {$K'$};
\draw (0,2.5) node [left] {$b'$};
\draw (0,3) node [left] {$a'$};
\draw (1.8,0) node [right] {$Z''$};
\draw (1.8,1.5) node [right] {$K''$};
\draw (1.5,2.5) node [right] {$b''$};
\draw (1.5,3) node [right] {$a''$};
\draw (4.5,2.5)--(6,3);
\draw[line width=2] (4.5,2.5)--(6,1.5)--(4.5,1.5)--(6,3);
\draw[orange, line width=1] (4.5,1.5)--(6,0)--(4.5,0)--(6,1.5);
\draw[fill=white] (4.5,0) circle (0.4);
\draw[fill=white] (6,0) circle (0.4);
\draw[fill=white] (4.5,1.5) circle (0.4);
\draw[fill=white] (6,1.5) circle (0.4);
\draw[fill=white] (4.5,2.5) circle (0.13);
\draw[fill=white] (6,3) circle (0.13);
\draw (-0.3,0)++(4.5,0) node [left] {$Z'$};
\draw (-0.3,1.5)++(4.5,0) node [left] {$K'$};
\draw (0,2.5)++(4.5,0) node [left] {$c'$};
\draw (1.8,0)++(4.5,0) node [right] {$Z''$};
\draw (1.8,1.5)++(4.5,0) node [right] {$K''$};
\draw (1.5,3)++(4.5,0) node [right] {$c''$};
\draw[->,thick] (2.5,0.7)--(3.5,0.7);
\end{tikzpicture}
\caption{Operation (O2) applied to vertices $a,b$ in $H$ (left) and corresponding operation apllied to vertices $a',a'',b',b''$ in $H^*$ (right).}
\label{fig:decomp-operation-o2}
\end{center}
\end{figure}

Finally, we show an analogue of \cref{lem:decomposition}, where we explain how decompositions can be used algorithmically. Recall that $T(H,n,t)$ denotes an upper bound for the complexity of \lhomo{H} on instances with $n$ vertices, given along a tree decomposition of width $t$. In the following lemma we do not assume that $|H|$ is a constant.

\begin{lemma}[General decomposition lemma] \label{lem:general-decomposition}
Let $H$ be a connected, non-bi-arc graph, and let $\Gamma$ be a decomposition of $H$ (i.e., $\Gamma$ is either an $F$-, a $BP$-, or a $B$-decomposition) with factors $H_1,H_2$.
If there exist constants $c \geq 1$ and $d > 2$ such that $T(H_1,n,t)=\cO\left(c^t \cdot (n \cdot |H_1|)^d\right)$ and $T(H_2,n,t)=\cO\left(c^t \cdot (n \cdot |H_2|)^d\right)$, then $T(H,n,t)=\cO\left(c^t \cdot (n \cdot (|H|+2))^d\right)$.
\end{lemma}

Since every type of decomposition is handled in a slightly different way, for the sake of readability we split  \cref{lem:general-decomposition} into three lemmas, one for each type of decomposition. Their proofs are similar to the proof of \cref{lem:decomposition}.
We also take more care about constants to make sure that the estimation holds even if $H$ is not fixed.

\begin{lemma}[$F$-decomposition lemma] \label{lem:f-decomposition}
Let $H$ be a connected, non-bi-arc graph, and let $(F,K,Z)$ be an $F$-decomposition of $H$.
If there exist constants $\alpha,c \geq 1$ and $d > 2$ such that $T(H_1,n,t) \leq \alpha \cdot c^t \cdot (n \cdot |H_1|)^d$ and $T(H_2,n,t) \leq \alpha \cdot c^t \cdot (n \cdot |H_2|)^d$, then $T(H,n,t) \leq \alpha \cdot c^t \cdot (n \cdot |H|)^d$, provided that $n$ is sufficiently large.
\end{lemma}
\begin{proof} Let $(G,L)$ be an instance of \lhomo{H} with $n$ vertices, given along with a tree decomposition of with $t$.
Define $Q:=\{x \in V(G): L(x) \cap K \neq \emptyset \}$ and note that by \cref{prop:lists-general}~\cref{it:lists-general-incomparable} no vertex from $Q$ has a vertex from $F$ in its list.  
First, we observe that an analogue of \cref{clm:inside} in the proof of \cref{lem:decomposition} holds in this situation too (the proof is essentially the same as the proof of \cref{clm:inside}, so we do not repeat it).
\begin{claim}\label{cla:f-decomp-inside}
If there exists a list homomorphism $h \colon (G,L) \to H$, the image of each $C \in \concomp{G \setminus Q}$ is entirely contained either in $F$ or in $Z$. \hfill$\blacksquare$\end{claim}

For every $x \in V(G) - Q$ we define $L_1(x) := L(x) \cap F$ and for every component $C \in \concomp{G - Q}$ we check if there exists a homomorphism $h_C:(C, L_1) \to H_1$.
Let $\mathbb{C}_1$ be the set of those $C \in \concomp{G-Q}$, for which $h_C$ exists.

Now let us define an instance $(G, L_2)$ of \lhomo{H_2} as follows.
If $x \in Q$, then we set $L_2(x):= L(x)$, note that in this case $F \cap L(x) = \emptyset$.
On the other hand, if $x \notin Q$, then it is in exactly one $C \in \concomp{G-Q}$.
If $C \in \concomp{G \setminus Q} \setminus \mathbb{C}_1$, then we set $L_2(x):=L(x) \setminus F$.
If $C \in \mathbb{C}_1$, then we set $L_2(x):= L(x) \setminus F \cup \{f\}$. 

\begin{claim}
There exists a list homomorphism $h: (G,L) \to H$ if and only if there exists a list homomorphism $h':(G,L_2) \to H_2$.
\end{claim}
\begin{claimproof}
First, assume that $h:(G,L) \to H$ exists. We define $h': V(G) \to V(H_2)$ as follows:
\[h'(x):= \begin{cases}
f & \textrm{if }h(x) \in F, \\
h(x) & \textrm{otherwise}.
\end{cases}\]

It is straightforward to verify that $h'$ is a homomorphism from $G$ to $H_2$.
Let us show that it respects lists $L_2$. Indeed, consider a vertex $x \in V(G)$.
If $h(x) \notin F$, then $h'(x) = h(x) \in L(x) \setminus F \subseteq L_2(x)$.
On the other hand, if $h(x) \in F$, then $L(x) \cap F \neq \emptyset$, so $x$ is in some $C \in \concomp{G-Q}$. Since $h(x) \in F$, then by \cref{cla:f-decomp-inside} the image of $C$ is contained in $F$, so $C \in \mathbb{C}_1$.
Thus $h'(x) = f \in L_2(x)$.

Now assume that there exists a homomorphism $h':(G,L_2) \to H_2$.
We define $h: V(G) \to V(H)$ as follows. If $h'(x) \neq f$, then we set $h(x) := h'(x)$.
On the other hand, if $h'(x) = f$, then we know that $x$  belongs to some $C \in \mathbb{C}_1$. In this case we set $h(x):=h_C(x)$.
Note that $h$ respects lists $L$, as in the first case $h'(x)=h(x)$ clearly belongs to $(L_2(x) \setminus \{f\}) \subseteq L(x)$, and in the second one it follows from the definition of $h_C$. 

Now let us show that $h$ is a homomorphism. Consider and edge $xy \in E(G)$.
If $h'(x)=h'(y)=f$, then $x,y \notin Q$, so they must both belong to the same $C \in \mathbb{C}_1$ and $h(x)h(y)=h_C(x)h_C(y) \in E(H)$, by the definition of $h_C$.
Similarly, if $h'(x),h'(y) \in K \cup Z$, we have that $h(x)h(y)=h'(x)h'(y) \in E(H)$ by the definition of $h'$.

So suppose that $h'(x) =f$ and $h'(y) \in K \cup Z$. Note that since $h'$ is a homomorphism, we have $h'(y) \in K$. Recall that the set $K$ is complete to $F$, so $h(x) \in F$ is adjacent to $h(y)=h'(y) \in K$.
\end{claimproof}

Computing $Q$ and $\concomp{G-Q}$ can be done in total time $\Oh(n|H| + n^2)=\Oh(( n \cdot |H|)^2)$. 
Computing $h_C$ for all $C \in \concomp{G \setminus Q}$ requires time at most
\[
\sum_{C \in \concomp{G \setminus Q}} T(H_1,|C|,t) \leq \sum_{C \in \concomp{G \setminus Q}} \alpha \cdot c^t \cdot  (|H_1| \cdot |C|)^d \leq \alpha \cdot c^t \cdot  (|H_1| \cdot n)^d.
\]
Computing lists $L_2$ can be performed in time $\Oh(|H| \cdot n)$. Finally, computing $h'$ requires time $T(H_2,n,t) \leq  \alpha \cdot c^t \cdot  (|H_2| \cdot n)^d$.
Similarly to the proof of \cref{lem:decomposition},  if $n$ is sufficiently large, then the total running time can be bounded by $\alpha \cdot c^t \cdot  (|H| \cdot n)^d$, as claimed.
\end{proof}

The case of a $BP$-decomposition is slightly more complicated, as there might be vertices of $G$, whose lists contain both vertices from $P$ (which is a part of $H_1$) and from $M$ (which is a part of the separator). These vertices require special handling.

\begin{lemma}[$BP$-decomposition lemma] \label{lem:bp-decomposition}
Let $H$ be a connected, non-bi-arc graph, and let $(B,P,M,K,Z)$ be a $BP$-decomposition of $H$ with factors $H_1,H_2$.
If there exist constants $\alpha,c \geq 1$ and $d > 2$ such that $T(H_1,n,t) \leq \alpha \cdot c^t \cdot (n \cdot |H_1|)^d$ and $T(H_2,n,t) \leq \alpha \cdot c^t \cdot (n \cdot |H_2|)^d$, then $T(H,n,t) \leq \alpha \cdot c^t \cdot (n \cdot |H|)^d$, provided that $n$ is sufficiently large.
\end{lemma}
\begin{proof} Let $(G,L)$ be an instance of \lhomo{H} with $n$ vertices, given along with a tree decomposition of with $t$.
We define $Q:=\{x \in V(G): L(x) \cap (M \cup K) \neq \emptyset\}$.
We say that a vertex $x \in V(G)$ is of \emph{type $B$} (resp. \emph{type $P$}) if $L(x) \cap B \neq \emptyset$ (resp. $L(x) \cap P \neq \emptyset$). Note that by \cref{prop:lists-general}~\cref{it:lists-general-incomparable} no vertex is simultaneously of type $B$ and type $P$. 
Moreover, no vertex of type $B$ is in $Q$. Let $Q'$ be the set of vertices of type $P$ in $Q$. Note that by \cref{prop:lists-general}~\cref{it:lists-general-incomparable} each vertex from $Q'$ has in its list some vertices from $M$ and some vertices from $P$, but no vertices from $K$.

Notice that now we cannot just consider components of $G-Q$ independently from each other and from the vertices from $Q$, as we did in the case of an $F$-decomposition.
This is because the coloring we choose for these components influences the coloring of vertices from $Q'$, which in turn influences the coloring of some other components of $G-Q$.
Thus, instead of considering components of $G-Q$ separately, we gather them together with their neighbors from $Q'$ and only then compute the homomorphisms to $H_1$. This intuition is formalized by the following argument. 

Let $\mathbb{S}_1,\mathbb{S}_2$, and $\mathbb{S}$ be  families of sets defined as follows \[\mathbb{S}_1 := \bigcup_{C \in \concomp{G-Q}} \Bigl \{ V(C) \Bigr \}  \quad \text{ and } \quad  \mathbb{S}_2 := \bigcup_{x \in Q'} \Bigl \{\{x\} \Bigr \} \quad \text{ and } \quad \mathbb{S} := \mathbb{S}_1 \cup \mathbb{S}_2.\]
Let $\sim \; \subseteq \mathbb{S}^2$ be a relation, defined as follows: $S_1 \sim S_2$ if and only if (i) $S_1 \in \mathbb{S}_1$, i.e., it is a vertex set of some $C \in \concomp{G-Q}$ and (ii) $S_2 \in \mathbb{S}_2$, i.e., $S_2= \{x\}$ for some $x \in Q'$, and (iii) $x$ is adjacent to some vertex of type $B$ from $S_1$.
Let $\simeq \; \subseteq \mathbb{S}^2$ be the smallest equivalence relation containing $\sim$ and consider the set $\mathbb{S} /_\simeq$ of equivalence classes of $\simeq$.

Now we form the set $\mathbb{C}'$ as follows. Consider a class $\pi \in \mathbb{S}/_\simeq$, it is a subset of $\mathbb{S}_1 \cup \mathbb{S}_2$, i.e., a family of pairwise disjoint subsets of $(V(G) \setminus Q) \cup Q'$.
We construct a subgraph $C'$ of $G$, whose vertex set is the union of all sets in $\pi$.
An edge $xy \in E(G)$ exists in $C'$ if its both endvertices belong to one element of $\pi$ (i.e., it is an edge of some $C \in \concomp{G-Q}$, such that $V(C) \in \pi$), or is an edge joining some $x \in Q'$, such that $\{x\} \in \pi$, with a vertex $y$, which is a vertex of type $B$ in $C$, such that $V(C) \in \pi$.
Note that the edges of the second kind are exactly the edges which justifed including $\{x\}$ and $V(C)$ in the same equivalence class $\pi \in \mathbb{S}/_\simeq$. Moreover, each $C' \in \mathbb{C}'$ is connected.
Finally, we define $\mathbb{C} := \{G[V(C')] : C' \in \mathbb{C}'\}$, observe that it is a set of pairwise disjoint induced subgraphs of $G$.

The following analogue of \cref{clm:inside} and \cref{cla:f-decomp-inside} is a crucial step towards proving the lemma.

\begin{claim}\label{cla:bp-decomp-inside}
If there exists a list homomorphism $h \colon (G,L) \to H$, the image of each $C \in \mathbb{C}$ is entirely contained either in $P \cup B$ or in $M \cup Z$.
\end{claim} 
\begin{claimproof}
Consider $C \in \mathbb{C}$ and assume that there exist $x,y \in V(C)$ such that $h(x) \in P \cup B$ and $h(y) \in M \cup Z$.
Recall that there exists $C' \in \mathbb{C}'$, which is a connected spanning subgraph of $C$.
Let $\cP$ be a $x$-$y$-path in $C'$.
Recall that no vertex of $C$ (and thus of $\cP$) has a vertex from $K$ in its list. This implies that on $\cP$ there exists an edge $zz'$, such that $h(z) \in P$ and $h(z') \in M$.
Clearly, $z' \in Q'$. Moreover, $z \notin Q'$, as in $C'$ there are no edges between vertices of $Q'$. So $z$ must be a vertex of type $P$ from $V(G) \setminus Q$.
However, in $C'$ there are no edges between $Q'$ and vertices of type $P$ in $V(G) \setminus Q$, a contradiction.
\end{claimproof}

Let $C \in \mathbb{C}$. For every $x \in V(C)$ we define $L_1(x) := L(x) \cap (P \cup B)$.
Denote by $\mathbb{C}_1$ the set of these subgraphs $C \in \mathbb{C}$, for which there exists a list homomorphism $h_C: (C,L_1) \to H_1$. We define an instance $(G, L_2)$ of \lhomo{H_2} as follows:
\begin{align*}
L_2(x):=\begin{cases}
L(x) \setminus P \cup \{p\} & \textrm{ if }x \text{ is of type $P$ and belongs to some subgraph in } \mathbb{C}_1,\\
L(x) \setminus B \cup \{b\} & \textrm{ if }x \text{ is of type $B$ and belongs to some subgraph in } \mathbb{C}_1,\\
L(x) \setminus (P \cup B) & \textrm{ otherwise}.
\end{cases}
\end{align*}
Recall that no vertex $x$ is simultaneously of type $P$ and $B$, so $L_2$ is well-defined. Now let us show the following.
\begin{claim}
There exists a list homomorphism $h: (G,L) \to H$ if and only if there exists a list homomorphism $h':(G,L_2) \to H_2$.
\end{claim}
\begin{claimproof}
First, assume that $h: (G,L) \to H$ exists. Define $h': V(G) \to V(H_2)$ in the following way:
\[
h'(x) := 
\begin{cases}
p & \text{ if } h(x) \in P,\\
b & \text{ if } h(x) \in B,\\
h(x) & \text{ otherwise.}
\end{cases}
\]
The fact that $h'$ is a homomorphism follows from the definition of $H_2$ and the fact that $h$ is a homomorphism. We need to show for every $x$ we have $h'(x) \in L_2(x)$.

If $h'(x) \notin \{p,b\}$, then $h'(x) = h(x) \in L(x) \setminus (P \cup B) \subseteq L_2(x)$.
So assume that $h'(x) \in \{p,b\}$. This means that $h(x) \in P \cup B$, so $x \in (V(G) \setminus Q) \cup Q'$. This means that $x$ belongs to some $C \in \mathbb{C}$, and since $h(x) \in P \cup B$, by \cref{cla:bp-decomp-inside}, we observe that $C \in \mathbb{C}_1$.
Moreover, if $h'(x) = p$, then $x$ is of type $P$, and if $h'(x)=b$, then $x$ is of type $B$. Therefore $h'(x) \in L_2(x)$.

Now suppose there exists a list homomorphism $h' : (G,L_2) \to H_2$.  We define a mapping $h: V(G) \to V(H)$ as follows.
If $h'(x) \notin \{p,b\}$, then $h(x) := h'(x)$.
Otherwise, by the definition of $L_2$, the vertex $x$ must belong to exactly one subgraph $C$ in $\mathbb{C}_1$. In this case we set $h(x):=h_C(x)$.
Note that for every $x \in V(G)$ we have $h(x) \in L(x)$: if $h'(x) \notin \{p,b\}$, then $h(x)=h'(x) \in L_2(x) \setminus \{p,b\} \subseteq L(x)$, and if $h'(x) \in \{p,b\}$, then $h(x) \in L(x)$ by the definition of $h_C$.

Now let us show that $h$ is a homomorphism. Let $xy$ be an edge of $G$ and consider the following cases.
\begin{compactenum}
\item If $h'(x),h'(y) \notin \{p,b\}$, then $h(x)h(y)=h'(x)h'(y) \in E(H_2 - \{p,b\})  \subseteq E(H)$, by the fact that $h'$ is a homomorphism.
\item If $h'(x)=b$ and $h'(y)=p$, then $x$ is of type $B$ and is in $V(G) \setminus Q$, while $y$ is of type $P$ and belongs to $(V(G) \setminus Q) \cup Q'$.
Recall that by the definition of $\sim$, vertices $x,y$ belong to the same subgraph $C \in \mathbb{C}$.
Moreover, $C \in \mathbb{C}_1$, as only in this case $p$ and $b$ could appear on the lists of $x,y$.
This gives $h(x)h(y)=h_C(x)h_C(y) \in E(H_1) \subseteq E(H)$, by the definition of $h_C$. 
\item The case $h'(x)=h'(y)=b$ cannot occur, as $h'$ is a homomorphism and $b$ has no loop.
\item If $h'(x)=h'(y)=p$, then $p \in L_2(x) \cap L_2(y)$, so both $x,y$ are of type $P$, and each of them belongs to some subgraph in $\mathbb{C}_1$. Let us call these subgraphs $C_x$ and $C_y$, respectively, and observe that $h(x) = h_{C_x}(x)$ and $h(y)=h_{C_y}(y)$.
As $h_{C_x}$ and $h_{C_y}$ are homomorphisms to $H_1$, we know that $h_{C_x}(x),h_{C_y}(y)\in P\cup B$. But since neither $L(x)$ or $L(y)$ contains elements from $B$, we have $h_{C_x}(x),h_{C_y}(y)\in P$.
This means that $h(x)h(y) = h_{C_x}(x)h_{C_y}(y)$ is an edge of $H$, as $P$ is a reflexive clique.
\item If $h'(x)=b$ and $h'(y) \notin \{p,b\}$, then $h'(y)=h(y) \in K$, as $h'$ maps $xy$ to an edge of $H_2$.
Again, let $C$ be the subgraph from $\mathbb{C}_1$ which contains $x$ and observe that $x$ must be of type $B$, as $b \in L_2(x)$. So $h(x)=h_C(x) \in B$, which implies that $h(x)h(y) \in E(H)$, as $B$ is complete to $K$.
\item If $h'(x)=p$ and $h'(y) \notin \{p,b\}$, then $h'(y)=h(y) \in K\cup M$. Again, let $C$ be the subgraph in $\mathbb{C}_1$, which contains $x$ and observe that $p$ is of type $P$. So $h(x)=h_C(x) \in P$, which implies that $h(x)h(y) \in E(H)$, as $P$ is complete to $K \cup M$.
\end{compactenum}
This concludes the proof of the claim.
\end{claimproof}

The running time analysis is analogous to the case of an $F$-decomposition.
Note that $\mathbb{S}/_\simeq$ and thus $\mathbb{C}$ can be computed in total time $\Oh(n \cdot |H|)^2$.
With a reasoning analogous to the one in the proof of \cref{lem:f-decomposition}, we conclude that if $n$ is sufficiently large, then the problem can be solved in time $\alpha \cdot c^t \cdot (n \cdot |H|)^d$.
\end{proof}

The final lemma, concerning $B$-decompositions, is the most complicated one.
Since in this case the first factor, $H_1$, is bipartite, we know that every subgraph $C$ of $G$ that is mapped to $H_1$ must be bipartite too. However, unlike in the case when the whole graph $H$ is bipartite (i.e., \cref{lem:decomposition}), we do not know in advance which bipartition class of $C$ is mapped to which bipartition class of $H_1$. Thus for every $C$ we need to consider both possibilities.

\begin{lemma}[$B$-decomposition lemma] \label{lem:b-decomposition}
Let $H$ be a connected, non-bi-arc graph, and let $(B_1,B_2,K,M_1,M_2,Z)$ be a $B$-decomposition of $H$.
If there exist constants $\alpha,c,d$ such that $T(H_1,n,t) \leq \alpha/2 \cdot c^t \cdot (n \cdot |H_1|)^d$ and $T(H_2,n,t) \leq \alpha \cdot c^t \cdot (n \cdot |H_2|)^d$, then $T(H,n,t) \leq \alpha \cdot c^t \cdot (n \cdot |H|)^d$, provided that $n$ is sufficiently large.
\end{lemma}
\begin{proof} Let $(G,L)$ be an instance of \lhomo{H} with $n$ vertices, given along with a tree decomposition of with $t$.
We define $Q:=\{x \in V(G): L(x) \cap (M_1 \cup M_2 \cup K) \neq \emptyset\}$.
Furthermore we define two subsets of $Q$ as follows: $Q_1 := \{x \in Q : L(x) \cap B_1 \neq \emptyset\}$ and $Q_2 := \{x \in Q : L(x) \cap B_2 \neq \emptyset\}$.
Note that by \cref{prop:lists-general}~\cref{it:lists-general-incomparable}, if $L(x)$ contains a vertex from $K$, then $x \notin Q_1 \cup Q_2$. Furthermore, if $L(x)$ contains a vertex from $M_1$, then $x \notin Q_1$, and if $L(x)$ contains a vertex from $M_2$, then $x \notin Q_2$. This means that $Q_1$ is disjoint with $Q_2$.
Moreover, every vertex from $Q_1$ has in its list a vertex from $B_1$ and a vertex from $M_2$, while every vertex from $Q_2$ has in its list a vertex from $B_2$ and a vertex from $M_1$.

Similarly to e.g.~\cref{cla:f-decomp-inside}, we observe that every connected component $C$ of $G-Q$ must be entirely mapped either to $B_1 \cup B_2$ or to $Z$. Recall that $H_1 = H[B_1 \cup B_2]$ is bipartite, so if $C$ is not bipartite, then we can safely remove all vertices from $B_1 \cup B_2$ from the lists of vertices of $C$, let us still call these lists $L$.
Let $S_1$ be the set of bipartite components of $G-Q$.

Now let us consider an auxiliary graph $G'$ defined as follows. The vertex set of $G'$ is $Q_1 \cup Q_2$. We add to $G'$ all edges of $G$ with one end in $Q_1$ and the other in $Q_2$.
Observe that for each component $C$ of $G'$ and any list homomorphism $(G,L) \to H$, either all vertices of $C$ are mapped to $B_1 \cup B_2$ (with $V(C) \cap Q_1$ mapped to $B_1$ and $V(C) \cap Q_2$ mapped to $B_2$), or all vertices of $C$ are mapped to $M_1 \cup M_2 \cup Z$ (with $V(C) \cap Q_1$ mapped to $M_2 \cup Z$ and $V(C) \cap Q_2$ mapped to $M_1 \cup Z$).
Again, recall that $H[B_1 \cup B_2] = H_1$ is bipartite. Thus if for any $C \in \concomp{G'}$, the graph $G[V(C)]$ is not bipartite, then we can safely remove all vertices from $B_1 \cup B_2$ from the lists of the vertices of $C$. We still call these lists $L$.
Let $S_2$ be the set of those $C \in \concomp{G'}$, for which $G[V(C)]$ is bipartite. Note that in this case $C = G[V(C)]$.

Observe that $S_1 \cup S_2$ is a collection of pairwise vertex-disjoint bipartite induced subgraphs of $G$.
Furthermore, no vertex appearing in a subgraph from $S_1 \cup S_2$ has a vertex from $K$ in its list.
Moreover, we observe that for any list homomorphism $(G,L) \to H$, each $C \in S_1 \cup S_2$ must be entirely mapped to $B_1 \cup B_2$ or to $M_1 \cup M_2 \cup Z$.
This is straightforward for $C \in S_1$, as these vertices do not contain any vertex from $K \cup M_1 \cup M_2$ in their lists.
For $C \in S_2$, we can observe that in the other case we would obtain $uv \in E(C)$, such that
$u$ is mapped to a vertex in $B_1$ and $v$ is mapped to a vertex in $M_2$, or
$u$ is mapped to a vertex in $B_2$ and $v$ is mapped to a vertex in $M_1$. However, notice that for no edge of $C$ such configuration of lists is possible.

Let $C \in S_1 \cup S_2$ and suppose it gets mapped to $B_1 \cup B_2$. If $C \in S_2$, then we know which bipartition class of $C$ is mapped to $B_1$ and which is mapped to $B_2$, i.e., $V(C) \cap Q_1$ is mapped to $B_1$ and $V(C) \cap Q_2$ is mapped to $B_2$. On the other hand, if $C \in S_1$, both options might be possible.

We define a set $\mathbb{S}_1$ as follows: for each $C \in S_1$ with bipartition classes $X,Y$, we add to $\mathbb{S}_1$ two triples: $(C,X,Y)$ and $(C,Y,X)$.
Similarly we define a set $\mathbb{S}_2$: for each $C \in S_2$, we add to $\mathbb{S}_2$ the triple $(C, V(C) \cap Q_1, V(C) \cap Q_2$). Finally, define $\mathbb{S} := \mathbb{S}_1 \cup \mathbb{S}_2$.

The intuition behind $\mathbb{S}$ is that it represents the subgraphs that could be mapped to $B_1 \cup B_2$, taking into considerations that some of them (i.e., the ones from $S_1$) could be mapped to $B_1 \cup B_2$ in two ways. A triple $(C,X,Y) \in\mathbb{S}$ indicates that $C$ could be potentially mapped to $B_1 \cup B_2$ in a way that $X$ is mapped to $B_1$ and $Y$ is mapped to $B_2$.

Now we proceed similarly to the proof of \cref{lem:bp-decomposition}. We define a relation $\sim \; \subseteq \mathbb{S}^2$, so that $(C_1,X_1,Y_1) \sim (C_2,X_2,Y_2)$ if and only if the following conditions hold: (i) $(C_1,X_1,Y_1) \in \mathbb{S}_1$, and (ii) $(C_2,X_2,Y_2) \in \mathbb{S}_2$, and (iii) there is an edge with one endvertex in $X_1$ and the other in $Y_2$, or an edge with one endvertex in $Y_1$, and the other in $X_2$.
Note that if $(C_1,X_1,Y_1) \sim (C_2,X_2,Y_2)$ and there is a list homomorphism $h \colon (G,L) \to H$, then either:
(i) both $C_1$ and $C_2$ are mapped to $B_1 \cup B_2$ where $h(X_1 \cup X_2) \subseteq B_1$ and $h(Y_1 \cup Y_2) \subseteq B_2$, or
(ii) no vertex from $X_1 \cup X_2$ is mapped to $B_1$ and no vertex from $Y_1 \cup Y_2$ is mapped to $B_2$.
This is because lists of vertices from $X_1 \cup Y_1$ do no contain any vertex from $K \cup M_1 \cup M_2$, while lists of vertices of $X_2$ (resp. $Y_2$) do not contain any vertex from $M_1 \cup B_2 \cup K$ (resp. $M_2 \cup B_1 \cup K$). 

Now let $\simeq$ be the smallest equivalence relation containing $\sim$, and let $\mathbb{S}/_\simeq$ be the set of equivalence classes of $\simeq$. We will define a set $\mathbb{C}$ of subgraphs of $G$.

Consider an equivalence class $\pi = \{ (C_1,X_1,Y_1),\ldots,(C_s,X_s,Y_s) \}$ from $\mathbb{S}/_\simeq$, and define $X_\pi := \bigcup_{i \in [s]} X_i$ and $Y_\pi := \bigcup_{i \in [s]} Y_i$. 
Let $C_\pi$ be the subgraph of $G$ induced by $X_\pi \cup Y_\pi$.
Recall that by the definition of $\sim$, either 
(i) all vertices from $X_\pi$ must be mapped to $B_1$ and all vertices from $Y_\pi$ must be mapped to $B_2$, or
(ii) no vertex from $X_\pi$ is mapped to $B_1$ and no vertex from $Y_\pi$ is mapped to $B_2$.
Thus if $C_\pi$ is not bipartite, then we can safely remove all vertices from $B_1$ from the lists of $X_\pi$ and all vertices from $B_2$ from the lists of $Y_\pi$, we still call these lists $L$.

On the other hand, if $C_\pi$ is bipartite, we note that its bipartition classes are exactly $X_\pi$ and $Y_\pi$. We add to $\mathbb{C}$ the triple $(C_\pi,X_\pi,Y_\pi)$.
Now $\mathbb{C}$ can be seen as a collection of bipartite induced subgraphs of $G$, containing only vertices from $(V(G) \setminus Q) \cup (Q_1 \cup Q_2)$, such that each vertex from $V(G) \setminus Q$ appears in at most two of them, while every vertex from $Q_1 \cup Q_2$ appears in one of them.

We claim that at the current step the following holds.
\begin{claim} \label{clm:b-decomp-b-inC}
If $x$ is a vertex such that $L(x) \cap B_1 \neq \emptyset$ (resp. $L(x) \cap B_2 \neq \emptyset$), then there exists $(C,X,Y) \in \mathbb{C}$, such that $x \in X$ (resp. $x \in Y$).
\end{claim}
\begin{claimproof}
Let us show the claim for the case that $L(x) \cap B_1 \neq \emptyset$, as the other case is symmetric.
Note that since $x$ has a vertex from $B_1$ in its list, then either $x \in Q_1$ or $x \in V(G) \setminus Q$. 

In the first case, note that there is exactly one $(C',X',Y') \in \mathbb{S}_2$, such that $x \in X' \cup Y'$, and, since $x \in Q_1$, we know that $x \in X'$. Consider the class $\pi$ in $\mathbb{S}/_\simeq$ containing $(C',X',Y')$, and the triple $(C_\pi,X_\pi,Y_\pi)$. Clearly $x \in X_\pi$.
Note that if $C_\pi$ is not bipartite, we would have removed all vertices of $B_1$ from the list of $x$. But since $L(x) \cap B_1 \neq \emptyset$, the graph $C_\pi$ must be bipartite and thus $(C_\pi,X_\pi,Y_\pi)$ is the desired triple in $\mathbb{C}_1$.

So suppose that $x \in V(G) \setminus Q$, and therefore $x$ is in some bipartite connected component $C'$ of $G - Q$. Let $X',Y'$ be the bipartition classes of $C'$, such that $x \in X'$. Recall that both triples $(C',X',Y')$ and $(C',Y',X')$ are in $\mathbb{S}_1$. Let $\pi$ and $\pi'$ be the equivalence classes from $\mathbb{S}/_\simeq$, containing $(C',X',Y')$ and $(C',Y',X')$, respectively.

Observe that if $\pi = \pi'$, then $x \in X_\pi \cap Y_\pi$ and $C_\pi$ is not bipartite. Indeed, the definition of $\sim$ implies that there is an odd length walk starting and terminating at $x$. So in this case we have removed all vertices from $B_1 \cup B_2$ from the list of $x$.

So finally assume that $\pi \neq \pi'$ and consider $(C_{\pi},X_{\pi},Y_{\pi})$ and $(C_{\pi'},X_{\pi'},Y_{\pi'})$. Note that $x \in X_\pi \cap Y_{\pi'}$.
If $C_\pi$ is not bipartite, then we would have removed all vertices of $B_1$ from the list of $x$. Therefore $C_\pi$ must be bipartite and $(C_{\pi},X_{\pi},Y_{\pi}) \in \mathbb{C}$ is the desired triple.
\end{claimproof}

The following claim follows from the definition of  $\mathbb{C}$ and the reasoning there.

\begin{claim}\label{cla:b-decomp-inside}
Let $(C,X,Y) \in \mathbb{C}$ and suppose that there exists a list homomorphism $h \colon (G,L) \to H$.
\begin{compactenum}[(a)]
\item The image of $C$ is entirely contained in $B_1 \cup B_2$ or in $M_1 \cup M_2 \cup Z$. \label{it:b-inside-either}
\item If  the image of $C$ is contained in $B_1 \cup B_2$, then either $h(X) \subseteq B_1$ and $h(Y) \subseteq B_2$ or $h(X) \subseteq B_2$ and $h(Y) \subseteq B_1$.\hfill$\blacksquare$ \label{it:b-inside-xy}
\end{compactenum}
\end{claim} 

Let $(C,X,Y) \in \mathbb{C}$. For every $x \in X$ we define $L_1(x) := L(x) \cap B_1$ and for every $x \in Y$ we define $L_1(x):=L(x) \cap B_2$.
Denote by $\mathbb{C}_1$ the set of these triples $(C,X,Y) \in \mathbb{C}$, for which there exists a homomorphism $h_{(C,X,Y)}: (C,L_1) \to H_1$.

We define an instance $(G, L_2)$ of \lhomo{H_2} as follows:
\begin{align*}
L_2(x):=\begin{cases}
L(x) \setminus (B_1 \cup B_2) \cup \{b_1,b_2\} & \textrm{if there is some $(C_1,X_1,Y_1) \in \mathbb{C}_1$, such that $x \in X_1$}\\
								& \textrm{and some $(C_2,X_2,Y_2) \in \mathbb{C}_1$, such that $x \in Y_2$},\\
L(x) \setminus (B_1 \cup B_2) \cup \{b_1\} & \textrm{if there is some $(C_1,X_1,Y_1) \in \mathbb{C}_1$, such that $x \in X_1$,}\\
								& \textrm{but no $(C_2,X_2,Y_2) \in \mathbb{C}_1$, such that $x \in Y_2$},\\
L(x) \setminus (B_1 \cup B_2) \cup \{b_2\} & \textrm{if there is some $(C_2,X_2,Y_2) \in \mathbb{C}_1$, such that $x \in Y_2$,}\\
								& \textrm{but no $(C_1,X_1,Y_1) \in \mathbb{C}_1$, such that $x \in X_1$},\\
L(x) \setminus (B_1 \cup B_2) & \textrm{otherwise}.
\end{cases}
\end{align*}
Note that $L_2$ are $H_2$-lists.
Furthermore, by \cref{clm:b-decomp-b-inC} we observe that $L_2(x) \cap \{b_1,b_2\} = \emptyset$ if (i) $L(x)  \cap (B_1 \cup B_2) = \emptyset$, or (ii) $L(x)  \cap (B_1 \cup B_2) \neq \emptyset$, but for every $(C,X,Y) \in \mathbb{C}$, such that $x \in X \cup Y$, we have $(C,X,Y) \notin \mathbb{C}_1$.
This means that removing $B_1 \cup B_2$ from the list of $x$ was justified.

Finally, we are ready to show the main claim of the lemma.

\begin{claim}
There exists a list homomorphism $h: (G,L) \to H$ if and only if there exists a list homomorphism $h':(G,L_2) \to H_2$.
\end{claim}
\begin{claimproof}
First, assume that $h: (G,L) \to H$ exists. Define $h': V(G) \to V(H)$ in the following way:
\[
h'(x) := 
\begin{cases}
b_1 & \text{ if } h(x) \in B_1,\\
b_2 & \text{ if } h(x) \in B_2,\\
h(x) & \text{ otherwise.}
\end{cases}
\]
The fact that $h'$ is a homomorphism follows from the definition of $H_2$ and the fact that $h$ is a homomorphism. We need to show that for every $x$ we have $h'(x) \in L_2(x)$.

If $h'(x) \notin \{b_1,b_2\}$, then $h'(x) = h(x) \in L(x) \setminus (B_1 \cup B_2) \subseteq L_2(x)$. So consider the case that $h'(x) =b_1$ (the case if $h'(x) = b_2$ is symmetric).
This means that $h(x) \in B_1$ and thus $L(x) \cap B_1 \neq \emptyset$, so, by \cref{clm:b-decomp-b-inC}, there is some $(C,X,Y) \in \mathbb{C}$, such that $x \in X$. 
Since $h(x) \in B_1$, \cref{cla:b-decomp-inside}~\cref{it:b-inside-either} implies that $h$ maps all vertices of $C$ to $B_1 \cup B_2$,
and by \cref{cla:b-decomp-inside}~\cref{it:b-inside-xy} we have that $h(X) \subseteq B_1$ and $h(Y) \subseteq B_2$.
Therefore $(C,X,Y) \in \mathbb{C}_1$ and thus $b_1 \in L_2(x)$ by the definition of $L_2$.

Now suppose there exists a list homomorphism $h' : (G,L_2) \to H_2$.  We define a mapping $h: V(G) \to V(H)$ as follows.
If $h'(x) \notin \{b_1,b_2\}$, then $h(x) := h'(x)$.
If $h'(x) \in  \{b_1,b_2\}$, then, by the definition of $L_2$, there is $(C,X,Y) \in \mathbb{C}_1$, such that $x \in X$ if $h'(x) = b_1$ or $x \in Y$ if $h'(x)=b_2$. Moreover, recall that by the construction of $\mathbb{C}$ this $(C,X,Y)$ is unique. Then we set $h(x) = h_{(C,X,Y)}(x)$. 

For every $x \in V(G)$ we have $h(x) \in L(x)$: if $h'(x) \notin \{b_1,b_2\}$, then $h(x)=h'(x) \in L_2(x) \setminus \{b_1,b_2\} \subseteq L(x)$, and if $h'(x) \in \{b_1,b_2\}$, then $h(x) \in L(x)$ by the definition of $h_{(C,X,Y)}$.

So let us show that $h$ is a homomorphism. Consider an edge $xy$ of $G$ and the following cases.
\begin{compactenum}
\item If $h'(x),h'(y) \notin \{b_1,b_2\}$, then $h(x)h(y)=h'(x)h'(y) \in E(H_2 - \{b_1,b_2\}) \subseteq E(H)$, by the fact that $h'$ is a homomorphism.

\item If $h'(x)=b_1$ and $h'(y)=b_2$, then, by the definition of $L_2$, there are $(C_1,X_1,Y_1) \in \mathbb{C}_1$ and $(C_2,X_2,Y_2) \in \mathbb{C}_1$, such that $x \in X_1$ and $y \in Y_2$. 
But since there is an edge from $X_1$ to $Y_2$, we must have $(C_1,X_1,Y_1) = (C_2,X_2,Y_2)$, by the definition of $\sim$ and $\mathbb{C}$.
So we have  $h(x)h(y)=h_{(C_1,X_1,Y_1)}(x)h_{(C_1,X_1,Y_1)}(y) \in E(H_1) \subseteq E(H)$.

\item The case $h'(x)=h'(y)=b_1$ or $h'(x)=h'(y)=b_2$ cannot occur, as $h'$ is a homomorphism and $b_1,b_2$ have no loops.

\item Finally, consider the case that $h'(x) \in \{b_1,b_2\}$ and $h'(y) \notin \{b_1,b_2\}$.
Assume that $h'(x) = b_1$, as the other case is symmetric. This implies that $h(x) \in B_1$.
Since $h'$ is a homomorphism, we observe that $h(y)=h'(y) \in K \cup M_2$. Therefore $h(x)h(y) \in E(H)$.
\end{compactenum}
This concludes the proof of the claim.
\end{claimproof}

We observe that $\mathbb{C}$ can be computed in total time $\Oh((n \cdot |H|)^2)$. Computing $h_{(C,X,Y)}$ for all $(C,X,Y) \in \mathbb{C}$ can be done in total time:
\[
\sum_{(C,X,Y) \in \mathbb{C}} T(H_1,|C|,t) \leq \sum_{(C,X,Y) \in \mathbb{C}} \alpha/2 \cdot c^t \cdot (|H_1| \cdot |C|)^d \leq 2 \left( \alpha/2 \cdot c^t \cdot (|H_1| \cdot n)^d \right) = \alpha \cdot c^t \cdot (|H_1| \cdot n)^d,
\]
as every vertex of $G$ might appear twice in subgraphs in $\mathbb{C}$. The rest of the argument is exactly the same as in the case of an $F$-decomposition and a $BP$-decomposition. Finally we conclude that if $n$ is sufficiently large, then the problem can be solved in time $\alpha \cdot c^t \cdot (n \cdot |H|)^d$.
\end{proof}

\subsection{Special case: strong split graphs}
As we have seen, in \cref{lem:decompositions-equivalent} we have assumed that $H$ is not a strong split graph.
This case is somehow special, as if $H$ is a strong split graph, then the straightforward bipartite decomposition of $H^*$ does not correspond to any of decompositions defined in the preceding section.
We will consider these graphs $H$ separately.

\begin{lemma} \label{lem:strongsplit}
Let $H$ be a strong split graph with partition $(B,P)$.
Let $H'$ be the graph obtained from $H$ by removing all edges with both endvertices in $P$, including loops.
The \lhomo{H} problem with instance $(G,L)$ with $n$ vertices, given along with a tree decomposition of width $t$, can be solved in time $\Oh\left(i^*(H')^t \cdot (n \cdot |H|)^{\Oh(1)}\right)$.
\end{lemma}

\begin{proof}
Let $X$ (resp. $Y$) be the set of vertices $x \in V(G)$, such that $L(x) \cap B \neq \emptyset$ (resp. $L(x) \cap P \neq \emptyset$).
Observe that for every $p \in P$ and $b \in B$ we have $N(b) \subseteq N(p)$, so, by \cref{prop:lists-general}~\cref{it:lists-general-incomparable}, the sets $X$ and $Y$ are disjoint. Since we can assume that there are no vertices of $G$ with empty lists, we conclude that $X,Y$ is a partition of $V(G)$.

Furthermore, notice that if $X$ is not independent, then we can immediately report that $(G,L)$ is a no-instance. So assume that $X$ is independent.
Let $G'$ be the graph obtained from $G$ by removing all edges with both endvertices in $Y$, including loops.
Clearly $H'$ and $G'$ are bipartite, with bipartition classes $B,P$ and $X,Y$, respectively.
Furthermore, any tree decomposition of $G$ is also a tree decomposition of $G'$.

Recall that (i) $P$ is a reflexive clique and (ii) $L(X) \subseteq B$ and $L(Y) \subseteq P$. Furthermore, (iii) in order to obtain $G'$ and $H'$ we have removed all edges with two endvertices inside $Y$ and $P$, respectively.
These three facts imply the following.

\begin{claim}
Consider a function $h : V(G) \to V(H)$, such that for every $x \in V(G)$ it holds that $h(x) \in L(x)$.
Then $h$ is a homomorphism from $G$ to $H$ if and only if it is a homomorphism from $G'$ to $H'$. \hfill $\blacksquare$
\end{claim}

Thus we can determine whether $(G,L) \to H$ by calling the algorithm for bipartite target graphs, given by \cref{thm:main-bipartite-algo}.
\end{proof}

In the following corollary we bound the running time of the algorithm from \cref{lem:strongsplit} in terms of the original target graph.

\begin{corollary} \label{cor:strongsplit} 
Let $\widehat{H}$ be a graph and let $H$ be a strong split graph, which was obtained from $\widehat{H}$ by a series of decompositions.
The \lhomo{H} problem with instance $(G,L)$ with $n$ vertices, given along with a tree decomposition of width $t$, can be solved in time $\Oh(i^*(\widehat{H})^t \cdot (n \cdot |H|)^{\Oh(1)})$.
\end{corollary}
\begin{proof}
Let $H'$ be defined as in \cref{lem:strongsplit}.
Note that if $H'$ is the complement of a circular-arc graph, then the problem can be solved in polynomial time, so the claim clearly holds. Otherwise, let $H''$ be the connected, induced, undecomposable subgraph of $H'$, whose complement is not a circular-arc graph, such that $i(H'')=i^*(H')$.

Recall from \cref{lem:induced-hstar} that $H^*$ is an induced subgraph of $\widehat{H}^*$.
Observe that $H^*$ consists of two copies of $H'$, induced by the sets $B' \cup P''$ and $B'' \cup P'$, with additional edges joining every vertex from $P'$ with every vertex from $P''$.
Thus $H'$ is an induced subgraph of $H^*$, and therefore  $H''$ is an induced subgraph of $\widehat{H}^*$.
This means that $i^*(\widehat{H})=i^*(\widehat{H}^*) \geq i(H'') = i^*(H')$.
Therefore the algorithm from \cref{lem:strongsplit} solved $(G,L)$ in time $\Oh\left(i^*(H')^t \cdot (n \cdot |H|)^{\Oh(1)}\right)=\Oh\left(i^*(\widehat{H})^t \cdot (n \cdot |H|)^{\Oh(1)}\right)$.
\end{proof}

\subsection{Solving \lhomo{H} for general target graphs}

By \cref{lem:decompositions-equivalent}, it is straightforward to observe the following.

\begin{observation} \label{obs:i-star-undecomp}
If $H$ is a connected, undecomposable, non-bi-arc graph, then $i^*(H)$ is the size of the largest incomparable set in $H$. \qedhere
\end{observation}

In this section we will prove the following, slightly stronger version of \cref{thm:main} a), where the input tree decomposition is not assumed to be optimal.

\begin{cthm}{5' a)} \label{thm:main-algo} 
Let $H$ be non-bi-arc graph. 
Even if $H$ is given as an input, the $\lhomo{H}$ problem with instance $(G,L)$ can be solved in time $\Oh \left(i^*(H)^{t} \cdot (n\cdot |H|)^{\Oh(1)}\right)$ for any lists $L$, provided that $G$ is given with its tree decomposition of width $t$. \end{cthm}

The main idea is similar to the one in the proof of \cref{thm:main-bipartite-algo}: given an instance of \lhomo{H}, we recursively decompose $H$ into smaller subgraphs and reduce the initial instance to a number of instances of list homomorphism to these smaller subgraphs. Finally, we solve the problem for leaves of the recursion tree, and then, using \cref{lem:general-decomposition} in a bottom-up fashion, we will compute the solution to the original instance.

The only thing missing is how to solve the instances corresponding to leaves of the recursion tree. We describe this in the following lemma.

\begin{lemma}\label{lem:leaves-general}
Let $H$ be an arbitrary graph. Then any $n$-vertex instance $(G,L)$ of \lhomo{H} can be solved in time $\Oh\left(i(H^*)^t \cdot (n \cdot |H|)^{\Oh(1)}\right)$, assuming a tree decomposition of $G$ with width at most $t$ is given.
\end{lemma}
\begin{proof}
Let $(G^*,L^*)$ be the associated instance of \lhomo{H^*}.
By \cref{prop:associated-equiv}, we know that $(G,L) \to H$ if and only if there is a clean homomorphism $(G^*,L^*) \to H^*$. We will focus on finding such a clean homomorphism.

First, recall that by \cref{prop:lists-general}~\cref{it:lists-general-incomparable} and \cref{obs:associated-properties}~\cref{it:associated-twins}, the instance $(G^*,L^*)$ is consistent. So, by \cref{cor:listsize}, the size of each list in $L^*$ is at most $i(H^*)$.
Moreover, by \cref{obs:associated-properties}~\cref{it:associated-twins}, for every $x \in V(G)$, the vertices in $L^*(x')$ are precisely the twins of vertices in $L^*(x'')$.
Finally, by \cref{obs:associated-properties}~\cref{it:associated-treedecomp}, in polynomial time we can obtain a tree decomposition $\cT^*$ of $G^*$ with width at most $2t$, in which vertices of $G^*$ come in pairs: whenever any bag contains $x'$, it also contains $x''$.

Consider the straightforward dynamic programming algorithm for \lhomo{H^*}, using the tree decomposition $\cT^*$ of $G^*$. We observe that since we are looking for a clean homomorphism, we do not need to remember partial solutions, in which the colors of twins do not agree. Thus, even though the size of each bag of $\cT^*$ is at most $2t$, the number of partial colorings we need to consider is bounded by $(\max_{x \in V(G^*)}|L^*(x)|)^t \leq i(H^*)^t$.
So the claim follows.
\end{proof}

Finally, let us wrap everything up and prove \cref{thm:main-algo}. The proof is analogous to the proof of \cref{thm:main-bipartite-algo}.
\begin{proof}[Proof of \cref{thm:main-algo}]
Again, we can assume that $n$ is sufficiently large, as otherwise we can solve the problem by brute-force.

Let $(G,L)$ be an instance of \lhomo{H}, where $G$ has $n$ vertices and is given along with its tree decomposition of width at most $t$.
We proceed as in the \cref{thm:main-bipartite-algo}. We consider a recursion tree $\cR$, obtained by decomposing $H$ recursively. For each node corresponding to some graph $H'$, we construct its children recursively, unless none of the following happens (i) $H'$ is a bi-arc graph, (ii) $H'$ is bipartite, (iii) $H'$ is a strong split graph, or (iv) $H'$ is undecomposable.

We compute the solutions in a bottom-up fashion.
First, consider a leaf of the recursion tree, let the corresponding target graph for this node of $\cR$ be $H'$.
If $H'$ is a bi-arc graph, we can solve the problem in polynomial time.
If $H'$ is bipartite, we solve the problem in time $\beta \cdot i^*(H') \cdot (n \cdot |H'|)^{d_1} \leq \beta \cdot i^*(H) \cdot (n \cdot |H'|)^{d_1}$ for some constants $\beta$ and $d_1$, using the algorithm from \cref{thm:main-bipartite-algo}.
If $H'$ is a strong split graph, we can solve the problem in time $\gamma \cdot i^*(H) \cdot (n \cdot |H'|)^{d_2}$, for constants $\gamma,d_2$, using the algorithm from \cref{cor:strongsplit}.

So finally consider the remaining case, i.e., that $H'$ is connected, non-bi-arc, non-bipartite, undecomposable, which is not a strong split graph. Furthermore we know that $H'$ was obtained from $H$ by a sequence of decompositions.

Recall that by \cref{lem:leaves-general}, we can solve the instances of \lhomo{H'} in time $\delta \cdot i(H'^*)^t \cdot (n \cdot |H'|)^{d_3}$, for some constants $\delta,d_3$.
Let us consider  the graph $H'^*$, by \cref{lem:induced-hstar} we know that $H'^*$ is an induced subgraph of $H^*$.
Also, $H'^*$ is either connected (if $H'$ is non-bipartite), or consists of two disjoint copies of $H'$ (if $H'$ is bipartite). Moreover, by \cref{lem:decompositions-equivalent}, we observe that $H'^*$ is undecomposable.
Thus, by the definition of $i^*(H)$, we observe that
\begin{align*}
i(H'^*) \leq & \max \{ i(H'') \colon H''
 \text{ is an undecomposable, connected,  induced subgraph of }H^*, \\
 & \text{whose complement is not a circular-arc graph}\}  =  i^*(H).
\end{align*} 
So the algorithm from \cref{lem:leaves-general} solves the instances corresponding to leaves of $\cR$ in time $\delta \cdot i^*(H)^t \cdot (n \cdot |H'|)^{d_3}$.
Define $\alpha := \max(2\beta,\gamma,\delta)$ and $d := \max(d_1,d_2,d_3,3)$.

By applying \cref{lem:general-decomposition} (in fact, \cref{lem:f-decomposition}, \cref{lem:bp-decomposition}, and \cref{lem:b-decomposition}, note that we adjusted the constants $\alpha$ and $d$ so that their assumptions are satisfied) for every non-leaf node of $\cR$ in a bottom-up fashion, we conclude that we can solve \lhomo{H} in time $\alpha \cdot i^*(H)^t \cdot (n \cdot |H|)^d = \Oh\left(i^*(H)^t \cdot (n \cdot |H|)^{\Oh(1)}\right)$, which completes the proof.
\end{proof}

\newpage
\section{Hardness for general target graphs}
We aim to show the following. Note that again we are able to show the hardness parameterized by the pathwidth. Actually, we can show that the lower bound holds even if the input graph $G$ is bipartite.

\begin{cthm}{5' b)} \label{thm:main-hard-pw}
Let $H$ be a fixed non-bi-arc graph.
Unless the SETH fails, there is no algorithm that solves the \lhomo{H} problem on bipartite instances with $n$ vertices and pathwidth $t$ in time $(i^*(H)-\epsilon)^t \cdot n^{\Oh(1)}$, for any $\epsilon >0$.
\end{cthm}

The crucial observation is the following.

\begin{proposition} \label{prop:bipartite-associted} 
Let $H$ be a graph and let $(G,L)$ be a consistent instance of \lhomo{H^*}.
Define $L' \colon V(G) \to 2^{V(H)}$ as $L'(x) := \{u \colon \{u',u''\} \cap L(x) \neq \emptyset\}$.
Then $(G,L) \to H^*$ if and only if $(G,L') \to H$.
\end{proposition}
\begin{proof}
First, consider $f \colon (G,L) \to H^*$. For every $x \in V(G)$, we define $f'(x)$ to be the unique vertex $u$ of $H$, such that $f(x) \in \{u',u''\}$. Note that since $\{u',u''\} \cap L(x) \neq \emptyset$, we have $u \in L'(x)$. Now consider an edge $xy$ of $G$. Since $f$ is a homomorphism, we have $f(x)f(y) \in E(H^*)$. Without loss of generality assume that $f(x) =u'$ and $f(y) = v''$ for some some edge $uv$ of $H$ (possibly $u=v$). But then $f'(x)f'(y) = uv \in E(H)$, so $f'$ is a homomorphism from $(G,L')$ to $H$.

Now consider a homomorphism $f' \colon (G,L') \to H$. Since $(G,L)$ is consistent, we know that $G$ is bipartite with bipartition classes $X_G$ and $Y_G$, where $L(X_G) \subseteq \{ u' : u \in V(H)\}$ and $L(Y_G) \subseteq \{ u'' : u \in V(H)\}$.
Let $x \in V(G)$ and $f'(x) = u$. Then we set $f(x) = u'$ if $x \in X_G$ or $f(x) = u''$ if $x \in Y_G$.
First, let us show that $f$ respects the lists $L$. Consider a vertex $x \in X_G$ (the case of a vertex in $Y_G$ is symmetric). Since $f'(x) =u$, we observe that $u \in L'(x)$, which means that $u' \in L(x)$. So $f$ respects lists $L$.
Now consider an edge $xy$ of $G$, such that $x \in X_G$ and $y \in Y_G$. Assume that $f'(x) =u$ and $f'(y) =v$, where $uv$ is an edge of $H$ (possibly $u=v$). Then $f(x)f(y)= u'v''$, which is an edge of $H^*$. So $f$ is a homomorphism from $G$ to $H^*$.
\end{proof}

Clearly \cref{thm:main-hard-pw} implies \cref{thm:main-hard-bipartite-pw}. We will show that the reverse implication also holds. Recall that we have already shown that \cref{thm:main-hard-bipartite-pw} is equivalent to \cref{thm:hardness-bipartite}, so in fact we will prove that all these three theorems are equivalent.

\paragraph{(\cref{thm:main-hard-bipartite-pw} $\to$ \cref{thm:main-hard-pw})} 
Suppose \cref{thm:main-hard-bipartite-pw} holds and \cref{thm:main-hard-pw} fails. 
So there is a non-bi-arc graph $H$ and an algorithm $A$ that solves $\lhomo{H}$ in time $(i^*(H) - \epsilon)^{\pw{G}} \cdot n^{\Oh(1)}$ for every bipartite input $(G,L)$, provided that $G$ is given along with its optimal path decomposition.

Recall that $i^*(H) = i^*(H^{*})$. Let $(G,L)$ be an arbitrary instance of \lhomo{H^*}. We can assume that $G$ is bipartite and connected, and the instance $(G,L)$ is consistent.
Consider the instance $(G,L')$ of \lhomo{H}, constructed as in \cref{prop:bipartite-associted}. The algorithm $A$ solves this instance in time $(i^*(H) - \epsilon)^{\pw{G}} \cdot n^{\Oh(1)}$. By \cref{prop:bipartite-associted}, this is equivalent to solving the instance $(G,L)$ of \lhomo{H^*} in time $(i^*(H^*) - \epsilon)^{\pw{G}} \cdot n^{\Oh(1)}$, contradicting \cref{thm:main-hard-bipartite-pw}.

\newpage
\section{Conclusion}
Let us conclude the paper with some side remarks and pointing out several open problems.

\subsection{Special cases: reflexive and irreflexive graphs}

Recall that the crucial tool for our algorithmic results were the decompositions of a connected graph $H$, introduced in \cref{sec:decomposition} (for bipartite graphs $H$) and in \cref{sec:general-decompositions} (for general graphs $H$).
Let us analyze how the general decompositions behave in two natural special cases: if $H$ is either a reflexive or an irreflexive graph.
We use the notation introduced in \cref{def:f-decomposition}, \cref{def:bp-decomposition}, and \cref{def:b-decomposition}.

First we consider the case that $H$ is reflexive, i.e., every vertex of $H$ has a loop. 
Let us point out that a $B$-decomposition cannot occur in this case, as the sets $B_1,B_2$ are empty. In case of an $F$-decomposition we obtain exactly the decomposition defined by Egri {\em et al.}~\cite[Lemma 8]{DBLP:conf/stacs/EgriMR18}. Finally, in the case of a $BP$-decomposition, note that the set $B$ is empty and therefore each vertex in $P$ has exactly the same neighborhood. Thus the total contribution of the vertices in $P$ to $i^*(H)$ is at most 1. Therefore the only type of decomposition that can be algorithmically exploited in reflexive graph is the $F$-decomposition, as observed by Egri {\em et al.}~\cite{DBLP:conf/stacs/EgriMR18}.

Now let us consider the case that $H$ is irreflexive, i.e., no vertex of $H$ has a loop. Observe that the sets $K$ and $P$ are reflexive cliques, so they are empty in our case. Thus $BP$-decompositions and $F$-decompositions do not occur in this case (recall that $H$ is connected). Therefore the only possibility left is a $B$-decomposition, in which the set $K$ is empty.
Let us point out that this decomposition is very similar to the bipartite decomposition, in particular, the graph $H_1$ is bipartite (while $H_2$ might be non-bipartite). This gives even more evidence that the case of bipartite graphs $H$ is a crucial step to understanding the complexity of the \lhomo{H} problem.
Actually, if $H$ is bipartite, then the $B$-decomposition turns out to be equivalent to the bipartite decomposition introduced in \cref{sec:decomposition}. The argument used in \cref{lem:decomposition} is significantly simpler than the one in \cref{lem:b-decomposition}, as in the bipartite case we can assume that the instance is consistent and we know which vertices will be mapped to $B_1$, and which to $B_2$. In the general case we do not know it, so we need to consider both possibilities, as we do in \cref{lem:b-decomposition}.

\subsection{Typical graphs $H$}

Knowing the precise complexity bounds for \lhomo{H}, we might be interested in a question, how hard is to find a list homomorphism to a \emph{typical} graph $H$.
We say that a property $P$ \emph{holds for almost all graphs}, if for a graph $G$, chosen with uniform probability from the set of all graphs with $n$ vertices, the probability that $G$ satisfies $P$ tends to 1 as $n \to \infty$. Properties that hold for almost all graphs can be studied in terms of the \emph{random graph $G(n,1/2)$}. Formally, it is a probability space  over all graph with vertex set $[n]$, where the edge $ij$ exists with probability $1/2$. Here we extend the usual model by allowing loops on vertices (also with probability $1/2$), but this does not change much in the reasoning. See the monograph of Alon and Spencer for more information~\cite{DBLP:books/daglib/0021015}.

It is well-known and straightforward to observe that for any fixed graph $\Ob$, almost all graphs contain $\Ob$ as an induced subgraph. By applying this for, say, $\Ob = C_6$, we obtain that almost all graphs are non-bi-arc graphs.
Also, it is known that almost all graphs are connected~\cite{Bollobas,moon1965almost}.
Finally, let us consider incomparable sets in a typical graph $H$.

\begin{lemma} \label{lem:almost-all-incomparable}
For almost all graphs $H$, every pair of distinct vertices is incomparable.
\end{lemma}
\begin{proof}
Consider a random graph with vertex set $[n]$, where each edge (including loops) appears independently from others with probability $1/2$. For each $i,j \in [n]$ (non-necessarily distinct) we introduce a random variable $I_{i,j}$, whose value is $1$ if and only if the edge $ij$ exists. Clearly $N(i) \subseteq N(j)$ if and only if $I_{j,k} \geq I_{i,k}$ for every $k \in [n]$.

For every $i \neq j$ and every $k$, the probability that $I_{j,k} \geq I_{i,k}$ is $3/4$. Thus, for fixed $i \neq j$, the probability that $N(i) \subseteq N(j)$ is $(3/4)^n$.
So the probability that there is a pair $i,j$, such that $N(i) \subseteq N(j)$, is at most $n(n-1) (3/4)^{n}$, which tends to 0 as $n \to \infty$. This means that with probability tending to $1$ all pairs of distinct vertices are incomparable.
\end{proof}

Note that this in particular implies that almost all graphs are undecomposable, as each decomposition necessarily contains pairs of comparable vertices. Similarly, almost all graphs are not strong split graphs.
Since for connected, undecomposable graphs $H$, the value of $i^*(H)$ it is equal to the size of the largest incomparable set in $H$ (recall \cref{obs:i-star-undecomp}), we conclude that for almost all graphs $H$ we have $i^*(H) = |V(H)|$.

All these observations lead to the following corollary (note that again we can strengthen the statement given in the introduction by replacing the treewidth by the pathwidth).

\begin{ccor}{6'}
For almost all graphs $H$ with possible loops the following holds.
Even if $H$ is fixed, there is no algorithm that solves  $\lhomo{H}$  for every instance $(G,L)$  on $n$ vertices in time $\Oh((|V(H)|-\epsilon)^{\pw{G}} \cdot n^{\Oh(1)})$ for any $\epsilon > 0$, unless the SETH fails.
\end{ccor}

\subsection{Generalized algorithm}

We believe that the decompositions that we discovered can be used for many problems concerning the complexity of variants of the \lhomo{H} problem, e.g., for various other parameterizations.

First, let us point out that in the proofs of \cref{lem:decomposition} and \cref{lem:general-decomposition},
we did not really require that the running time is of the form $\Oh (c^{\tw{G}} \cdot (n \cdot |H|)^d)$, for a constant $c$.
Actually, it is sufficient that the bound on the running time is convex and non-decreasing.
Let us start with a technical lemma.

\begin{lemma}\label{lem:convex}
Let $f \colon [1,\infty) \to \R$ be a non-decreasing, convex function.
Then for any $x,y \in [1,\infty)$ it holds that $f(x) + f(y) \leq f(x+y) + f(1)$.
\end{lemma}
\begin{proof}
First, observe that for each $z \geq 1$ and for each $\theta \in [0,1]$, we have
\[
f(\theta z) \leq f(\theta z + (1-\theta)) \leq \theta f(z) + (1-\theta) f(1).
\]
Consider $x,y \in [1,\infty)$. Define $\theta := \frac{x}{x+y}$ and observe that $\theta \in [0,1]$. So we have the following.
\begin{align*}
f(x) + f(y) = & f(\theta (x+y)) + f((1-\theta)(x+y)) \\
\leq & \left( \theta f(x+y) + (1-\theta) f(1) \right) + \left( (1-\theta) f(x+y) + \theta f(1) \right)\\
= & f(x+y) + f(1).
\end{align*}
\end{proof}

Next, recall that in each variant of the decomposition lemma, all instances for which we computed partial results were induced subgraphs of $G$, the original instance.
This motivates the following, generalized statement, which can be used as a black-box to prove algorithmic and hardness results about \lhomo{H}.

\begin{theorem}\label{lem:very-general-decomposition}
Let $H$ be a graph.
In time $|H|^{\Oh(1)}$ we can construct a family $\cH$ of $\Oh(|H|)$ connected graphs, such that:
\begin{compactenum}[(1)]
\item $H$ is a bi-arc graph if and only if every $H' \in \cH$ is a bi-arc graph, \label{factors:hard}
\item if $H$ is bipartite, then each $H' \in \cH$ is an induced subgraph of $H$, and is either the complement of a circular-arc graphs or is undecomposable, \label{factors:bipartite}
\item otherwise, for each $H' \in \cH$, the graph $H'^*$ is an induced subgraph of $H^*$ and at least one of the following holds: \label{factors:types}
\begin{enumerate}
\item $H'$ is a bi-arc graph, or \label{factors:easy}
\item $H'$ a strong split graph and has an induced subgraph $H''$, which is not a bi-arc graph and is an induced subgraph of $H$, or \label{factors:types:strongsplit}
\item $(H')^*$ is undecomposable, \label{factors:types:undecomposable}
\end{enumerate}
\item for every instance $(G,L)$ of \lhomo{H} with $n$ vertices, the following implication holds: \label{factors:bottomup}

If there exists a non-decreasing, convex function $f_H \colon \N \to \R$,
such that for every $H' \in \cH$, for every induced subgraph $G'$ of $G$, and for every $H'$-lists $L'$ on $G'$,
we can decide whether $(G',L')\to H'$ in time $f_H(|V(G')|)$, then
we can solve the instance $(G,L)$  in time
\[
\Oh \left (|H| f_H(n) + n^2 \cdot |H|^3 \right).
\]
\end{compactenum}
\end{theorem}
\begin{proof}[Sketch of proof]
The proof is analogous to the proofs of \cref{thm:main-bipartite-algo} and \cref{thm:main-algo}, so we only provide the sketch.
Let us construct the recursion tree $\cR$ of $H$ as we did in \cref{thm:main-algo}. However, we do not terminate once we reach a bipartite graph, but in this case we continue  as in \cref{thm:main-bipartite-algo}.
Let $\cH$ the the set of graphs associated with the leaves of $\cR$.

Let us discuss property \cref{factors:hard}. Since we  never decompose strong split graphs, \cref{lem:decompositions-equivalent} implies that 
a decomposition of any graph $F$ considered in the construction of $\cR$ corresponds to a bipartite decomposition of $F^*$.
Furthermore, recall that a graph $F'$ is a bi-arc graph if and only if $F'^*$ is the complement of a circular-arc graph.
Thus in order to show \cref{factors:hard} it is sufficient to show that if $F'$ is bipartite and has a bipartite decomposition with factors $F_1$ and $F_2$, such that both $F_1,F_2$ are the complements of circular-arc graphs, then $F$ is the complement of a circular-arc graph. 
It can be verified using the co-circular-arc representations, introduced in \cref{sec:hardness-prelim}, see \cref{fig:bi-arc-decomposition}.

\begin{figure}[h]
\begin{tabular}{m{6cm}m{2cm}m{3cm}}
\begin{tikzpicture}
\def\r{1}
\def\s{2.5}
	\fill[fill=yellow!80] (\s,\s) -- (\s,\s) +(137-180:\s+0.6) arc (137-180:263-180:\s+0.6) -- (\s,\s);
	\fill[fill=yellow!80] (\s,\s) -- (\s,\s) +(137:\s+0.6) arc (137:263:\s+0.6) -- (\s,\s);
	\fill[fill=white] (\s,\s) circle (\s);
	\node at (0,5) {$F_1$};
	\fill[fill=yellow!80] (\s,\s) -- (\s,\s) +(26:\r+1) arc (26:-14:\r+1) -- (\s,\s);
	\fill[fill=yellow!80] (\s,\s) -- (\s,\s) +(167:\r+1) arc (167:207:\r+1) -- (\s,\s);
	\fill[fill=white] (\s,\s) circle (\r-0.3);
	\draw[dashed,gray!50] (\s, 2*\s+0.9) -- (\s,-.9);
    \node at (\s-0.2,-.7) {\scriptsize{$q$}};
    \node at (\s-0.2, 2*\s+0.7) {\scriptsize{$p$}};
    
    \draw[green,line width=1.2] (\s,\s) +(82:\s+0.1) arc (82:260:\s+0.1);
    \draw[green,line width=1.2] (\s,\s) +(-12:\s+0.3) arc (-12:140:\s+0.3);
    \draw[green,line width=1.2] (\s,\s) +(78:\s+0.4) arc (78:167:\s+0.4);
    \draw[green,line width=1.2] (\s,\s) +(30:\s+0.5) arc (30:150:\s+0.5);
    \draw[green,line width=1.2] (\s,\s) +(82-180:\s+0.2) arc (82-180:260-180:\s+0.2);
	\draw[green,line width=1.2] (\s,\s) +(-12-180:\s+0.3) arc (-12-180:140-180:\s+0.3);
    \draw[green,line width=1.2] (\s,\s) +(78-180:\s+0.4) arc (78-180:167-180:\s+0.4);
    \draw[green,line width=1.2] (\s,\s) +(30-180:\s+0.5) arc (30-180:150-180:\s+0.5);

	\node at (2.2,\r+1.8) {$F_2$};
    
    \draw[red,line width=1.2] (\s,\s) +(54:\r+0.5) arc (54:154:\r+0.5);
    \draw[red,line width=1.2] (\s,\s) +(26:\r+0.6) arc (26:140:\r+0.6);
    \draw[red,line width=1.2] (\s,\s) +(82:\r+0.7) arc (82:167:\r+0.7);
    \draw[red,line width=1.2] (\s,\s) +(30:\r+0.8) arc (30:150:\r+0.8);
    
    \draw[red,line width=1.2] (\s,\s) +(54-180:\r+0.5) arc (54-180:154-180:\r+0.5);
    \draw[red,line width=1.2] (\s,\s) +(26-180:\r+0.6) arc (26-180:140-180:\r+0.6);
    \draw[red,line width=1.2] (\s,\s) +(82-180:\r+0.7) arc (82-180:167-180:\r+0.7);
    \draw[red,line width=1.2] (\s,\s) +(30-180:\r+0.8) arc (30-180:150-180:\r+0.8);
    
    \draw[green,line width=1.2] (\s,\s) +(5:\r+0.9) arc (5:155:\r+0.9);
    \draw[green,line width=1.2] (\s,\s) +(5-180:\r+0.9) arc (5-180:155-180:\r+0.9);
    
    \draw[gray,line width=1.2] (\s,\s) +(40-180:\r-0.2) arc (40-180:-222-180:\r-0.2);
    \draw[gray,line width=1.2] (\s,\s) +(-12-180:\r) arc (-12-180:-240-180:\r);
    \draw[gray,line width=1.2] (\s,\s) +(-114-180:\r+0.2) arc (-114-180:67-180:\r+0.2);
    \draw[gray,line width=1.2] (\s,\s) +(42-180:\r+0.4) arc (42-180:-130-180:\r+0.4);
    
    \draw[gray,line width=1.2] (\s,\s) +(40:\r-0.1) arc (40:-222:\r-0.1);
    \draw[gray,line width=1.2] (\s,\s) +(-12:\r+0.1) arc (-12:-240:\r+0.1);
    \draw[gray,line width=1.2] (\s,\s) +(-114:\r+0.3) arc (-114:67:\r+0.3);
    \draw[gray,line width=1.2] (\s,\s) +(42:\r+0.4) arc (42:-130:\r+0.4);

\end{tikzpicture}
& \hskip 1cm $\Rightarrow$ \hskip 1cm &
\begin{tikzpicture}
\def\r{1}
	\fill[fill=yellow!80] (\r,\r) -- (\r,\r) +(26:\r+1.5) arc (26:-14:\r+1.5) -- (1,1);
	\fill[fill=yellow!80] (\r,\r) -- (\r,\r) +(167:\r+1.5) arc (167:207:\r+1.5) -- (1,1);
	\fill[fill=white] (\r,\r) circle (\r-0.3);
	\node at (-1.5,\r+1.3) {$F$};
    \node at (\r-0.2,-1.6) {\scriptsize{$q$}};
    \node at (\r-0.2,2*\r+1.6) {\scriptsize{$p$}};
    \draw[dashed,gray!50] (\r, 2*\r+1.8) -- (\r,-1.8);
    
    \draw[red,line width=1.2] (\r,\r) +(54:\r+0.6) arc (54:154:\r+0.6);
    \draw[red,line width=1.2] (\r,\r) +(26:\r+0.7) arc (26:140:\r+0.7);
    \draw[red,line width=1.2] (\r,\r) +(82:\r+0.8) arc (82:167:\r+0.8);
    \draw[red,line width=1.2] (\r,\r) +(30:\r+0.9) arc (30:150:\r+0.9);
    
    \draw[red,line width=1.2] (\r,\r) +(54-180:\r+0.6) arc (54-180:154-180:\r+0.6);
    \draw[red,line width=1.2] (\r,\r) +(26-180:\r+0.7) arc (26-180:140-180:\r+0.7);
    \draw[red,line width=1.2] (\r,\r) +(82-180:\r+0.8) arc (82-180:167-180:\r+0.8);
    \draw[red,line width=1.2] (\r,\r) +(30-180:\r+0.9) arc (30-180:150-180:\r+0.9);
    
    \draw[gray,line width=1.2] (\r,\r) +(40-180:\r-0.2) arc (40-180:-222-180:\r-0.2);
    \draw[gray,line width=1.2] (\r,\r) +(-12-180:\r-0.1) arc (-12-180:-240-180:\r-0.1);
    \draw[gray,line width=1.2] (\r,\r) +(-114-180:\r) arc (-114-180:67-180:\r);
    \draw[gray,line width=1.2] (\r,\r) +(42-180:\r+0.1) arc (42-180:-130-180:\r+0.1);
    
    \draw[gray,line width=1.2] (\r,\r) +(40:\r+0.2) arc (40:-222:\r+0.2);
    \draw[gray,line width=1.2] (\r,\r) +(-12:\r+0.3) arc (-12:-240:\r+0.3);
    \draw[gray,line width=1.2] (\r,\r) +(-114:\r+0.4) arc (-114:67:\r+0.4);
    \draw[gray,line width=1.2] (\r,\r) +(42:\r+0.5) arc (42:-130:\r+0.5);
	
	\draw[green,line width=1.2] (1,\r) +(23:\r+1) arc (23:202:\r+1);
    \draw[green,line width=1.2] (1,\r) +(1:\r+1.2) arc (1:170:\r+1.2);
    \draw[green,line width=1.2] (1,\r) +(20:\r+1.3) arc (20:178.75:\r+1.3);
    \draw[green,line width=1.2] (1,\r) +(11.5:\r+1.4) arc (11.5:171.5:\r+1.4);
    
	\draw[green,line width=1.2] (1,\r) +(23-180:\r+1.1) arc (23-180:202-180:\r+1.1);
    \draw[green,line width=1.2] (1,\r) +(1-180:\r+1.2) arc (1-180:170-180:\r+1.2);
    \draw[green,line width=1.2] (1,\r) +(20-180:\r+1.3) arc (20-180:178.75-180:\r+1.3);
    \draw[green,line width=1.2] (1,\r) +(11.5-180:\r+1.4) arc (11.5-180:171.5-180:\r+1.4);
\end{tikzpicture}
\end{tabular}
\caption{Assume that $F$ is bipartite and has a bipartite decomposition $(D,N,R)$ with factors $F_1,F_2$, which are both complements of circular-arc graphs. Fix some co-circular-arc representations of $F_1$ and $F_2$ (left). In the representation of $F_2$ (left, interior), arcs corresponding to sets $\{d_X,d_Y\},$ $N$, and $R$ are indicated respectively by green, red, and gray color. Observe that (i) since $N$ is a biclique, there must be some space between ``north'' and ``south'' red arcs (yellow), (ii) since $N$ is a separator, all gray arcs must intersect the opposite green arc, and (iii) since $N$ is bipartite-complete to $\{d_X,d_Y\}$, no green arc can intersect opposite red arcs.
So we can obtain a representation of $F$ by ``squeezing'' the non-trivial part of the representation of $F_1$ in the yellow area. } 
\label{fig:bi-arc-decomposition}
\end{figure}
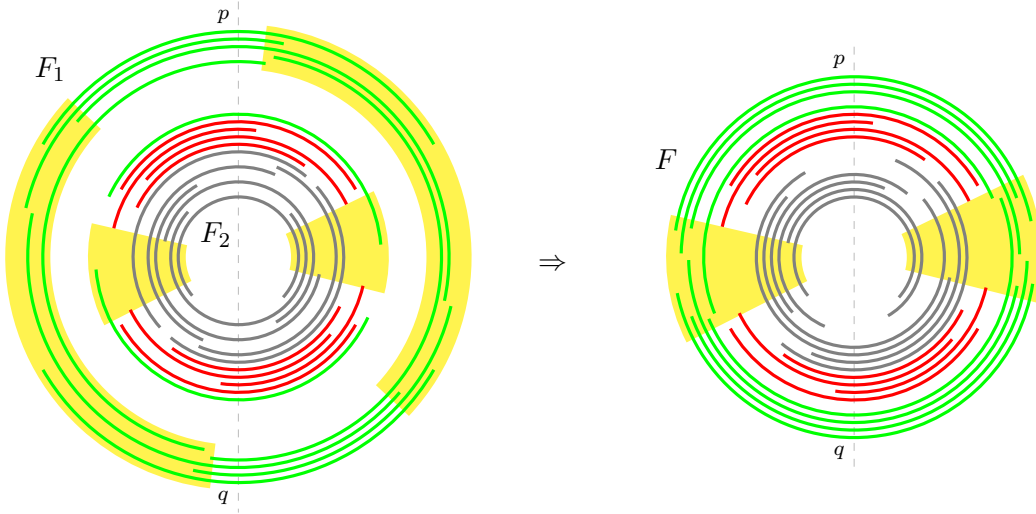

Property \cref{factors:bipartite} follows directly from the definition of the bipartite decomposition.
Property \cref{factors:types} follows from the definitions of decompositions, \cref{lem:decompositions-equivalent}, and \cref{lem:induced-hstar}.
The only possible issue is the additional condition in \cref{factors:types:strongsplit}.
Suppose after a series of decomposition applied to $H$ we obtain a strong split graph $H'$, which is not a bi-arc graph (otherwise we are in case \cref{factors:easy}).
Recall that the only possible vertices of $H'$ that are not originally in $H$ have loops, and all their neighbors also have loops (recall the $F$-decomposition). Thus such vertices belong to the reflexive clique in $H'$ and have no neighbors outside this clique.
In order to prove property \cref{factors:types:strongsplit} one needs to verify that the graph $H''$ obtained from $H'$ by removing all such vertices is still a non-bi-arc graph. This can be shown using the restricted structure of strong split graphs and the property shown in the previous paragraph: if some bipartite graph $F$ has a bipartite decomposition such that both factors are complements of circular-arc graphs, then $F$ is a complement of a circular-arc graph. The details are left to the reader.

Now let us discuss property \cref{factors:bottomup}.
Recall that we can solve the instances corresponding to the leaves of $\cR$, and then proceed in a bottom-up fashion.
Even though there might be more than one instance corresponding to a leaf node, recall that the numbers of vertices of such instances sum up to $n$.
So, consider a leaf node of $\cR$, associated with some $H' \in \cH$, and let $n_1,n_2,\ldots,n_p$ be the numbers of vertices in the instances of \lhomo{H'} created in the recursion. By \cref{lem:convex} we can bound the running time related to this leaf node by
\[
f_H(n_1) + f_H(n_2) + \ldots + f_H(n_p) \leq f_H(n_1+ n_2) + f_H(1) + f_H(n_3) + \ldots f_H(n_p) \leq \ldots \leq f_H(n) + nf_H(1).
\]

Now let us consider an internal node of $\cR$, it is associated with some subgraph $H'$ of $H$.
Recall that the computation for this node consists of the computations for child nodes and $\Oh(n^2 |H'|^2) = \Oh(n^2 |H|^2)$ additional steps.
Since the total number of nodes of $\cR$ is $\Oh(|H|)$, we can bound the total running time for the root node (i.e., the time needed to solve $(G,L)$) by
\begin{align*}
\Oh\left( \sum_{H' \in \cH} \left(f_H(n) + nf_H(1) \right) + n^2 \cdot |H|^3\right) &= \Oh\left( |H| f_H(n) + |H|nf_H(1) + n^2 \cdot |H|^3\right) \\ 
&= \Oh\left( |H| f_H(n)+ n^2 \cdot |H|^3\right),
\end{align*}
which completes the proof.
\end{proof}

Let us point out that if we apply \cref{lem:very-general-decomposition} to an irreflexive graph, then the case \cref{factors:types:strongsplit} does not occur. Indeed, the only strong split graph that is irreflexive has no edges, and, in particular,
is a bi-arc graph. Thus $H''$ with properties given in \cref{lem:very-general-decomposition}~~\cref{factors:types:strongsplit} cannot exist.

\subsection{Gadgets for general graphs}

Recall that we proved hardness for general graphs $H$ by reducing the problem to the bipartite case. In particular, all our gadgets, including the most important $\rneq(S)$-gadget, are constructed for bipartite graphs only.
However, we believe that in some other contexts it might be useful to be able to obtain such gadgets also for non-bipartite $H$.
Luckily, our constructions can be simply translated to this case too. Observe that the proof of \cref{prop:bipartite-associted} gives the following, slightly stronger version.

\begin{proposition}\label{prop:general-gadgets}
Let $H$ be an undecomposable non-bi-arc graph.
Let $G$ be a bipartite graph with bipartition classes $X,Y$ and $H^*$-lists $L$, where $L(X) \subseteq V(H)'$ and $L(Y) \subseteq V(H)''$.
Define $L' :V(G) \to 2^{V(H)}$ as $L'(v) = \{u \colon \{u',u''\} \cap L(v) \neq \emptyset\}$.
Then the following hold:
\begin{compactenum}
\item For $f \colon (G,L) \to H^*$, define $f' :V(G) \to V(H)$ by setting $f'(v)$ to be the unique vertex $u$ of $H$, such that $f(v) \in \{u',u''\}$. Then $f'$ is a list homomorphism from $(G,L')$ to $H$.
\item For $f' \colon (G,L') \to H$, define $f :V(G) \to V(H^*)$ by setting $f(v) := u'$ if $v \in X$ and $f(v) :=u''$ if $v \in Y$, where $u = f'(v)$. Then $f$ is a list homomorphism from $(G,L)$ to $H^*$.
\end{compactenum}
\end{proposition}

In particular we can extend \cref{lem:edge-gadget} and construct and $\rneq{(S)}$-gadget for every non-bi-arc graph that might appear as a leaf of the recursion tree of $H$.
\begin{corollary}
Let $H$ be a connected non-bi-arc graph, which is either
\begin{compactenum}[a)]
\item undecomposable, or
\item a strong split graph, such that the graph $H'$, obtained from $H$ by removing all edges with both endvertices in $P$ (including loops), is undecomposable.
\end{compactenum}
Let $S \subseteq V(H)$, such that $|S| \geq 2$ and for any $a,b \in S$ we have $N(a) \not\subseteq N(b)$ and $N(a) \not\subseteq N(b)$.
Then there exists a $\rneq(S)$-gadget, i.e., a graph $F$ with $H$-lists $L$ and two special vertices $x,y \in V(F)$, such that $L(x)=L(y)=S$ and
\begin{compactitem}
\item for any list homomorphism $h:(F,L)\to H$, it holds that $h(x) \neq  h(y)$,
\item for any distinct $a,b \in S$ there is a list homomorphism $h:(F,L)\to H$, such that $h(x)=a$ and $h(y)=b$.
\end{compactitem}
\end{corollary}
\begin{proof}
In case a), by \cref{lem:decompositions-equivalent}, $H^*$ is undecomposable.
Thus the claim follows directly from applying \cref{prop:general-gadgets} to the $\rneq(S')$-gadget for $H^*$ (where $S' := \{s' \colon s \in S\}$), constructed in \cref{lem:edge-gadget}.

In case b), let $(B,P)$ be the partition of $H$. Observe that for every $b \in B$ and $p \in P$ we have $N(b) \subseteq N(p)$, so $S$ is either contained in $B$, or in $P$. In other words, $S$ is contained in one bipartition class of $H'$. Let $(F,L)$ be the $\rneq(S)$-gadget for $H'$. Recall that $H'$ is an induced subgraph of $H^*$, so $L$ are $H^*$-lists. Now the claim follows from applying \cref{prop:general-gadgets} to $(F,L)$.
\end{proof}

Note that all other gadgets, e.g. distinguishers, could be generalized in a similar way.

\subsection{Further research directions}

In this paper we have shown tight complexity bounds for the list homomorphism problem, parameterized by the treewidth of the instance graph.

A very natural question, mentioned also in~\cite{DBLP:conf/stacs/EgriMR18}, is to provide analogous results for the non-list variant of the problem, denoted by \homo{H}.
Let us point out that despite of the obvious similarity of \homo{H} and \lhomo{H}, the methods used to obtain hardness proofs for these problems are very different.
For the list variant most hardness proofs (including all proofs in this paper) are purely combinatorial~\cite{FEDER1998236,DBLP:journals/combinatorica/FederHH99,DBLP:journals/jgt/FederHH03} and are based on considering some local structures in the target graph -- note that we can ignore the remaining vertices of $H$ by not including them in any list.
On the other hand, in the \homo{H} problem, we need to be able to capture the structure of the whole graph $H$ at once. This is the reason why hardness proofs for the non-list variant  usually require using some algebraic tools~\cite{DBLP:journals/jct/HellN90,DBLP:journals/tcs/Bulatov05,Siggers}.

Very recently, Okrasa and Rz\k{a}\.zewski were able to provide tight bounds for the complexity of \homo{H}, assuming two conjectures from algebraic graph theory from early 2000s~\cite{OkrasaSODA}.
It would be very interesting to strengthen these results, either by proving the mentioned two conjectures, or by providing a different reduction.

Another interesting direction is to consider one of many other variants of the graph homomorphism problem. Let us mention one, i.e., \emph{locally surjectve homomorphism}, denoted by \textsc{LSHom}($H$). In this problem we ask for a homomorphism from an instance graph $G$ to the target graph $H$, which is surjective on each neighborhood. In other words, if we map some vertex $x \in V(G)$ to some vertex $v \in V(H)$, then every neighbor of $v$ must appear as a color of some neighbor of $x$~\cite{DBLP:conf/mfcs/FialaPT05,DBLP:journals/ejc/FialaM06,DBLP:journals/tcs/FialaP05,DBLP:journals/jcss/OkrasaR20}.
We believe that it is interesting to show tight complexity bounds for this problem. One of the reasons why this problem is challenging is that the natural dynamic programming runs in time $2^{\Oh(|H| \cdot \tw{G})} \cdot (n+|H|)^{\Oh(1)}$.
Thus in order to show that this bound is tight, it is not sufficient to design edge gadgets encoding inequality and substitute all edges of the instance of $k \coloring$ with these edge gadgets, as we did in this paper.
Since the number of colors needs to be exponential in $H$, one should also design some \emph{vertex gadgets}, which will encode the exponential number of possible states.

Finally, instead of changing the problem, we can consider changing the parameter. We believe that an exciting question is to find tight bounds for \homo{H} and \lhomo{H}, parameterized by the \emph{cutwidth} of the instance graph, denoted by $\cw{G}$. Quite recently Jansen and Nederlof showed that the chromatic number of a graph can be found in time $2^{\Oh(\cw{G})} \cdot n^{\Oh(1)}$~\cite{DBLP:journals/tcs/JansenN19}, i.e., the base of the exponential factor does not depend on the number of colors. Jansen~\cite{BartPersonal} asked whether the same is possible for \homo{H} and \lhomo{H}, if the target $H$ is not complete?
Note that while the chromatic number of a graph can be found in time $2^n \cdot n^{\Oh(1)}$~\cite{DBLP:journals/siamcomp/BjorklundHK09}, for the \homo{H} and \lhomo{H} problems the $|H|^{\Oh(n)}$-time algorithm is essentially best possible, assuming the ETH~\cite{DBLP:journals/jacm/CyganFGKMPS17}. We believe that a similar phenomenon might occur if the cutwidth is a parameter.

\paragraph*{Acknowledgment.} We are grateful to Dani\'el Marx and L\'aszlo Egri for much advice regarding the topic and to Kamil Szpojankowski for an inspiring discussion about almost all graphs.

\newpage
\bibliographystyle{plain}
\bibliography{main}

\end{document}